\documentclass[conference]{IEEEtran}
\IEEEoverridecommandlockouts
\usepackage{xspace}
\usepackage{cite}
\usepackage{amsmath,amssymb,amsfonts}
\usepackage[misc]{ifsym}

\usepackage{algorithmic}
\usepackage[linesnumbered,ruled,vlined]{algorithm2e}
\usepackage[pdftex]{graphicx}
\usepackage{subfigure}
\usepackage{textcomp}
\usepackage{xcolor}
\newcommand{\jianxin}[1]{\textcolor{black}{#1}}
\newcommand{\jxrequire}[1]{\textcolor{black}{#1}}
\newcommand{\ignore}[1]{}
\DeclareMathOperator*{\argmin}{argmin}

\usepackage{amsthm}
\usepackage{url}

\newtheorem{definition}{Definition}
\newtheorem{theorem}{Theorem}
\newtheorem{lemma}{Lemma}
\newtheorem{property}{Property}

\def\BibTeX{{\rm B\kern-.05em{\sc i\kern-.025em b}\kern-.08em
    T\kern-.1667em\lower.7ex\hbox{E}\kern-.125emX}}

\begin{document}

\title{Top-$k$ Socio-Spatial Co-engaged Location Selection for Social Users
}

\author{
	Nur Al Hasan Haldar$^{\natural}$,  Jianxin Li$^{\#}$, Mohammed Eunus Ali$^{\dagger}$,Taotao Cai$^{\#}$, Timos Sellis$^{\S}$, Mark Reynolds $^{\natural}$ \vspace{1.6mm}\\ 
	\fontsize{10}{10}\selectfont\itshape
	$^{\natural}$The University of Western Australia, Australia  \ \ 
	$^{\#}$Deakin University, Australia \ \ 
	$^{\dagger}$BUET, Bangladesh\\
	$^{\S}$Swinburne University of Technology, Australia \\
	\fontsize{9}{9}\selectfont\ttfamily\upshape
	$^{\natural}$\{nur.haldar@research., mark.reynolds@\}uwa.edu.au; \ \ \
	$^{\#}$\{jianxin.li,taotao.cai\}@deakin.edu.au;\\
	$^{\dagger}$eunus@cse.buet.ac.bd; \ \ \ \
	\fontsize{9}{9}\selectfont\ttfamily\upshape
	$^{\S}$tsellis@swin.edu.au;
}

\maketitle

\begin{abstract}
	With the advent of location-based social networks,  users can tag their daily activities in different locations through check-ins. These check-in locations signify user preferences for various socio-spatial activities and can be used to improve the quality of services in some applications such as recommendation systems, advertising, and group formation. To support such applications, in this paper, we formulate a new problem of \textit{identifying top-k \textbf{S}ocio-\textbf{S}patial co-engaged \textbf{L}ocation \textbf{S}election} (\textit{SSLS}) for users in a social graph, that selects the best set of $k$ locations from a large number of location candidates relating to the user and her friends. The selected locations should be (i) \emph{spatially and socially relevant} to the user and her friends, and (ii) \emph{diversified in both spatially and socially} to maximize the coverage of friends in the spatial space. 
	To address such a challenging problem, we first develop an \texttt{Exact} solution by designing some pruning strategies based on the derived bounds on diversity. To make the solution scalable for large datasets, we also develop an approximate solution by deriving relaxed bounds and advanced termination rules to filter out insignificant intermediate results. To further accelerate the efficiency, we present one fast exact approach and a meta-heuristic approximate approach by avoiding the repeated computation of diversity at the running time. Finally, we have performed extensive experiments to evaluate the performance of our proposed algorithms against the adapted existing methods using four large real-world datasets.
		
\end{abstract}
\vspace{0.1cm}


\section{Introduction}
\label{sec:Intro}
Location-based Social Networks (LBSNs) that capture both the social and spatial information, are becoming popular. Conventional social network platforms, such as Facebook, have also enabled the location check-in features to allow social users to tag their daily activities at different places. Such location information along with social factors can be used to improve the quality of services in many applications such as recommendation systems, marketing, advertising, and group formation \cite{66konstas2009social,59ye2010location}. Given a user, the number of candidate locations might be quite large, and not all the locations are equally important to support the above applications as different locations may represent different aspects of the user's interest. Thus, selecting the \emph{key} locations of users from a large candidate set of locations by considering various socio-spatial factors is our main focus in this paper.


The social and spatial factors have a strong correlation in LBSNs where the check-in locations are established through social activities and spatial influences \cite{60ye2011exploiting}. 
Therefore, given a user and a large number of her visited locations, we should model both the social and spatial factors in a meaningful way to select a small set of places among the visited ones that can engage the user and her social connections.
More specifically, the social factors are significant in distinguishing the preferences of locations to a friend. Similarly, spatial factors can influence the user and her friends' interest in different spatial proximity.
Therefore, this work exploits both the social and spatial characteristics of relationships among the social network users and their locations to better support location-dependent applications.\ignore{We exploit the social implications embedded in locations, additionally the spatial properties w.r.t. friends are considered to select diversified subset of co-engaged locations for social users.} To that end, in this paper, we propose the problem of \textit{identifying top-k \textbf{S}ocio-\textbf{S}patial co-engaged \textbf{L}ocation \textbf{S}election for users}, denoted as \textit{SSLS}. A co-engaged location can be easily accessible by a user and her selected friends that can be covered by the location. Specifically, given a user, \textit{SSLS} will return a set of selected locations that satisfy the following two conditions: \\
i. (\textit{\textbf{relevance}}:) The selected locations should be both \emph{spatially and socially relevant} to the user and her social friends. \jianxin{Two users are called social friends to each other if they have a social connection, e.g., one user is following other.}\\
ii. (\textit{\textbf{diversity}}:) The selected locations should also be \emph{diversified both spatially and socially} in order to maximize the spatial and social coverage of the user's social friends.


\begin{figure}[htbp]
	\vspace{-4mm}
	\centering	
	{
		\includegraphics[scale=0.38]{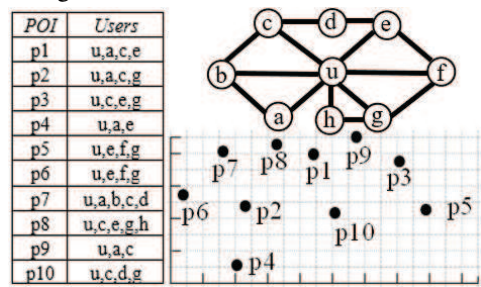}
	}
	\caption{\jianxin{An example of the \textit{SSLS} query}}
	\label{fig:exmplFirst}
	\vspace{-2mm}
\end{figure}  

\noindent\textbf{Applications.} \textit{SSLS} has a wide range of applications. Here, we discuss two applications to explain the \textit{SSLS} problem better:

\noindent\textbf{(1) Event Scheduling.} Let us consider a toy example of an event scheduling application in Figure \ref{fig:exmplFirst}. There are ten points-of-interests (POIs) $\{p_1, p_2, ..., p_{10}\}$ and a set of users who checked-in the places are given. In the example, an enterprise social network user $u$ wants to schedule a series of social events in multiple locations (say, in two locations), which will be preferable and convenient for both the user and her linked customers (e.g., friends). 
More specifically, the user $u$ wants to select the two locations such that they are (i) related: locations are the user's favorite ones where she visited earlier; (ii) socially and spatially relevant: locations where many of the user's friends also visited these places or some nearby places; (iii) spatial diversified: the selected locations are spatially distant, e.g., in different cities; and (iv) social diversified: each selected location should cover a set of friends such that the selected locations together can cover a maximum number of friends, and any two selected locations have a minimum overlap of friends to be covered.


In this application scenario, the \textit{SSLS} model will return $p_5$ and $p_7$ as the best two locations for scheduling events for $u$. This is because, (i) $p_5$ and $p_7$ are previously checked-in by $u$; (ii) all the friends (e.g., customers) of $u$ except $h$, have checked-in $p_5$ and $p_7$ (socially relevant); (iii) although the friend $h$ of $u$ did not have exact check-ins at either of the two locations, she has one check-in location $p_8$ near to $p_7$ (spatially relevant); (iv) \jianxin{$u$'s friends who checked-in $p_5, p_7$ are disjoint, i.e., $\{e, f, g\}$ with $p_5$, $\{a, b, c\}$ with $p_7$ (socially diverse);} and (v) $p_5$ and $p_7$ are itself spatially distant. 

\noindent\textbf{(2) Outlet Opening.} Nowadays, online business shops often maintain a Facebook page with many followers who like their products. Suppose an online business wants to open new outlets at $k$ number of locations that can attract most of its customers (followers) and their friends (potentially new customers). \jianxin{The business shop can predefine multiple regions (e.g., suburbs, cities) suitable for their future business. One can consider the candidate locations of the business shop as the check-ins of the customers within the predefined regions.} Thus, to select the $k$ locations from a large number of candidate locations (check-ins of the customers), the shop owner would like to consider the following: the selected locations are relevant to the current followers (spatial relevance) and their friends (social relevance). Also, the selected locations should be distant so that they can cover different areas (spatial diversity) and attract different groups of potential customers through these outlets (social diversity).
In this example, the check-ins of all users who liked the business page or their products (in a city) are considered as the candidate locations from where we need to select the top $k$ locations for opening outlets. Therefore, this example shows that without the loss of generality, our approach can be applied for location selections for a linked group of users.

We have proved that the \textit{SSLS} problem is NP-hard. To solve this problem, one may consider to directly use the existing greedy solutions on top-$k$ diversified spatial object selection, such as DisC \cite{17drosou2012disc} and SOS \cite{22guo2018efficient}. However, there exist some stringent gaps that make them inapplicable, including (i) Both DisC and SOS define diversity based on spatial distance only. Thus, they do not account for the important aspect of diversity in geo-social networks, which we refer to as \textit{social diversity}. We argue that both the spatial and social aspects need to be considered for selecting diversified objects in a geo-social network to get the best \textit{SSLS} set. (ii) Both the approaches depend on a user-defined distance threshold to select $k$ diversified objects. Nevertheless, it is hard for an end-user to define the best distance thresholds for different $k$ without knowing the underlying data distribution. 
Also, their selection processes cannot be personalized towards individual users with particular preferences. 
The main contributions of this work are as below: 

$\bullet$ \textbf{SSLS Formulation.} We formally define the top-$k$ Socio-Spatial co-engaged Location Selection problem. We provide detailed algorithms and metrics for using social and spatial relevance, and diversity to maximize the spatial and social coverage of the search space.	

$\bullet$ \textbf{Solutions.} 
First, we propose an \texttt{Exact} approach by developing some pruning strategies based on the \emph{derived lower bounds on diversity} of an already explored feasible set. 
We also devise an efficient exact method (\texttt{Exact$^+$}), a variation that derives bounds based on the \textit{relevance} of candidate locations, and avoids repetitive complex diversity computation of groups of locations like \texttt{Exact} method. Besides, we present an Approximate algorithm, in which we derive relaxed bounds and propose advanced termination criteria based on the score of the best feasible set and the diversity of remaining individual locations. Further, we introduce a greedy-based Fast Approximate approach that uses the bounds of \texttt{Exact$^+$}, and greedily selects the best locations. 

$\bullet$ \textbf{Extensive Evaluation.} We conduct extensive experiments to evaluate the effectiveness and efficiency of our proposed approaches using four real-world datasets. We also have compared the proposed algorithms with three adaptive greedy-based approaches, namely, SOS \cite{22guo2018efficient}, \textit{GMC} \cite{51vieira2011query}, and \textit{GNE} \cite{51vieira2011query}. 
\textbf{Organization.} We review the related work in Section \ref{sec:relatedWork}. Section \ref{sec:problem} formally defines the top-$k$ \textit{SSLS} problem. The \texttt{Exact} and \textit{Approximate} approach of top-$k$ \textit{SSLS} query are presented in Section \ref{sec:Exact} and Section \ref{sec:Approximate}, respectively. 
Section~\ref{sec:ExactPlus} presents the proposed \texttt{Exact$^+$}, and Fast Approximate solution. Finally, we report the experimental results in Section \ref{sec:experiment}, and conclude the paper in Section \ref{sec:Conclusion}. 

\section{Related Work}
\label{sec:relatedWork}

\jianxin{In this section, we first discuss the related work on LBSN queries in general, then present existing works about different forms of diversified object selection in spatial and metric space, and finally discuss the relevant works about spatial object selection.}

\jianxin{\textbf{Socio-Spatial Queries.} Various geo-social queries have been studied \cite{27armenatzoglou2013general,48sohail2020geo,44sohail2016retrieving,1ghosh2018flexible} that focus on retrieving useful information combining both the social relationships and the locations of the users. For example, the top-$k$ place query \cite{44sohail2016retrieving} fetches $k$ places of a user based on the distances from a query location and their popularity among the friends. A recent work on Geo-Social Temporal Top-$k$ query \cite{48sohail2020geo} ranks the retrieved locations according to their spatial, social relevance within a time interval. The computation of the relevance scores of these approaches are based on the given query location of a user, and do exploit socio-spatial features of a network (e.g. social diversity). However, the \textit{SSLS} query needs to select top-$k$ socially and spatially diverse locations which have higher socio-spatial scores w.r.t. the user and her neighbors' locations. Additionally, there exists some other works on socio-spatial queries such as location prediction \cite{34haldar2019location,28li2012multiple,72li2012towards} in social network. They investigate the user relationship and spatial information to infer location for a query user. Various personalized location recommendation queries \cite{71bao2012location,59ye2010location,70zheng2011recommending} consider location preferences with similar users. For example, Zheng et al.  \cite{70zheng2011recommending} recommend locations from friends' location histories such that the users can discover the locations that interest them. However, none of these works well exploit the characteristics of geographical social engagement.}

\textbf{Diversified Object Selection.} The diversity among the objects has been extensively studied to improve object selection problems (e.g. \cite{12catallo2013top,9drosou2014diverse,11fraternali2012top,10qin2012diversifying}), and it expands a wide variety of spectrum, e.g., diversified keyword search \cite{16angel2011efficient}, diversified query recommendation \cite{50zhu2011unified}. There are various definitions of selecting diversified objects which mainly depend on the content dissimilarity \cite{13ziegler2005improving}, information diversity \cite{14clarke2008novelty}, categorical diversity \cite{15agrawal2009diversifying}. \jianxin{There also exists several greedy solutions \cite{16angel2011efficient,62borodin2012max,63gollapudi2009axiomatic} that build the diversified result set in an incremental way.} Angel et al. \cite{16angel2011efficient} propose, \texttt{DivGen}, a content-based diversification algorithm which first computes the relevance of each document, and then updates the usefulness of all other documents based on the similarity to the highest scoring document. Another diversified query search framework was proposed by Qin et al. \cite{10qin2012diversifying}, where  datasets are transformed into Diversity Graph using node properties, and the selected diversified nodes have maximum total score with no two nodes are adjacent. \jianxin{The Maximum Marginal Relevance (MMR) function \cite{54carbonell1998use} maximizes relevance and diversity of a set w.r.t. a query element. Variations of MMR are considered in several domain specific greedy-based approaches \cite{55drosou2009diversity,56elhamifar2017online,51vieira2011query,57zhou2018neural}. These greedy-based approaches are monotone and generate the answer set by adding elements one by one in non-increasing order of their scores. The process stops when an approximate solution containing $k$ elements is identified. However, the results of our proposed \textit{SSLS} approach are not necessarily generated in non-increasing order and we provide both exact and approximate solutions for such problem by exploiting the relevance and diversity of the selected set.}

\textbf{Spatial Object Selection.} Works in this category are related to map services, spatial sampling, and POI selection problems. Existing map services retrieve a subset of spatial objects based on the relative weights of the retrieved objects that maximize the total weights \cite{5das2012efficient}. Nutanong et al. \cite{8nutanong2012multiresolution} define the problem of sampling large geo-spatial dataset in a region of user interest. Mahdian et al. \cite{7mahdian2015algorithmic} propose POI selection problem, that targets to identify a set of POIs with maximum utility according to some reference POIs. Meanwhile, DisC \cite{17drosou2012disc} essentially selects the subset of diversified objects, where two selected objects must be at least $r$ distance from each other, and there should be at least one object (un-selected) in the dataset within $r$ distance from a selected object. On the other hand, the Spatial Object Selection (SOS) \cite{22guo2018efficient} model selects $k$ diversified objects in such a way that any two selected objects must be at threshold distance from each other and the aggregate similarity (computed based on semantic attributes) from the selected set of objects to the whole dataset is maximized. However, these works do not consider any social factors e.g., social relevance and social diversity.

\section{Problem Formulation}
\label{sec:problem}

Let $G(V, E^{\prime}, L, E^{\prime\prime})$ be a socio-spatial graph, where $V$ is the set of users, $E^{\prime}$ is the set $\{(u, v)| u,v \in V\}$ of edges representing the social connections among users, $L$ is the set of locations of the users, $E^{\prime\prime}$ is the set $\{(u,l) | u\in V, l\in L\}$ of edges representing the spatial connections between users and locations. 
Let, $V_u \in V$  and $L_u \in L$ be the social friends and check-ins of user $u$, respectively; the goal of \textit{SSLS} query is to find the best $k$ \emph{socio-spatial relevant and diversified} locations from $L_u$ for user $u$. 
In this section, we first discuss the intuitions and metrics of socio-spatial relevance and diversity, and then formalize the top-$k$ \textit{SSLS} problem with the proof of NP-hardness. Table \ref{tab:notations} lists the notations used in this paper.
\subsection{Socio-spatial Relevance}
The study in \cite{65wang2014checkins} showed that social interest is the type of check-in incentive that stimulates interactions or influences among friends. Therefore, a location may have higher social importance to a user if a large number of her friends have checked-in the location. 
\jianxin{Based on this intuition, we define the social relevance score of a location $l_i\in L_u$ of user $u$ as, $S_{sc}(l_i, u) = \frac{|v \in V_u :(v, l_i)\in E^{\prime\prime}|}{|V_u|}$, where $|V_u|$ is the number of friends of $u$. Similarly, to define the spatial relevance score, the study \cite{60ye2011exploiting} has revealed that geographical proximities of POIs significantly influence social users’ check-in behavior. Ye et al. \cite{59ye2010location} also remarked that friends with nearby check-ins would have a higher probability of sharing common locations, as it is easier for the friends to participate in some activities at their mutually known places. Based on these intuitions of \cite{59ye2010location,60ye2011exploiting}, and accordance with the spatial score justified in \cite{58armenatzoglou2015geo}, we define spatial relevance score of location $l_i$ of $u$ as, $S_{sp}(l_i, u) = 1 - \frac{\sum_{v\in V_u}\min dist(l_i, L_v)} {d_m * |V_u|}$. Here, $\min dist(l_i, L_v)$ returns the smallest distance between $l_i$ and the locations $L_v$ of friend $v\in V_u$. The denominator $d_m* |V_u|$ is used to adjust the spatial relevance score within the range $(0,1]$, where $d_m = \max\{\min_{v\in V_u}dist(l_i, L_v)\}$ is calculated as the maximum value among the smallest distances between the location $l_i\in L_u$ and the location set $L_v$ of each friend $v\in V_u$. Finally, the socio-spatial relevance score $R_{ss}(l_i, u)$ of location $l_i\in L_u$ is defined as the weighted sum of $S_{sc}$ and $S_{sp}$, i.e., $R_{ss}(l_i, u) = \alpha\cdot S_{sc}(l_i, u)+(1-\alpha)\cdot S_{sp}(l_i,u)$, where $\alpha\in [0,1]$ specifies the relative importance of social and spatial costs. For simplicity, we will refer $R_{ss}(l_i, u)$ as $R_{ss}(l_i)$, as $l_i\in L_u$ itself represents a location of $u$. 
	As such, a set $S$ of locations of user $u$ can have its socio-spatial relevance score as $R_{ss}(S) = \sum_{l\in S}R_{ss}(l) = \alpha\cdot \sum_{l\in S}S_{sc}(l) + (1-\alpha)\cdot \sum_{l\in S} S_{sp}(l)$.}

\begin{table}
	\centering
	\caption{Basic Notations}
	\scalebox{0.95}{
		\begin{tabular}{p{1.2cm}|p{6.2cm}}
			\hline 
			\textbf{Symbols} & \textbf{Descriptions}\\ 
			\hline 
			$R_{ss}\{D_{ss}\}$	& Socio-spatial Relevance \{Diversity\} Score \\ 
			$L_u\{V_u\}$ & Set of check-in locations \{social friends\} of $u$\\
			$\alpha$	& Trade-off between spatial and social importance\\ 
			$\omega$ & Trade-off between relevance and diversity\\	
			\hline  
		\end{tabular} 
	}
	\label{tab:notations}
\end{table}
\subsection{Socio-spatial Diversity}
\label{sec:SocSpatDiversity}
\jianxin{Intuitively, the diversity requires to measure the dissimilarity (or distance) among objects in a set. A spatially diversified location pair should reside far from each other \cite{17drosou2012disc,3shi2014density}. Similar to the spatial distance function \cite{3shi2014density}, we calculate  spatial diversity $D_{sp}(l_i, l_j) = \frac{dist(l_i, l_j)}{maxD}$ between $l_i, l_j\in L_u$ as the normalized Euclidean distance $dist(l_i,l_j)$. The constraint $maxD$ can be assigned as the maximum distance among the location pairs in $L_u$. Similarly, the social diversity between a location pair of user $u$ depends on her friends who visited the locations~\cite{68su2019experimental}. 
	As defined in \cite{3shi2014density}, we calculate social diversity score $D_{sc}(l_i, l_j) = 1 - \frac{|V_{u, l_i} \cap V_{u, l_j}|}{|V_{u, l_i} \cup V_{u, l_j}|}$ of locations $l_i, l_j\in L_u$ using Jaccard distance, where $V_{u, l_i}$ and $V_{u, l_j}$ are the set of $u$'s friends who checked-ins at $l_i$ and $l_j$, respectively. Similar as \cite{3shi2014density}, we also define socio-spatial diversity of a location pair as the weighted sum of $D_{sc}(l_i, l_j)$ and $D_{sp}(l_i, l_j)$, i.e., $D_{ss}(l_i, l_j) = \alpha\cdot D_{sc}(l_i, l_j)+(1-\alpha)\cdot D_{sp}(l_i, l_j)$, where $\alpha\in [0,1]$ specifies the relative importance of social and spatial costs.
	As such, given a location set $S$, the socio-spatial diversity score of each location $l\in S$ w.r.t. $S$ is calculated as $D_{ss}(l, S) = \min\{D_{ss}(l, l_i)|l_i\in S\setminus l\} = D_{ss}(l, \bar{l})$ s.t., $\bar{l} = \underset{l_i\in S\setminus l}{\arg\min} D_{ss}(l, l_i)$\ignore{$\underset{l_i\in S\setminus l}{\argmin} {D_{ss}(l, l_i)}$}. We further calculate the socio-spatial diversity of a set $S$ as $D_{ss}(S) = \sum_{l\in S}D_{ss}(l, S\setminus l) = \sum_{l\in S}D_{ss}(l, \bar{l}) = D_{ss}(l_i, l_j) = \alpha\cdot \sum_{l\in S} D_{sc}(l, \bar{l}) + (1-\alpha)\cdot \sum_{l\in S} D_{sp}(l, \bar{l})$, where, $\bar{l} = \underset{l_i\in S\setminus l}{\arg\min} D_{ss}(l, l_i)$ is the location among the set $S$ that has minimum diversity with $l$.}

\subsection{Socio-spatial Score of Location Set}
\jianxin{
	We follow the existing works \cite{58armenatzoglou2015geo,3shi2014density} to derive a ranking function as the weighted linear combination of socio-spatial relevance and diversity. Given a location set $S\subseteq L_u$ of user $u$, Equation~\ref{equ:FS} describes the socio-spatial score function $F(S)$, where $\omega\in (0,1)$ specifies the relative importance of relevance and diversity; when $\omega > 0.5$, the relevance of the selected locations to the query is more important than their diversity. 
	\vspace{-2mm}
	\begin{equation}
	\small
	\label{equ:FS}
	\begin{split}
	& F(S)  = \omega\cdot R_{ss}(S) + (1-\omega)\cdot D_{ss}(S) 
	= \omega\cdot \big (\alpha \cdot\sum_{l\in S} S_{sc}(l) \\ 
	& + (1-\alpha)\cdot \sum_{l\in S} S_{sp}(l) \big) + (1-\omega)\cdot \big ( \alpha\cdot \sum_{l\in S} D_{sc}(l, \bar{l}) + \\ & (1-\alpha)\cdot \sum_{l\in S} D_{sp}(l, \bar{l})\big) \text{, \hspace{1mm} (s.t., } \bar{l} = \underset{l_i\in S\setminus l}{\arg\min} D_{ss}(l, l_i))
	\end{split}
	\end{equation}
	\normalsize
}
\vspace{-3mm}

\textbf{Problem Statement of Top-$k$ \textit{SSLS} Query.} \jianxin{Given a social graph $G$, a positive integer $k$, a query user $u$ with check-in locations $L_u$, trade-off parameters $\omega$ between relevance and diversity, and $\alpha$ signifying relative importance between social and spatial factors, and socio-spatial score function $F$, the top-$k$ \textit{SSLS} query returns a set $S$ of $k$ locations from $L_u$, s.t., $\forall S^* \subseteq L_u, F(S)> F(S^*)$, where $|S^*|=k$ and $S \neq S^*$. 
}

\textbf{Significance of $\alpha$, $\omega$ in \textit{SSLS} Query.}
\jianxin{
	The trade-off parameters have significant importance in the quality of the selection considering both the social and spatial aspects. For example, if an application prefers social factors including social relevance and social diversity, then we can set $\alpha =  1$ and $\omega$ as default (e.g., $\omega=0.5$). 
	Thus, it means that the selected locations should be checked-in by a diverse set of friends, and the locations are socially relevant to the user and her friends. 
	Given such setting, $S= \{p_6,p_7\}$ will be selected as the answer for the example in Figure~\ref{fig:exmplFirst}. If we increase $\omega$ to 0.6, i.e., the socio-spatial relevance is preferred, then the \textit{SSLS} query will return $S = \{p_7, p_8\}$ as the answer. Similarly, an end-user can tailor the result by varying different values for $\alpha$ and $\omega$.
}
\begin{theorem}
	The top-$k$ \textit{SSLS} problem is NP-hard.
\end{theorem}
\begin{proof}
	We consider a special case of the problem, assuming the socio-spatial relevance score of each location of user $u$ is 1, i.e., $\forall l\in L_u$, $R_{ss}(l) = 1$, and each location pair is connected with an edge where the edge-weight is represented by socio-spatial distances. We remove the edges between the location pairs where social diversity \jianxin{is} $0$, and present the location set $L_u$ as vertices of a graph $G$. Based on this setting, our top-$k$ \textit{SSLS} problem can be transformed into the problem of top-$k$ diverse vertices search in a graph. Additionally, we know that finding top-$k$ diverse set of vertices from $G$ is equivalent to finding maximum weight independent set (MWIS) of size $k$ \cite{10qin2012diversifying}. Further, in \cite{38garey2002computers}, the problem of MWIS has been proved as NP-hard. Hence, we can conclude the proof. 
\end{proof}

\section{An Exact Approach}
\label{sec:Exact}
For an exact solution to answer the Top-$k$ \textit{SSLS} query, we resort to an incremental Branch-and-Bound (BnB) strategy that progressively adds locations to \jianxin{build} the answer set. \jianxin{The key idea is to develop pruning strategies based on the derived lower bound on socio-spatial diversity of an intermediate set  that can avoid the exploring a large number of location sets. 
}
\vspace{-2mm}
\subsection{Computing Bounds \jianxin{on Diversity of Intermediate Set}}

We use the concept of score gain to decide whether a location should be added to an intermediate result set $S_I$ in the process of finding a top-$k$ \textit{SSLS} set. Initially, $S_I$ is empty, and $|S_I|<k$ holds always. We use $S_R = L_u \setminus S_I$ to denote the set of remaining locations of user $u$ w.r.t. $S_I$. 
If we add a location $l'\in S_R$ to $S_I$, the socio-spatial score $F(S'_I)$ of set $S'_I = \{S_I\cup l'\}$ becomes, $F(S'_I) = \omega\cdot R_{ss}(S'_I) + (1 - \omega)\cdot D_{ss}(S'_I)$, and consequently the socio-spatial score gain $\delta$ of $S'_I$ w.r.t. the previous set $S_I$ can be computed as, $\delta = F(S'_I) -  F(S_I)$,
\vspace{-2mm}
\begin{equation}
\small
\label{equ:gain}
\begin{split}
\Rightarrow & \hspace{0.2cm} \delta = \omega\cdot R_{ss}(S'_I) + (1 - \omega)\cdot D_{ss}(S'_I) - \omega\cdot R_{ss}(S_I) \\
& \hspace{-0.2cm} - (1 - \omega)\cdot D_{ss}(S_I) = \omega\cdot \big (R_{ss}(S'_I) - R_{ss}(S_I)\big ) \\
& \hspace{-0.2cm} + (1-\omega)\cdot \big ( D_{ss}(S'_I) - D_{ss}(S_I)\big ) = \omega\cdot \delta_r + (1-\omega)\cdot \delta_d	
\end{split}	
\end{equation}
\normalsize
Here, we consider $\delta_{r}$ and $\delta_{d}$ as the Relevance Gain and Diversity Gain of $S'_I$ w.r.t. the previous set $S_I$, respectively:

\noindent\textbf{Relevance Gain ($\delta_{r})$.} The relevance gain can be simplified as: $\delta_{r} = R_{ss}(S'_I) - R_{ss}(S_I) = R_{ss}(S_I\cup l') - R_{ss}(S_I) = R_{ss}(l')$. $\delta_{r} \in [0,1]$ can not be negative for any $l'\in L_u$. 

\noindent\textbf{Diversity Gain ($\delta_{d})$.}  
The diversity gain $\delta_{d} = D_{ss}(S'_I) - D_{ss}(S_I)$ is the difference in socio-spatial diversity of $S'_I  = S_I\cup l'$ to $S_I$. $\delta_{d}$ can be negative when $D_{ss}(S'_I) < D_{ss}(S_I)$.

Here, the value of $D_{ss}(S'_I)$ is dependent on the diversity of the added location $l'\in S_R$ w.r.t. $S_I$, and the updated aggregated diversity of the locations of $S_I$, such as,
\vspace{-2mm}
\begin{equation}
\small
\label{equ:divNewSetS}
\begin{split}
& D_{ss}(S'_I) = D_{ss}(l', S_I) + \sum_{l\in S_I}\min\{D_{ss}(l, S_I\setminus l), D_{ss}(l, l')\}\\
&  = \widehat{d} + \widehat{D},
\end{split}
\end{equation}
\normalsize
\jianxin{where the first part, $\widehat{d} = D_{ss}(l', S_I) = \min_{l\in S_I}\{D_{ss}(l', l)\}$ is the diversity of the newly added location $l'$ w.r.t. the intermediate set $S_I$, and the remaining part $\widehat{D} = \widehat{D}(S_I, l') = \sum_{l\in S_I}\min\{D_{ss}(l, S_I\setminus l), D_{ss}(l, l')\}$ 
	is the updated total diversity of the existing set $S_I$, when $l'\in S_R$ is added to $S_I$. 
	Using Equation~\ref{equ:divNewSetS}, we derive the diversity gain $\delta_{d}$ as, 
}
\vspace{-2mm}
\begin{equation}
\small
\label{equ:diver}
\begin{split}
\delta_{d} = D_{ss}(S'_I) - D_{ss}(S_I) = \widehat{d} + \widehat{D} - D_{ss}(S_I)
\end{split}
\end{equation}
\normalsize
Further, we obtain the socio-spatial score gain of the intermediate set $S_I$ using Equation~\ref{equ:gain} and Equation~\ref{equ:diver} as follows,
\vspace{-2mm}
\begin{equation}
\small
\label{equ:gainUpdated}
\delta = \omega\cdot \delta_r + (1 - \omega)\cdot (\widehat{d} + \widehat{D} - D_{ss}(S_I))
\end{equation}
\normalsize

Now, we will identify the \textit{eligible} locations from $S_R$ that can generate a positive socio-spatial score gain w.r.t. $S_I$. 
\begin{definition}[Eligible Location]
	\label{def:eligibleLoc}
	Given a current intermediate set $S_I$, and a location  $l'\in S_R$, $l'$ will be considered as an eligible location if 
	$\delta = \omega\cdot \delta_r + (1 - \omega)\cdot (\widehat{d} + \widehat{D} - D_{ss}(S_I)) > 0$.
\end{definition}

Next, we will define some lemmas using the socio-spatial diversity of a set of locations to deduce a lower bound on $\widehat{D}$.

\begin{lemma}\label{lemma:del_d}
	Given an intermediate set $S_I$, an eligible
	location $l'\in S_R$ w.r.t. $S_I$, the updated 
	aggregated socio-spatial diversity $\widehat{D}$
	of set $S_I$ w.r.t. the eligible location $l'$ will never exceed the total socio-spatial diversity $D_{ss}(S_I)$ of the intermediate set $S_I$, e.g., $\widehat{D} \le D_{ss}(S_I)$ is always true.
	%
	
\end{lemma}

\begin{proof} 
	We know, $\sum_{l\in S_I}\min\{D_{ss}(l, S_I\setminus l), D_{ss}(l, l')\}$ $\le$ $\sum_{l\in S_I} D_{ss}(l, S_I\setminus l)$ is true, as for any $l\in S_I$, $\min\{D_{ss}(l, S_I\setminus l), D_{ss}(l, l')\}$ is no larger than $D_{ss}(l, S_I\setminus l)$. Therefore, $\widehat{D} \le D_{ss}(S_I)$ holds, as $\widehat{D} = \sum_{l\in S_I}\min\{D_{ss}(l, S_I\setminus l), D_{ss}(l, l')\}$, and $D_{ss}(S_I) = \sum_{l\in S_I} D_{ss}(l, S_I\setminus l)$ (refer Section \ref{sec:SocSpatDiversity}).
\end{proof}

An \textit{eligible} location $l'$ may produce a negative gain in diversity $\delta_{d}$ that can lessen the socio-spatial score of an updated set $S'_I = S_I \cup l'$ comparing with $S_I$. The below lemma derives the condition when instead of having a negative diversity gain, the socio-spatial score of an intermediate set can generate a positive socio-spatial gain (e.g., $\delta>0$) for $S'_I$.
\begin{lemma}
	\label{lemma:postvGain}
	Given an intermediate set $S_I$, an eligible location $l'\in S_R$, s.t., $S'_I = S_I\cup l'$; if $S'_I$ has a negative gain in diversity but $\delta_r > \frac{(1-\omega)}{\omega}\cdot |\delta_{d}|$ w.r.t. $S_I$, then the socio-spatial score of $S'_I$ will be larger than that of $S_I$, e.g., $F(S'_I) > F(S_I)$.
\end{lemma}
\begin{proof}
	%
	Let, $F(S'_I)$ and $F(S_I)$ be the socio-spatial scores of $S'_I = S_I\cup l'$ and $S_I$ respectively. Hence, the socio-spatial gain of $S'_I$ is $\delta = F(S'_I) -  F(S_I)$. 
	If $\delta_{d} < 0$, but $\delta_r>\frac{(1-\omega)}{\omega}\cdot |\delta_{d}|$, then $\delta = \omega\cdot \delta_r + (1-\omega)\cdot \delta_{d}> 0$ is always true. Therefore, $\delta = F(S'_I) -  F(S_I) > 0$. Hence, $F(S'_I) > F(S_I)$.
\end{proof}
\vspace{-2mm}
\jianxin{Now, we will derive a lower bound on the updated diversity ($\widehat{D}$) of an intermediate set $S_I$, which will help us to discard a large number of locations from $S_R$ that can not generate a better solution w.r.t. the current intermediate set $S_I$.}

\subsubsection{\underline{Lower Bound of $\widehat{D}$}}
Maximum relevance gain of an intermediate set $S_I$ w.r.t. $S_R$ is $\delta_{r\_max} = \max_{l'\in S_R} R_{ss}(l')$. 
Similarly, the maximum possible diversity of locations in $S_R$ w.r.t. $S_I$ can be calculated as $\widehat{d}_{max} = \max_{l'\in S_R}D_{ss}(l', S_I)$.

An \textit{eligible} location $l'\in S_R$ always derives positive gain, e.g., $\delta > 0$, to an intermediate set $S_I$ (Definition \ref{def:eligibleLoc}).
Thus, we get, $\delta = \omega\cdot \delta_r + (1-\omega)\cdot \big(\widehat{d} + \widehat{D} - D_{ss}(S_I)\big) > 0$. Therefore, $\widehat{D} > D_{ss}(S_I) - \widehat{d} - \frac{\omega}{1- \omega}\cdot \delta_r$.
Now, we derive lower bound of $\widehat{D}$ by replacing $\widehat{d}$ and $\delta_r$ with their maximum possible values, 
\vspace{-3mm}
\begin{equation}
\small
\label{equ:del''LowerBoundCalc}
\widehat{D}{\downarrow} = D_{ss}(S_I) - \widehat{d}_{max} -\frac{\omega}{1 - \omega}\cdot \delta_{r\_max}
\end{equation}
\normalsize

\vspace{-2mm}
\subsubsection{\underline{Early Pruning based on $\widehat{D}$}}
\label{sec:AdvPruningStrategy}
Using Equation \ref{equ:del''LowerBoundCalc}, we derive that a location $l' \in S_R$ cannot be included into an intermediate set $S_I$, if $\widehat{D} \le \widehat{D}{\downarrow}$ is true. 
We formalize this pruning condition in Property \ref{lemma:pruningCondLemmaNew} assuming that we are yet to find a feasible solution of size $k$. 
\begin{property}
	\label{lemma:pruningCondLemmaNew}	
	Given an intermediate set $S_I$, s.t., $|S_I|<k$, we can prune a  location $l'\in S_R$ w.r.t. $S_I$, if $\widehat{D} \le \widehat{D}{\downarrow}$ satisfies.	
\end{property}
\noindent Further, we derive an advanced termination strategy based on the score of already explored best feasible set and the expected contributions of the remaining locations in the overall score.

\subsubsection{\underline{Advanced Pruning}}
\label{sec:AdvTermStrategy}
First, we derive a pruning condition for an intermediate set $S_I$ of size $(k-1)$, then generalize to any sets of size less than $k$. 
Let, $S_b$ be the previously identified best feasible set. Also, let  $l' \in S_R$ be an arbitrary location with relevance score $R_{ss}(l')$, and $D_{ss}(l', S_I)$ be the diversity of $l'$ w.r.t.  $S_I$. We denote the updated diversity of $S'_I = S_I\cup l'$ as, $\sum_{l\in S'_I\setminus l'}\min\{D_{ss}(l, S_I\setminus l), D_{ss}(l, l')\} = \widehat{D}$. The set $S'_I$ of size $k$ can replace an earlier identified best feasible set $S_b$, if $F(S'_I) = \omega\cdot R_{ss}(S_I \cup l') + (1-\omega)\cdot D_{ss}(S_I\cup l') > F(S_b)$, 
\begin{equation}
\small
\label{equ:earlyTermExact}
\begin{split}
& \hspace{-0.4cm} \Rightarrow \omega\cdot (R_{ss}(S_I) + \delta_r) + (1-\omega)\cdot \big ( \widehat{d} + \widehat{D} \big) > F(S_b) \\
& \hspace{-0.4cm} \Rightarrow \widehat{D} > \frac{1}{1 - \omega}\cdot \big (F(S_b) - \omega\cdot (R_{ss}(S_I) + \delta_r)\big) - \widehat{d} \\
\end{split}
\vspace{-2mm}
\end{equation}
\normalsize
The lower bound of $\widehat{D}$ for termination (when $|S_I| = (k-1)$) can be obtained by replacing $\widehat{d}$ and $\delta_r$ with their corresponding maximum possible values, e.g., $\widehat{d}_{max} = \max_{l'\in S_R} D_{ss}(l', S_I)$ and $\delta_{r\_max} = \max_{l' \in S_R}R_{ss}(l')$  respectively. Therefore, $\widehat{D}{\Downarrow} = \frac{1}{1 - \omega}\cdot \big (F(S_b) - {\omega}\cdot (R_{ss}(S_I) + \delta_{r\_max})\big) - \widehat{d}_{max}$.

Adopting the above procedure, we add an arbitrary subset of locations  $S'_R\subseteq S_R$ to the intermediate set $S_I$, such that (i) $|S'_R| = (k-|S_I|)$, (ii) the socio-spatial score of new set $S' = S_I\cup S'_R$ surpasses $F(S_b)$, e.g., $F(S') > F(S_b)$. Therefore,
\begin{equation}
\label{equ:newAdd}
\small
\begin{split}	
\hspace{-0.1cm} \omega\cdot (R_{ss}(S_I) + R_{ss}(S'_R)) + (1-\omega)\cdot D_{ss}(S_I \cup S'_R) > F(S_b)\\
\end{split}
\end{equation}
\normalsize

Now, we define the below lemma on socio-spatial diversity of a set $S' = S_I\cup S'_R$ of size $k$, using the diversity scores of the locations $l'\in S'_R$ w.r.t. current intermediate set $S_I$.
\begin{lemma}
	\label{lemma:lemmaNew}
	Given an intermediate set $S_I$, a subset $S'_R\subseteq S_R$ of locations, the socio-spatial diversity $S' = S_I\cup S'_R$ satisfies $D_{SS}(S_I\cup S'_R) \le \widehat{D} + \sum_{l'\in S'_R}D_{ss}(l', S_I)$, where $\widehat{D}$ is the updated diversity of $S_I$ w.r.t. an arbitrary location $l'\in S'_R$.
\end{lemma}
\begin{proof}
	Proof is omitted due to space limitations.
\end{proof}

\noindent Applying Lemma \ref{lemma:lemmaNew} in (\ref{equ:newAdd}), we get,  $\omega\cdot (R_{ss}(S_I) + R_{ss}(S'_R))$
\begin{displaymath}
\small
\begin{split}	
\noindent & + R_{ss}(S'_R)) + (1-\omega)\cdot (\widehat{D} + \sum_{l'\in S'_R}D_{ss}(l', S_I)) > F(S_b)\\
& \Rightarrow \widehat{D} > \frac{F(S_b) - \omega\cdot \big(R_{ss}(S_I) + R_{ss}(S'_R)\big)}{(1-\omega)} - \sum_{l'\in S'_R}D_{ss}(l', S_I)
\end{split}
\end{displaymath}
\normalsize
Now, we will derive the lower bound $\widehat{D}{\Downarrow}$ by replacing $R_{ss}(S'_R)$ and $\sum_{l'\in S'_R}D_{ss}(l', S_I)$ with their maximum values, 
\vspace{-1mm}
\begin{equation}
\small
\label{equ:lowerBoundTerminationOld}
\begin{split}
& \widehat{D}{\Downarrow} = \frac{F(S_b) - \omega\cdot \big(R_{ss}(S_I) + 	R^{Max}_{ss}(S'_R)\big)}{(1-\omega)} - D^{Max}_{ss}
\end{split}
\end{equation}
\normalsize
Here, $R^{Max}_{ss}(S'_R) =  \max_{l'\in S_R}(\sum_{k-|S_I|}R_{ss}(l'))$ is the aggregated top $(k-|S_I|)$ relevance scores among the locations in $S_R$, and $D^{Max}_{ss} = \max_{l'\in S_R}(\sum_{k-|S_I|}D_{ss}(l', S_I))$ is the sum of the top $(k-|S_I|)$ diversity scores of the locations $l'\in S_R$ w.r.t. $S_I$.
Finally, we formalize the pruning condition in Property~\ref{lemma:advPruning} when a feasible set has been retrieved already.
\begin{property}[Location Pruning]
	\label{lemma:advPruning}
	Let $S_I$ be an intermediate set s.t. $|S_I|< k$, $|S_I|+|S_R|\ge k$, and $S_b$ be the best feasible set of size $k$ that has been identified already. 
	Using Equation~\ref{equ:lowerBoundTerminationOld}, we can prune location $l'\in S_R$ w.r.t. $S_I$, if $\widehat{D} \le \widehat{D}{\Downarrow}$ satisfies. 
\end{property}

\jianxin{The \texttt{Exact} algorithm progressively adds locations, and checks whether the locations can generate a positive gain in the socio-spatial score. Further, it prunes a large number of locations using the lower bound of an intermediate set.}

\subsection{Algorithm}
\label{Algo:Exact}
Algorithm \ref{Algo:top-k-ssls-basic-new} summarizes the \texttt{Exact} approach for answering the \textit{SSLS} query. It takes socio-spatial graph $G$, query user $u$, an integer $k$ as inputs, and returns a set $S$ of $k$ locations that maximizes socio-spatial score $F(S)$. We initialize an intermediate set $S_I$ as empty, and $S_R$ contains the remaining locations $l\in L_u\setminus S_I$ arranged in descending order of socio-spatial relevance scores $R_{ss}$.
\jianxin{
	A priority queue, $Q$, maintains a tuple of intermediate set $S_I$, set $S_R$ of remaining locations, and socio-spatial score of $S_I$. An inner loop fetches next location $l$ from $S_R$ (Line \ref{Algo1:fetchNext}), and an entity $(S_I - \{l\}, S_R)$ is pushed to $Q$. If no feasible set is retrieved yet, the process further prunes $S_R$ using Property \ref{lemma:pruningCondLemmaNew} (Line \ref{Algo1:pruneE}). Otherwise, Property \ref{lemma:advPruning} (Line \ref{Algo1:pruneT}) is used to prune. Finally, an entity $(S_I, S_R)$ is pushed into $Q$ when $|S_R|>0$. The process continues until $Q$ is empty. Finally, the final result set $S$ of size $k$ is returned.
}
\begin{algorithm}[th]
	\small
	\label{Algo:top-k-ssls-basic-new}
	\caption{SSLS: \texttt{Exact}} 
	\KwIn{Socio-spatial graph \textit{$G$}, set size $k$, query user $u$}	
	\KwOut{Location set $S$ of size $k$} 
	Initialize: $S_I \gets \emptyset, S \gets \emptyset$, $F(S_b)\gets 0$, $flagFS \gets false$,
	
	Append $\langle l, R_{ss}(l, u) \rangle$ into $S_R$ in non-increasing $R_{ss}(l, u)$
	
	
	$Q.push(S_I, S_R, 0)$ 
	
	\While {$Q$ is not empty}{
		$S_I, S_R \gets Q.pop()$\\	
		\If{$|S_I|=k$ or $|S_R| = \emptyset$}{
			continue
		}				
		\While{$|S_I|<k$ and $|S_I|+|S_R|\ge k$}{\label{Algo1:innerLoop}
			$l\gets nextLocation(S_R)$ \label{Algo1:fetchNext}\\	
			$S_I.append(l); S_R.remove(l)$\label{Algo1:appendSI}\\
			$Q.push(S_I - \{l\}, S_R, F(S_I-\{l\}))$\label{Algo1:QPush1}\\
			\If{$flagFS == false$}{ \label{Algo1:lineIf}
				$S_R \gets pruneE(S_I, S_R)$ \text{*** Property \ref{lemma:pruningCondLemmaNew}} \label{Algo1:pruneE}
			}\Else{ \label{Algo:else1}
				$S_R \gets pruneT(S_I, S_R, F(S_b))$ \text{*** Property \ref{lemma:advPruning}} \label{Algo1:pruneT}
			}	\label{Algo:else2}
			\If{$|S_R|>0$}{
				$Q.push(S_I, S_R, F(S_I))$ \label{Algo1:QPush2}\\
			}			
			\If{$|S_I|==k$ and $F(S_I)>F(S_b)$}{ \label{Algo1:if}
				$S\gets S_I$;
				$F(S_b)\gets F(S_I)$\\
				$flagFS \gets true$; 
				break;
			}	\label{Algo1:end}
		}
	}
	\vspace{-1mm}	
\end{algorithm}
\normalsize

\textbf{Time Complexity.} \jianxin{Time complexity of \texttt{Exact} is $O(^nC_k)$, as in worst case the \texttt{Exact} needs to check all combinations of $k$ from $n$ number of locations. 
	However, in practice, the actual running time is much less as large number of locations can be pruned using the developed pruning strategies.} 


\textbf{Steps of \texttt{Exact} with an example.} We use the example in Figure~\ref{fig:exmplFirst} to demonstrate the steps of \texttt{Exact} algorithm. First, we will show the steps to compute the socio-spatial relevance score of a location (say, $p_6$), and socio-spatial diversity of a location pair (say, $\{p_6,p_2\}$) of $u$ using the check-in information available in Figure~\ref{fig:exmplFirst}. 

The user $u$ has seven friends, among them three friends checked-in the location $p_6$ that results social relevance score $S_{sc}(p_6,u) = \frac{3}{7} = 0.43$. Now, we will calculate the spatial diversity score $D_{sp}(p_6, p_2)$. The locations $l_6$ and $l_2$ are checked-in by $u$'s friends $\{e,f,g\}$ and $\{a,c,g\}$ respectively. Therefore, the spatial diversity score $S_{sp}(p_6, p_2)$ is calculated as, $D_{sp}(p_6, p_2) =  1 - \frac{|\{e,f,g\}\cap \{a,c,g\}|}{|\{e,f,g\}\cup \{a,c,g\}|} = 1 - \frac{1}{5} = 0.80$. The calculated social diversity and social relevance of $u$'s locations are shown in Figure~\ref{fig:relScoreExample}.
\begin{figure}[htbp]
	\vspace{-3mm}
	\centering	
	\includegraphics[scale=0.38]{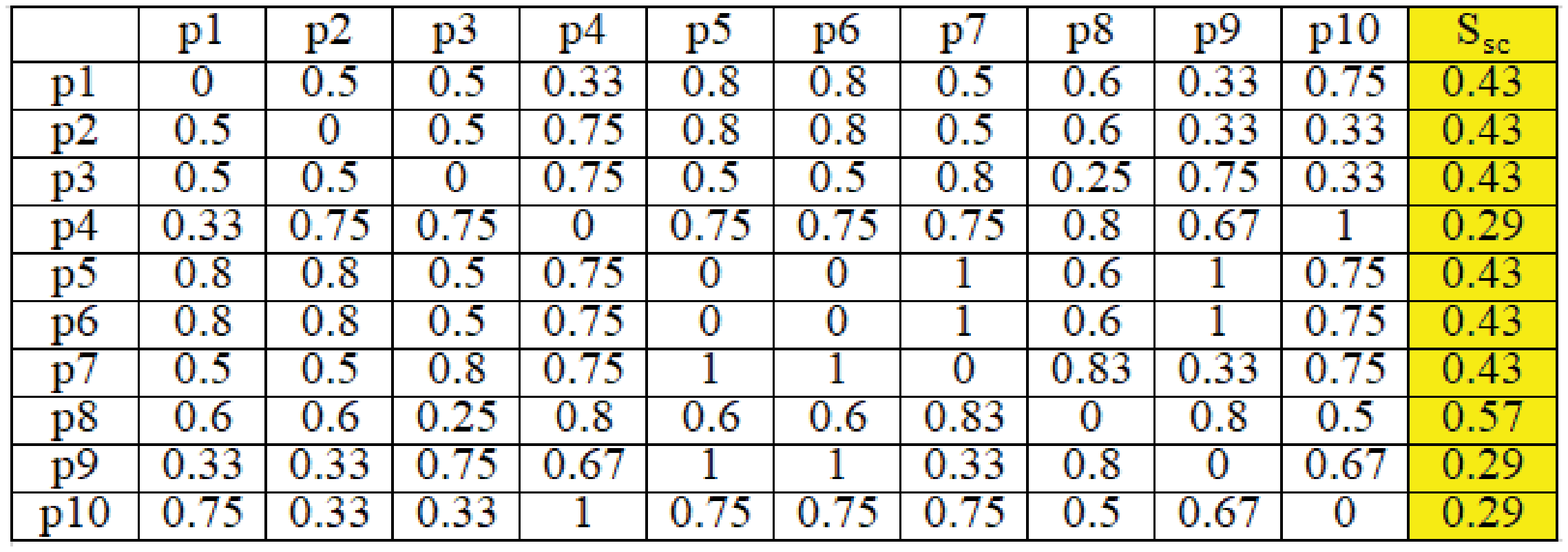}
	\caption{\jianxin{Social Diversity and Social Relevance Scores of $u$'s locations (refer Figure~\ref{fig:exmplFirst})}}
	\label{fig:relScoreExample}
\end{figure}

Now, to calculate the spatial relevance score of the location of $p_6\in L_u$ of $u$, we need to calculate $d_m = \max\{\min_{v\in V_u}dist(l_i, L_v)\}$ as the maximum value among the smallest distances between the location $p_6\in L_u$ and the location set $L_v$ of each friend $v \in V_u$, e.g., $v = \{a,b,c,e,f,g,h\}$. We will demonstrate first to calculate the value $\min_{v\in V_u}dist(l_i, L_v)$ using the check-in information of one friend $a\in V_u$ as reference. The check-ins of friend $a\in L_u$ are $\{p_1, p_2, p_4, p_7, p_9\}$. Therefore, among the location set $\{p_1, p_2, p_4, p_7, p_9\}$, we get $dist\{p_6,p_7\} =3.5$ as the minimum spatial distance among the location $p_6$ and the check-ins by the friend $a\in L_u$ (see Figure~\ref{fig:exmplFirst} for the relative distances between the points). Following this, we get $d_m = \max\{\min_{v\in V_u}dist(l_i, L_v)\} = \max\{3.5,3.5,3.5,0,0,0,9.5\} = 9.5$ for the location $p_6$ where, $V_u = \{a,b,c,e,f,g,h\}$. Therefore, we calculate the spatial relevance score of location $p_6$ as, $S_{sp}(p_6, u) = 1 - \frac{\sum_{v\in V_u}\min dist(l_i, L_v)} {d_m * |V_u|} = 1 - \frac{(3.5 + 3.5 + 3.5 + 0 + 0 + 0 + 9.5)} {9.5 * 7} = 0.699$. Now, we will calculate the social diversity between the locations $p_6$ and $p_2$ as an example. Among the locations of $u$, we get $maxD = 15$ as the maximum distance among the location pairs checked-in by user $u$ (see Figure~\ref{fig:exmplFirst} where distance between the pair $(p_6, p_5)$ is maximum as $dist(p_6, p_5) = 15$). Therefore, we calculate the spatial diversity between the location pair $(p_6, p_2)$ as $D_{sp}(p_6, p_2) = \frac{4}{15} = 0.27$. The spatial diversity and the spatial relevance scores of the locations are shown in Figure~\ref{fig:divScoreExample}.
\begin{figure}[htbp]
	\vspace{-3mm}
	\centering	
	\includegraphics[scale=0.38]{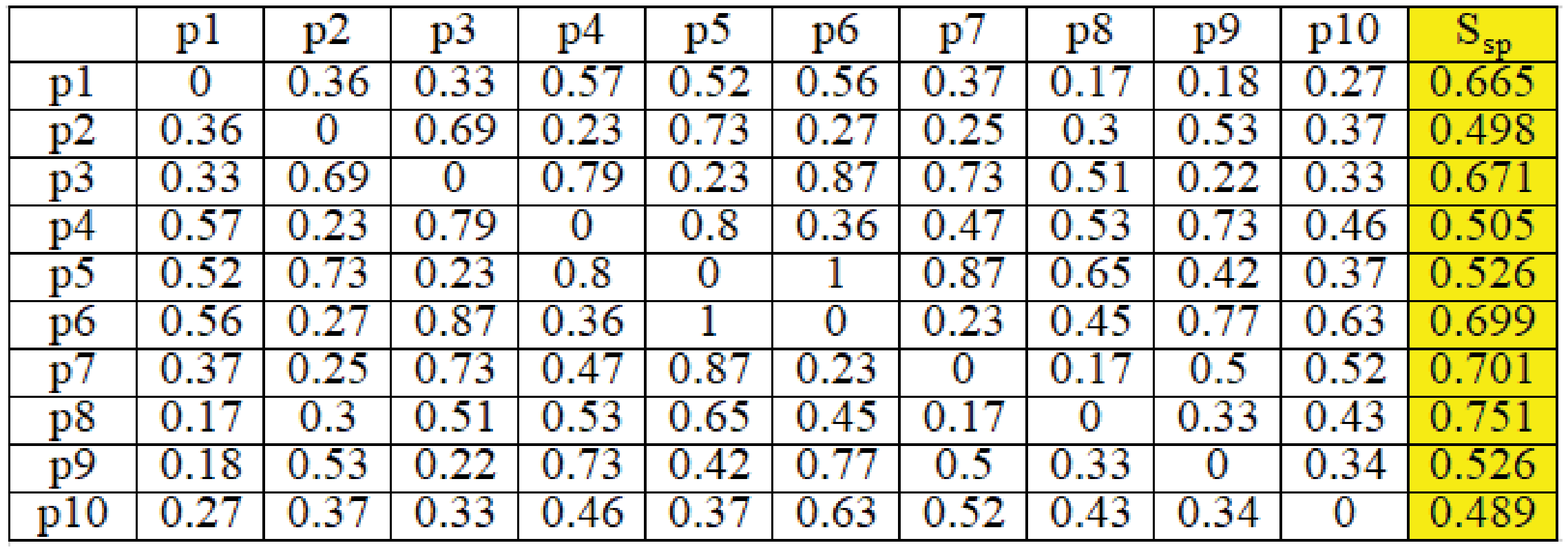}
	\caption{\jianxin{Spatial Diversity and Spatial Relevance Scores of $u$'s locations (refer Figure~\ref{fig:exmplFirst})}}
	\label{fig:divScoreExample}
\end{figure}

Now, we will calculate the socio-spatial relevance score and socio-spatial diversity of the locations considering equal weight in social and spatial factors, e.g., $\alpha = 0.5$. Therefore, we calculate the socio-spatial relevance score $R_{ss}(p_6, u)$ of $p_6$ as $R_{ss}(p_6, u) = 0.5*0.43 + 0.5*0.699 = 0.564$. Similarly, the socio-spatial diversity $D_{ss}(p_6,p_2)$ is calculated as $D_{ss}(p_6, p_2) = (0.5*0.80 +0.5*0.27) = 0.53$. The Table in Figure \ref{fig:scoreExample} shows the socio-spatial relevance scores ($R_{ss}$) and diversity of $u$'s locations calculated using $\alpha=0.5$. Note, we only need to  pre-compute the socio-spatial relevance $R_{ss}$ of the locations. We consider $\omega=0.5$ in the \textit{SSLS} query.
\begin{figure}[htbp]
	\vspace{-3mm}
	\centering	
	\includegraphics[scale=0.38]{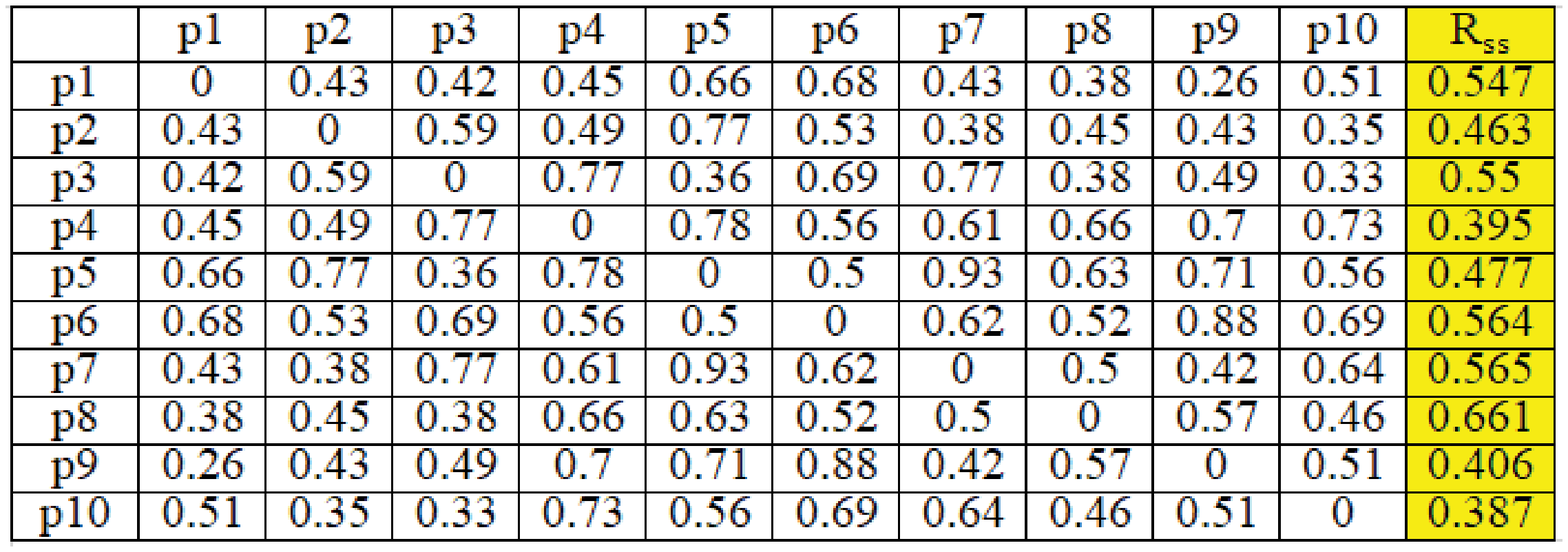}
	\caption{\jianxin{Socio-spatial Diversity and Socio-spatial Relevance Scores of $u$'s locations (refer Figure~\ref{fig:exmplFirst})}}
	\label{fig:scoreExample}
\end{figure}
\jianxin{Figure~\ref{fig:nodeExploration} illustrates the node exploration  towards searching for the top-$2$ \textit{SSLS} locations. Each state (node in tree) is marked with a number denoting the node exploration sequence. 
	A priority queue, $Q$ is initialized with $S_I = \emptyset$ and $S_R = \{p_8, p_7, p_6, p_3, p_1, p_5, p_2, p_9, p_4, p_{10}\}$, where $S_R$ contains $u$'s locations in non-increasing order of relevance scores ($R_{ss}$). The entries $(\emptyset, \{p_7, p_6, ..., p_4, p_{10}\}, 0)$ (Algorithm~\ref{Algo:top-k-ssls-basic-new}, Line~\ref{Algo1:QPush1}) and $(p_8, \{p_7, p_6, ..., p_4, p_{10}\}, 0.331)$ (Line~\ref{Algo1:QPush2}) are pushed to $Q$ for further exploration. Next, $(\{p_8\}, \{p_7, p_6, ..., p_4, p_{10}\}, 0.331)$ is dequeued from $Q$. We begin exploring from $p_8$, and $S_I$ becomes $\{p_8, p_7\}$ (step $2$). In the meantime, $(\{p_8\}, \{p_6, ..., p_4, p_{10}\}, 0.331)$ is pushed to $Q$ (Line~\ref{Algo1:QPush1}), 
	and we get the first feasible solution $S_b = \{p_8, p_7\}$ with $F(S_b) =$ $0.5*(0.661+0.565)+0.5*(0.5+0.5) = 1.113$. Continuing the process (till step $10$), the best feasible set $S_b = \{p_8, p_5\}$ with $F(S_b) = 1.199$ is obtained in this branch.}

\jianxin{Further, we dequeue $(\emptyset, \{p_7, p_6, ..., p_4, p_{10}\}, 0)$ and explore the branch with node $p_7$ (Step $11$). After processing lines \ref{Algo1:appendSI} and \ref{Algo1:QPush1} of Algorithm~\ref{Algo:top-k-ssls-basic-new}, we check pruning condition at Line~\ref{Algo1:pruneT} using Property~\ref{lemma:advPruning}, where $\widehat{D}{\Downarrow} = \frac{1.199 - 0.5*(0.565+0.564)}{0.5} - 0.93 = 0.339$ is computed using Equation~\ref{equ:lowerBoundTerminationOld}. As, $\widehat{D} > 0.339$ is true w.r.t. each location in current $S_R = \{p_7, p_6, ..., p_4, p_{10}\}$, we continue exploring the branch with node $p_7$ and update the best feasible set as $S_b = \{p_7, p_5\}$ with $F(S_b) = 1.451$. In the next iteration while exploring node $p_6$ (step20), we calculate $\widehat{D}{\Downarrow} = \frac{1.451 - 0.5*(0.564+0.55)}{0.5} - 0.88 = 0.908$ w.r.t. $S_I = \{p_6\}$ and $S_b = \{p_7, p_5\}$. All the locations in $S_R$ satisfy Property~\ref{lemma:advPruning}, therefore, we terminate processing $S_I=\{p_6\}$. By exploring the remaining branches, we obtain $S=\{p_7,p_5\}$ as the top-$2$ \textit{SSLS} set for the query user $u$.}	
\begin{figure}[htbp]
	\vspace{-3mm}
	\centering	
	\includegraphics[scale=0.65]{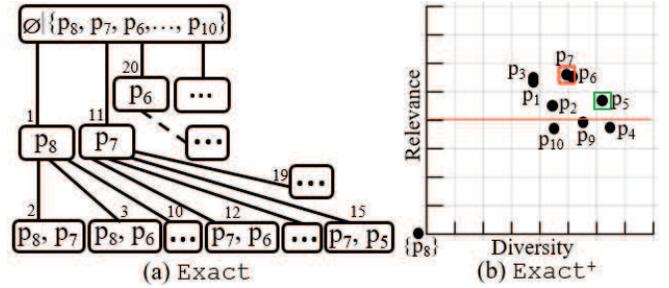}
	\caption{\jianxin{Node exploration steps of \texttt{Exact} and \texttt{Exact$^+$}}}
	\label{fig:nodeExploration}
	\vspace{-3mm}
\end{figure}

\section{An Approximate Approach}
\label{sec:Approximate}
\jianxin{One major limitation of \texttt{Exact} approach is the high computational cost, which makes it unrealistic for a large number of candidate locations. To validate the location pruning in \texttt{Exact}, it is required to calculate the updated diversity $\widehat{D}$ of the current intermediate set $S_I$ w.r.t. each location $l'\in S_R$\ignore{ using Equation~\ref{equ:D_hat}}. Calculating $\widehat{D}$ for each intermediate set is expensive when the size of $S_R$ is large.}


Therefore, to improve pruning and advanced termination, we derive relaxed bounds on diversity. We first define the maximum possible score of $\widehat{D}$ for an intermediate set $S_I$ when an eligible location $l'\in S_R$ is added to $S_I$. 
Lemma \ref{lemma:del_d} deduces $\widehat{D} \le D_{ss}(S_I)$ is true for any intermediate set $S_I$. 
Therefore, the maximum possible value of $\widehat{D}$ can be obtained as,
\begin{equation}
\small
\label{equ:max}
\widehat{D}_{max} = max(\widehat{D}) = D_{ss}(S_I)
\vspace{-2mm}
\end{equation}
\normalsize 
Note that $\widehat{D}$ is dependent on $S_R$, and $\widehat{D}_{max} = D_{ss}(S_I)$ is true w.r.t. $l'\in S_R$ only when $\forall l\in S_I, D_{ss}(l, S_I\setminus l) < D_{ss}(l',l)$ strictly holds. 
Therefore, to make an efficient approximate approach, we design pruning and termination in the below subsection using lower bound on $\widehat{d}$. We consider $\widehat{D} = \widehat{D}_{max} = D_{ss}(S_I)$ always true w.r.t. each location $l'\in S_R$. 

\textbf{Computing bounds on diversity of locations.}
Let us consider $S_I$ be an intermediate set of size $(k-1)$, $l'\in S_R$ be an \textit{eligible} location, and $S_b$ be the best feasible set identified already. For an arbitrary \textit{eligible} location $l'\in S_R$ that can be added to $S_I$, we continue to derive Equation \ref{equ:earlyTermExact}, 
\begin{displaymath}
\small
\begin{split}
& \hspace{0.5cm} \widehat{D} > \frac{1}{1 - \omega}\cdot \big(F(S_b) - \omega\cdot (R_{ss}(S_I) + \delta_r)\big) - \widehat{d}\\
& \Rightarrow \widehat{D} > \frac{1}{1 - \omega}\cdot \big(F(S_b) - \omega\cdot (R_{ss}(S_I) + \delta_r) - (1 - \omega)\cdot D_{ss}(S_I) \\
& \hspace{5cm} + (1 - \omega)\cdot D_{ss}(S_I)\big) - \widehat{d}\\
& \Rightarrow \widehat{D} > \frac{1}{1 - \omega}\cdot \big(F(S_b) - (\omega\cdot R_{ss}(S_I) + (1 - \omega)\cdot D_{ss}(S_I)) \\ & \hspace{4cm}  + (1 - \omega)\cdot D_{ss}(S_I) - \omega\cdot \delta_r \big) - \widehat{d}\\
& \Rightarrow \frac{F(S_b) - F(S_I)}{1 - \omega} + D_{ss}(S_I) - \widehat{D} - \frac{\omega}{1-\omega}\cdot \delta_r <  \widehat{d}
\end{split}
\end{displaymath}
\normalsize
We will first derive a relaxed bound for termination using the above equation. Therefore, we replace the upper bound of $\widehat{D}$ with its maximum possible value, e.g., $\widehat{D}_{max} = D_{ss}(S_I)$,    
\begin{displaymath}
\small
\begin{split}
& \Rightarrow \frac{F(S_b) - F(S_I)}{1 - \omega} + D_{ss}(S_I) - \widehat{D}_{max} -  \frac{\omega}{1-\omega}\cdot \delta_{r} < \widehat{d} \\
& \Rightarrow \frac{F(S_b) - F(S_I)}{1 - \omega} < \widehat{d} + \frac{\omega}{1-\omega}\cdot \delta_{r} \hspace{0.12cm} (\text{\normalsize putting \small } \widehat{D}_{max} = D_{ss}(S_I))\\
& \Rightarrow F(S_b) < F(S_I) + (1-\omega)\cdot \widehat{d} + \omega\cdot \delta_r
\end{split}
\end{displaymath}
\normalsize
Now, we will  generalize the above condition for any intermediate set $S_I$ of size $|S_I| < k$. So, we need to add an arbitrary subset $S'_R\subseteq S_R$ to $S_I$ such that $|S'_R| = (k-|S_I|)$. Hence, 
\begin{displaymath}
\small
\begin{split}
& F(S_b) < F(S_I) + (k- |S_I|)\cdot ((1-\omega).\widehat{d} + \omega\cdot \delta_r)\\
& \Rightarrow \frac{F(S_b) - F(S_I) - \omega\cdot(k- |S_I|)\cdot \delta_r}{(1-\omega)\cdot(k- |S_I|)}  < \widehat{d}\\
\end{split}
\end{displaymath}
\normalsize

Next, we will derive the lower bound $\widehat{d}{\downarrow}$ by replacing the expression $(k- |S_I|)\cdot\delta_{r}$ with the total socio-spatial relevance score of top $(k- |S_I|)$ relevant locations from $S_R$. 
We calculate the total socio-spatial relevance score of the top $(k- |S_I|)$ locations as $\max_{l' \in S_R}\sum_{k-|S_I|}R_{ss}(l')$. Hence, we get the lower bound of $\widehat{d}$ as follows,
\begin{equation}
\small
\label{equ:lowerBoundDelDashD}
\begin{split}
& \frac{F(S_b) - F(S_I) - \omega\cdot \max_{l' \in S_R}\sum_{k-|S_I|}R_{ss}(l')}{(1 - \omega)\cdot (k - |S_I|)}  = \widehat{d}{\downarrow}
\end{split}
\end{equation} 
\normalsize

\textbf{Pruning and Termination Rules.} We terminate processing an intermediate set $S_I$ when $\forall l'\in S_R$, $\widehat{d} \le \widehat{d}{\downarrow}$ is true. This is because, there exists no location in $S_R$ that can form a better set containing $S_I$ than the best feasible set $S_b$. Otherwise, we need to prune the particular locations $l' \in S_R$ that satisfy $D_{ss}(l', S_I)\le \widehat{d}{\downarrow}$.
As we consider, $\widehat{D}_{max} = D_{ss}(S_I)$ is always true for an intermediate set $S_I$ regardless of $S_R$, the derived lower bound $\widehat{d}{\downarrow}$ may produce the answer set to miss some eligible locations. Nevertheless, the approach achieves high efficiency with the sacrifice of a certain precision. 

\textbf{Algorithm.} 
For our Approximate (\textit{AP}) solution, we modify the \texttt{Exact} algorithm to introduce the advanced termination and pruning as described above. Here, we only need to replace the $pruneT$ methods at Line \ref{Algo:else2} in Algorithm~\ref{Algo:top-k-ssls-basic-new} using the above mentioned termination and pruning rules based on $\widehat{d}{\downarrow}$ when an intermediate set $S_I$ contains more than one location. \jianxin{The time complexity of \textit{AP} is similar to \texttt{Exact}, as both the algorithms execute same number of iterations in worst case.}




\jianxin{\textbf{Approximation Ratio}.
	We derive a theoretical bound on the approximation ratio of our Approximate approach (\textit{AP}). We define the ratio as the socio-spatial score of the \textit{SSLS} set returned by \texttt{Exact} algorithm divided by the score of the AP. Let's assume, $S'$ be the approximate set, and $S^*$ be the exact solution of size $k$, where locations $l'\in S_R$ are added progressively to $S'$ and $l^*\in S_R$ to $S^*$. To accelerate the Approximate approach, we had derived a relaxed lower bound $\widehat{d}{\downarrow}$ considering $\widehat{D}(S_I, l) = D_{ss}(S_I)$ is always true $\forall l\in S_R$ (refer Section V). This means, \textit{AP} will discard some \textit{eligible} locations $l \in S_E \subseteq S_R$, whose diversities (w.r.t. an intermediate set $S_I$) lie between $\widehat{d}{\downarrow} - (D_{ss}(S_I) - \widehat{D}(S_I, l)) \le D_{ss}(l, S_I) < \widehat{d}{\downarrow}$. Let, $\hat{l}^* \in S_E\subseteq S_R$ produces maximum socio-spatial score w.r.t. the  intermediate set, therefore, $\hat{l}^*$ will be part of the exact solution $S^*$. Similarly, let $\hat{l}'$ produces the maximum socio-spatial score among the locations whose diversity w.r.t. $S_I$ is more than $\widehat{d}{\downarrow}$, e.g., $D_{ss}(\hat{l}', S_I)\ge \widehat{d}{\downarrow}$. Hence, $D_{ss}(\hat{l}', S_I) > D_{ss}(\hat{l}^*, S_I)$, and $\hat{l}'$ will be part of the approximate solution. Therefore, $F(S_I\cup \hat{l}^*) > F(S_I\cup \hat{l}')$ is true, and we derive, $R_{ss}(\hat{l}^*) \ge R_{ss}(\hat{l}') + \frac{1-\omega}{\omega}\cdot (\psi_i)$, where $\psi_i = D_{ss}(\hat{l}', S_I) - D_{ss}(\hat{l}^*, S_I)>0$ is the difference in the diversity of the locations $\hat{l}'$ and $\hat{l}^*$ w.r.t. corresponding intermediate set (e.g., $S_I$).}

\jianxin{Following the above process, let's assume that we find the exact set $S^* = \{l^*_1,..., l^*_k\}$ and the approximate solution $S = \{l'_1,..., l'_k\}$ arranged in decreasing order of relevance score, and $\psi_k = D_{ss}(l_k', S') - D_{ss}(l_k^*, S^*)$. As we progressively add the locations, each time the lower bound $\widehat{d}{\downarrow}$ gets update. We assign $\widetilde{d} = \min \{\widehat{d}{\downarrow}\}$ as the minimum score among the lower bounds $\widehat{d}{\downarrow}$ we derived at each step. Therefore, the lowest total diversity of $S'$ will be $D_{ss} = k\cdot \widetilde{d}$, and the socio-spatial score of the lowest scoring approximate set is $F(S') = \omega\cdot \sum_{l'_i\in S'}R_{ss}(l'_i) + (1-\omega)\cdot k\cdot \widetilde{d}$. 
	Similarly, for the exact set, we calculate the best total diversity score as $k\cdot (\widetilde{d}-\epsilon)$, where the diversity of each location is slightly smaller by $\epsilon$ than $\widetilde{d}$. Also, we calculate the best total relevance score $R_{ss}(S^*) = \sum_{l^*_i\in S^*} R_{ss}(l_i^*) = \sum_{l'\in S'}R_{ss}(l') + \frac{1-\omega}{\omega}\cdot (\sum_{k}\psi_k)$. Let $\psi = \frac{1}{k}\cdot\sum_{k}\psi_k$, therefore, $F(S^*) = \omega\cdot (\sum_{l'\in S'}R_{ss}(l') + \frac{(1-\omega)}{\omega}\cdot k\cdot \psi) + (1-\omega)\cdot k\cdot (\widetilde{d}-\epsilon) = \omega\cdot \sum_{l'\in S'}R_{ss}(l') + (1-\omega)\cdot(k\cdot \psi) + (1-\omega)\cdot k\cdot (\widetilde{d}-\epsilon) = \omega\cdot \sum_{l'\in S'}R_{ss}(l') + (1-\omega)\cdot k \cdot (\widetilde{d} +\psi -\epsilon)$.}
\jianxin{Hence, the approximation ratio will be bounded by:
	\begin{displaymath}
	\small
	\frac{F(S^*)}{F(S')} = \frac{\omega\cdot \sum_{l'\in S'}R_{ss}(l') + (1-\omega)\cdot k \cdot (\widetilde{d} +\psi -\epsilon)}{\omega\cdot \sum_{l'\in S'}R_{ss}(l') + (1-\omega)\cdot k\cdot \widetilde{d}}
	\end{displaymath}
	\normalsize
	Let us assume, $\epsilon_{A} = \frac{\psi-\epsilon}{\widetilde{d}}$, s.t., $0\le \epsilon_{A} <1$. Now, if we emphasize on higher diversity (e.g., $\omega = 0$), the approximation ratio will be $1+\epsilon_{A}$ , while, it returns $1$ when emphasize on the  relevance (e.g., $\omega = 1$).	
}

\section{A Fast Exact Algorithm}
\label{sec:ExactPlus}
\jianxin{The socio-spatial diversity of a location is dependent on the other locations in a set. The pruning strategies based on the bound derived by diversity need to re-calculate the diversity scores of the locations whenever the intermediate set gets an update. Therefore, the algorithms based on the bound derived by diversity (e.g., \texttt{Exact}) consume more time to execute. In this section, we develop an efficient exact method (\texttt{Exact+}) that considers bounds on the relevance scores of the candidate locations. Such a practice will help to search the exact results by reducing the complex diversity computation of intermediate sets (as performed in \texttt{Exact}). The key idea of \texttt{Exact+} is motivated by the following observations: (i) As the relevance score of each member in a set is independent of the other members; it will be computationally efficient to design pruning strategies on relevance scores. (ii) Lemma \ref{lemma:postvGain} suggests that a location with a relevance score more than $\frac{(1-\omega)}{\omega}\cdot |\delta_{d}|$ is eligible to be added to an intermediate set. Therefore, we can easily derive a lower bound on relevance score using the above observations to prune a large number of irrelevant locations. 
}

\subsection{Computing Bounds on Relevance} 
\label{sec:relBound}
Here, we introduce some lemmas to derive bounds for pruning locations and early termination. 
The below lemma aims to compute the maximum possible socio-spatial diversity of a set $S'_I = S_I\cup l'$ when an arbitrary location $l'\in S_R$ is added to an intermediate result set $S_I$.
\begin{lemma}[Maximum Socio-spatial Diversity of an Updated Intermediate Set]
	\label{lemma:maxDivNewSet}
	Given an intermediate set $S_I$, an arbitrary location 
	$l' \in S_R$, the maximum Socio-spatial diversity $D^{M}_{ss}$ of an updated set $S'_I = S_I \cup l'$ will be $D^{M}_{ss}(S'_I) = D_{ss}(S_I) + D_{max}$,
	where $D_{max} = \max_{l' \in S_R}D_{ss}(l', S_I)$ is the maximum diversity generated by an arbitrary location of $S_R$ w.r.t. $S_I$.
\end{lemma}
\begin{proof}
	Socio-spatial diversity of $S'_I = S_I \cup l'$ is $D_{ss}(S'_I) = \widehat{d} + \widehat{D}$ (Equation \ref{equ:divNewSetS}). 
	Therefore, we get maximum socio-spatial diversity of $S'_I$ as, $D^{M}_{ss}(S'_I) = max(D_{ss}(S'_I)) = max(\widehat{D} + \widehat{d}) = max(\widehat{D}) + max(D_{ss}(l', S_I))$. Since $l'\in S_R$ is an arbitrary location, therefore, $D^{M}_{ss}(S'_I) = max(\widehat{D}) + \max_{l' \in S_R}D_{ss}(l', S_I) = max(\widehat{D}) + D_{max}$. Hence, $D^{M}_{ss}(S'_I) = D_{ss}(S_I) + D_{max}$, as $max(\widehat{D}) = D_{ss}(S_I)$ (Equation \ref{equ:max}).
\end{proof}

Now, we will derive the lower bound for the socio-spatial relevance score ($R^{\downarrow}_{ss}$). Such bound will identify the locations that can be added to the current intermediate set. Meanwhile, we label the \textit{reference} location ($l_{ref}$) that has maximum socio-spatial relevance score among the remaining locations in $S_R$.


\begin{lemma}[Lower Bound of Relevance Score]
	\label{lemma:lowerBoundRel}
	Given an intermediate set $S_I$, reference location $l_{ref}$, and the remaining location set $S_R$, the lower bound of Socio-Spatial Relevance Score is $R^\downarrow_{ss} = R_{ss}(l_{ref}) + \frac{(1-\omega)}{\omega}\cdot \big ( D_{ss}(S_I\cup l_{ref}) - D_{ss}(S_I) - D_{max} \big )$, where $D_{max} = \max_{l' \in S_R}D_{ss}(l', S_I)$ is the maximum diversity of locations in $S_R$ w.r.t. $S_I$.
\end{lemma}

\begin{proof}
	Suppose the \textit{reference} location $l_{ref}\in S_R$ has been added to the intermediate set $S_I$, the socio-spatial score of the updated intermediate set $S'_I = S_I\cup l_{ref}$ can be computed as,
	$F(S'_I)  = \omega\cdot R_{ss}(S_I\cup l_{ref}) + (1-\omega)\cdot D_{ss}(S'_I) 
	= \omega\cdot (R_{ss}(S_I)+R_{ss}(l_{ref})) + (1-\omega)\cdot D_{ss}(S'_I)$

	Given another location $l'\in S_R\setminus l_{ref}$ s.t. $S''_I = S_I\cup l'$, it needs to be probed only when $F(S_I\cup l')>F(S'_I)$ according to the selection criteria. Hence, we simplify the condition below.
	\vspace{-3mm}
	\begin{displaymath}	
	\small
	\label{equ:approxAppInc}
	\begin{split}
	& \omega\cdot (R_{ss}(S_I)+R_{ss}(l')) + (1-\omega)\cdot D_{ss}(S''_I) > \\ 
	& \hspace{2.5cm} \omega\cdot (R_{ss}(S_I) + R_{ss}(l_{ref})) + (1-\omega)\cdot D_{ss}(S'_I) \\
	& \Rightarrow \omega\cdot R_{ss}(l') > \omega\cdot R_{ss}(l_{ref}) + (1-\omega)\cdot \big(D_{ss}(S'_I) - D_{ss}(S''_I)\big) \\
	& \Rightarrow R_{ss}(l') > R_{ss}(l_{ref}) + \frac{(1-\omega)}{\omega}\cdot \big(D_{ss}(S'_I) - D_{ss}(S_I\cup l')\big)
	\end{split}
	\end{displaymath}
	\normalsize	
	Now, we substitute $D_{ss}(S_I\cup l')$ with its maximum value $D_{ss}(S_I) + D_{max}$ using Lemma \ref{lemma:maxDivNewSet}. Therefore, we get $R^\downarrow_{ss} = R_{ss}(l_{ref}) + \frac{(1-\omega)}{\omega}\cdot \big ( D_{ss}(S'_I) - D_{ss}(S_I) - D_{max} \big )$.
\end{proof}

\vspace{-2mm}
If $S_I$ contains single location, we compute the lower bound as $R^\downarrow_{ss} = R_{ss}(l_{ref}) + \frac{(1-\omega)}{\omega}\cdot \big ( D_{ss}(l_{ref}, S_I) - D_{max} \big )$, as $D_{ss}(S'_I) - D_{ss}(S_I) = D_{ss}(l_{ref}, S_I)$ if $|S_I| = 1$. Now, using Lemma \ref{lemma:lowerBoundRel}, we identify the potential locations that can be added to the current intermediate set. 
\begin{property}[Potential Locations]
	\label{prop:potLoc}
	A location $l\in S_R$ is a potential candidate location w.r.t. $S_I$ if $R_{ss}(l) \ge R^\downarrow_{ss}$.
\end{property}

\subsection{Advanced Termination}
\label{sec:advTerm}
The \texttt{Exact$^+$} algorithm needs to iteratively check the remaining locations until the best result set is determined.
However, it is time-consuming to process all intermediate sets and checks for the feasible set at each iteration. 
Therefore, we need to introduce some lemmas to derive early termination criteria. First, similar to Lemma \ref{lemma:maxDivNewSet}, we derive below lemma\ignore{Lemma~\ref{lemma:condition11}} on maximum possible socio-spatial diversity of an answer set.
\begin{lemma}[Maximum Socio-spatial Diversity of an Answer Set]
	\label{lemma:condition11}
	Given set $S_I$, an arbitrary subset of locations $S'_R\subseteq S_R$ of size $(k-|S_I|)$ s.t., $S' = S_I\cup S'_R$ and $S_I\cap S'_R= \emptyset$; the maximum Socio-spatial diversity $D^M_{ss}(S')$ of the set $S' = S_I\cup S'_R$ is $D^{M}_{ss}(S') = D_{ss}(S_I) + D^{Max}_{ss}$, where $D^{Max}_{ss} = \max_{l'\in S_R}(\sum_{k-|S_I|}D_{ss}(l', S_I))$ is the sum of the top $(k-|S_I|)$ socio-spatial diversity scores of the locations $l'\in S_R$ w.r.t. $S_I$.
\end{lemma}
\begin{proof}
	\vspace{-2mm}
	Proof is omitted due to space limitations.
\end{proof}

For any intermediate set $S_I$ and a feasible solution $S'$ of size $k$ containing $S_I$, s.t. $S_I\subset S'$, we derive the below lemma on maximum possible socio-spatial score of $S'$. 
\begin{lemma}[Maximum Socio-spatial Score of an Answer Set]
	\label{lemma:maxScoreUpdated1}
	Given an intermediate set $S_I$, an arbitrary subset of locations $S'_R\subseteq S_R$ of size $(k-|S_I|)$ s.t., $S' = S_I\cup S'_R$, the maximum possible socio-spatial score of $S'$ is $F_{max}(S') = F(S_I) + \omega.R^{Max}_{ss}(S'_R) + (1-\omega).D^{Max}_{ss}$,
	where, $R^{Max}_{ss}(S'_R) =  \max_{l'\in S_R}(\sum_{k-|S_I|}R_{ss}(l'))$ is the sum of top $(k-|S_I|)$ socio-spatial relevance scores of the remaining set $S_R$, s.t., $S_R\supseteq S_R'$ and $D^{Max}_{ss} = \max_{l'\in S_R}(\sum_{k-|S_I|}D_{ss}(l', S_I))$.		
	\normalsize
\end{lemma}
\begin{proof}
	Let an arbitrary location set $S'_R\subseteq S_R$ of size $(k - |S_I|)$ is added to $S_I$ s.t. $S' = S_I\cup S'_R$. The socio-spatial score $F(S')$ of $S'$ is, $F(S') = \omega.R_{ss}(S_I\cup S'_R) + (1-\omega).D_{ss}(S')$ 
	\begin{equation}
	\hspace{-2mm}
	\label{equ:lemmaFmax}
	\small
	\begin{split}
	\hspace{-0.5cm}& \Rightarrow F(S') = \omega\cdot R_{ss}(S_I) + \omega\cdot R_{ss}(S'_R) + (1-\omega)\cdot D_{ss}( S')
	\end{split}
	\end{equation}
	\normalsize	
	To achieve the maximum socio-spatial score $F_{max}(S')$ of $S'$, we need to replace the two unknown variables $R_{ss}(S'_R)$ and $D_{ss}( S')$ in Equation~\ref{equ:lemmaFmax} with their maximum possible scores. 
	
	As $S'_R\subseteq S_R$ is an arbitrary subset of $S_R$, the maximum possible socio-spatial relevance score $R_{ss}(S'_R)$ of $S'_R$ can be calculated as $R^{Max}_{ss}(S'_R) = \max_{l' \in S_R}\sum_{k-|S_I|}R_{ss}(l')$. Similarly, from Lemma \ref{lemma:condition11}, we get the maximum possible socio-spatial diversity of $S'$ as $D^{M}_{ss}(S') = D_{ss}(S_I) + D^{Max}_{ss}$. After substituting $R_{ss}(S'_R)$ with $R^M_{ss}(S'_R)$, and $D_{ss}(S')$ with $D^{M}_{ss}(S')$ in Equation~\ref{equ:lemmaFmax}, we get $F_{max}(S') = \omega\cdot R_{ss}(S_I) + \omega\cdot R^{Max}_{ss}(S'_R)  + (1-\omega)\cdot (D_{ss}(S_I) + D^{Max}_{ss})$. Therefore, the lemma is proved as $\omega\cdot R_{ss}(S_I) + (1-\omega)\cdot D_{ss}(S_I) = F(S_I)$.
\end{proof}
\normalsize
\begin{property}[Advanced Termination]
	\label{prop:advTerm}
	Given an intermediate set  $S_I$, a $k$-sized answer set $S'\supset S_I$, and the best feasible set $S_b$, if $F(S_b) > F_{max}(S')$, we will terminate processing $S_I$.
\end{property}

\jianxin{The \texttt{Exact$^+$} algorithm incrementally adds locations and checks for a feasible set. It prunes some locations using the lower bound on relevance score, and further terminates processing large number of intermediate sets using Property~\ref{prop:advTerm}.}

\subsection{Algorithm}
\jianxin{Algorithm~\ref{Algo:ExactPlus} summarizes the major steps of \texttt{Exact$^+$}, for processing the \textit{SSLS} query. 
	Given a socio-spatial graph $G$, query user $u$, the top-$k$ \textit{SSLS} query returns a set $S$ of size $k$. Initially, the locations of user $u$ \jianxin{are} added to $S_{Rel}$ in non-increasing order of their relevance scores and marked as unvisited. In each iteration, the unvisited locations of $S_{Rel}$ are copied to $S_R$, and the top relevant location of $S_{R}$ is added to $S_I$ (Line \ref{Algo2:updateSetPop}). Further, the advanced termination of the current intermediate set is probed using Property \ref{prop:advTerm} (Line \ref{Algo2:earlyTerm}). 
	In Line \ref{Algo2:relBound}, the lower bound on relevance score ($R^\downarrow_{ss}$) is calculated using\ignore{Equation \ref{equ:lowerBRelScore}} Lemma~\ref{lemma:lowerBoundRel}, and the potential locations ($V_P$) are identified using Property \ref{prop:potLoc} (Line \ref{Algo2:potLoc}). The intermediate set $S_I$ is updated with the location $l_{top}\in V_P$ that generates maximum socio-spatial score (Line \ref{Algo2:updateList2}). The inner loop continues until a set of $k$ locations is found, and finally, it returns the best set $S$. 
} 

\begin{algorithm}[tbh]
	\label{Algo:ExactPlus}
	\small
	\caption{SSLS: \texttt{Exact$^+$}} 
	\KwIn{Socio-spatial graph \textit{$G$}, set size $k$, query user $u$}	
	\KwOut{Location set $S$ of size $k$} 
	Initialize: $S_I\gets \emptyset$, $S\gets \emptyset$, $bestScore\gets 0$,
	
	append $\langle l, R_{ss}(l, u) \rangle$ into $S_{Rel}$ in non-increasing $R_{ss}$
	
	mark all locations of $S_{Rel}$ unvisited 
	
	\While{no unvisited location exist in $S_{Rel}$}{ \label{Algo2:while1}
		$S_R\gets unvisited(S_{Rel})$ \label{Algo2:lineUnvisited}
		
		
		
		$l\gets S_R.pop(0); S_I.append(l)$\label{Algo2:updateSetPop}
		
		\While {$|S_I| < k$ and $|S_I|+|S_R|\ge k$}{							
			\If{$advTerm(bestScore, S_I, S_R, k)$ 
			} {\label{Algo2:earlyTerm}
				break
			}
			
			$l_{ref} \gets topLocation(S_R)$ \label{Algo2:l_ref}		
			
			
			$R^\downarrow_{ss} \gets relBound(l_{ref}, S_I, S_R)$ \label{Algo2:relBound} 
			
			$V_{P} \gets potentialLocs(S_R, R^\downarrow_{ss})$ \text{*** Property \ref{prop:potLoc}} \label{Algo2:potLoc}
			
			
			$l_{top} \gets \arg\max_{l_i\in V_P}F(S_I\cup l_i)$ \label{Algo2:lTop}
			
			$S_I.append(l_{top}); S_R.remove(l_{top})$ \label{Algo2:updateList2}
			
		}
		
		\If{$|S_I|==k$}{
			\If{$F(S_I) > bestScore$}{
				$bestScore\gets F(S_I)$;
				$S\gets S_I$ 
			}	
			$S_I\gets \emptyset$ 
		}
		
		mark $l$ in $S_{Rel}$ as visited\\
		
	}\vspace{-1mm}
\end{algorithm}

\textbf{Time Complexity.} \jianxin{The worst case time complexity of \texttt{Exact$^+$} algorithm is $O(n^3k)$, where $n$ is the number of locations of a user. 
	The outer loop and inner loop take $O(n)$ and $O(k)$, respectively. The complexity of other major parts are: $advTerm$ process in $O(n^2)$, $topLocations$ selection in $O(1)$, $relBound$ computation in $O(n^2)$, $potentialLocs$ selection in $O(n)$, and $l_{top}$ selection in $O(n)$. 
}

\textbf{Steps of \texttt{Exact$^+$}.} \jianxin{We use the example in Figure~\ref{fig:exmplFirst} to demonstrate the steps of \texttt{Exact$^+$} (Algorithm~\ref{Algo:ExactPlus}) for selecting top-$2$ \textit{SSLS} set for user $u$. 
	We show the node exploration steps of the first iteration of \texttt{Exact$^+$} in Figure~\ref{fig:nodeExploration} (b). The calculated relevance and diversity scores of the locations are available in Figure~\ref{fig:scoreExample}. 
	We set the trade-off parameters as $\alpha=0.5$, $\omega=0.5$.}

\jianxin{First, we add $u$'s locations in $S_{Rel} = \{p_8, p_7, .., p_{10}\}$ in non-increasing $R_{ss}$ score. Next, the top relevant location $p_8$ (shown in left bottom corner of Figure \ref{fig:nodeExploration} (b)) is added to intermediate set $S_I$. As the termination condition is not satisfied at Line \ref{Algo2:earlyTerm}, we process to explore the remaining locations in $S_R$. First, we select the reference location as $l_{ref} = p_7$ (Line~\ref{Algo2:l_ref}) shown within red box in Figure \ref{fig:nodeExploration} (b), and the remaining locations in $S_R$ are shown as black dots. The \textit{Y} and \textit{X} axes denote the relevance scores and diversity of the locations in $S_R$ w.r.t. $p_8$, respectively. Now, the lower bound in relevance score w.r.t. $l_{ref} = p_7$ is calculated using Lemma \ref{lemma:lowerBoundRel}, e.g., $R^\downarrow_{ss} = 0.565 + \frac{(1-0.5)}{0.5}(0.5 - 0 - 0.66) = 0.405$. The horizontal line in red depicts the lower bound in relevance score. The points $\{p_7, p_6, .., p_2, p_9\}$ on or above the line are labeled as potential locations ($V_P$), and $\{p_4, p_{10}\}$ are pruned w.r.t. intermediate set $S_I = \{p_8\}$. The location $p_5$ (marked in green box) among $V_P$ produces the maximum score $F(S_I\cup p_5) = 1.199$ (Line \ref{Algo2:updateList2}). Therefore, in this iteration, we get the best set of $2$ locations as $\{p_8, p_5\}$. The process continues until no unvisited locations exist in $S_{Rel}$. We finally get $S=\{p_7, p_5\}$ as the top-$2$ \textit{SSLS} solution for $u$ with socio-spatial score $F(S) = 1.451$.
}

\textbf{Fast Approximate.} From our empirical evaluation, we find that greedily selecting the best locations using \texttt{Exact$^+$}, the results rapidly converge towards an optimal solution in the first few iterations. To make a reasonable trade-off between performance and accuracy, we consider an early termination of \texttt{Exact$^+$} after two iterations in our Fast Approximate (\textit{FA}) algorithm. \jianxin{In Figure~\ref{fig:exmplFirst}, \textit{FA} will select $S=\{p_7, p_5\}$ for user $u$ considering the first two iterations of \texttt{Exact$^+$}.}

\section{Experimental Evaluation}
\label{sec:experiment}
In this section, we present the experimental evaluation of our proposed approaches for \textit{Top-$k$ SSLS} queries: the \texttt{Exact} solution (\textit{E}); the \textit{Approximate} solution (\textit{AP}); the \texttt{Exact$^+$} solution (\textit{EP}); and the \textit{Fast Approximate} solution (\textit{FA}). We implement the algorithms using Python 3.6 on Windows environment with 3.40GHz CPU and 64GB RAM. To further \jianxin{validate}, we compared with three baselines adapted from existing works:\\
$\bullet$ \textbf{GMC \cite{51vieira2011query}.} It combines relevance and diversity, and greedily selects the elements based on their marginal contributions. Locations with the highest partial contributions will be selected.\\
$\bullet$ \textbf{Adaptive-SOS \cite{22guo2018efficient}.} 
To make the adaption of \textit{SSLS} to SOS \cite{22guo2018efficient}, denoted as \textit{AS}, we model the social similarity of a pair of locations by using their common users who checked in the location pair. Thus, an edge can be added between the two locations if the similarity is more than a threshold (e.g., 0.4).\\
$\bullet$ \textbf{GNE \cite{51vieira2011query}.} It randomly adds a location from the top ranked locations into a temporary result set. Then, it performs swaps between elements of the temporary result set and the most diverse elements of the candidate set. 
\begin{table}
	\centering
	\small
	\caption{\jianxin{Dataset Statistics}}	
	\scalebox{0.93}{
		\begin{tabular}{|p{0.42cm}|p{0.85cm}|p{1.1cm}|p{1.25cm}|p{1.1cm}|p{0.25cm}|p{0.37cm}|p{0.42cm}|\ignore{p{2.1cm}|p{1.6cm}|}}
			\hline 
			Data & Users & Edges & Checkins & Places & AC\ignore{Avg Checkins} & AF\ignore{Avg Friends} & AFC\ignore{Avg Friends with Exact checkin} \ignore{& Time Period & AvgUnqChkn\ignore{Unique-checkins}}\\ 
			\hline 
			GW  & 107,092 & 456,830 & 6,442,892 & 1,280,969 & 60 & 8.5 & 4.6 \ignore{& Feb 09 - Oct 10 & 24}\\ 
			\hline 
			BK & 51,405 & 214,078 & 4,491,143 & 772,783 & 87 & 7.7 & 3.8 \ignore{& Apr 08 - Oct 10 & 22}\\ 
			\hline
			FL & 189,537 & 2,028,873\ignore{4,057,746} & 12,592,819\ignore{46,342,145} & 4,896,634\ignore{12,254,541} & 66 & 21.4 & 0.3\ignore{& & 49} \\ 
			\hline
			YL & 270,323 & 1,913,501\ignore{3,827,002} & 5,425,778 & 192,609 & 20 & 14.2 & 10.4 \ignore{& & 16}\\ \hline 
		\end{tabular} 
	}
	\label{tab:dataset}
	\vspace{-3mm}
\end{table}
\noindent\textbf{Datasets.} \jianxin{We conduct experiments using four real-world large datasets: \textit{Gowalla} (GW), \textit{Brightkite} (BK), \textit{Flickr} (FL), and \textit{Yelp} (YL). Gowalla \cite{snapnets} and Brightkite \cite{snapnets}, each contains the social connections of the users, and the check-ins available over the period Feb. 2009 - Oct. 2010 and Apr. 2008 - Oct. 2010 respectively. 
	\textit{Flickr} data was collected using Flickr public API in 2017-18. We establish a social link between a user pair using the \textit{following} information, and consider a check-in if a user has a photo geo-tagged the location. \textit{Yelp} (collected from {https://www.yelp.com/dataset/}, Round 13, Year 2019) contains friendship network and POIs of users in the form of reviews, and location-tags in users' tips. 
	Table~\ref{tab:dataset} presents brief statistics of the four datasets, where the last three columns show the Average Check-ins (AC) by users, Average Friendships (AF), and Average number of Friends that users have at places they have Checked-in (AFC).}
In Figure~\ref{fig:friendshipDistribution}, we show the number of users have friends in the given ranges, where the x-axis labels `100', `200', `500', `1K' , `$>$1K' denote the number of friends in the ranges `10-100', `101-200', `201-500', `501-1K', and `$>$K' respectively.
\begin{figure}[htbp]
	\vspace{-5mm}
	\centering
	\subfigure[Gowalla]{\label{friendshipGW}
		\includegraphics[scale=0.19]{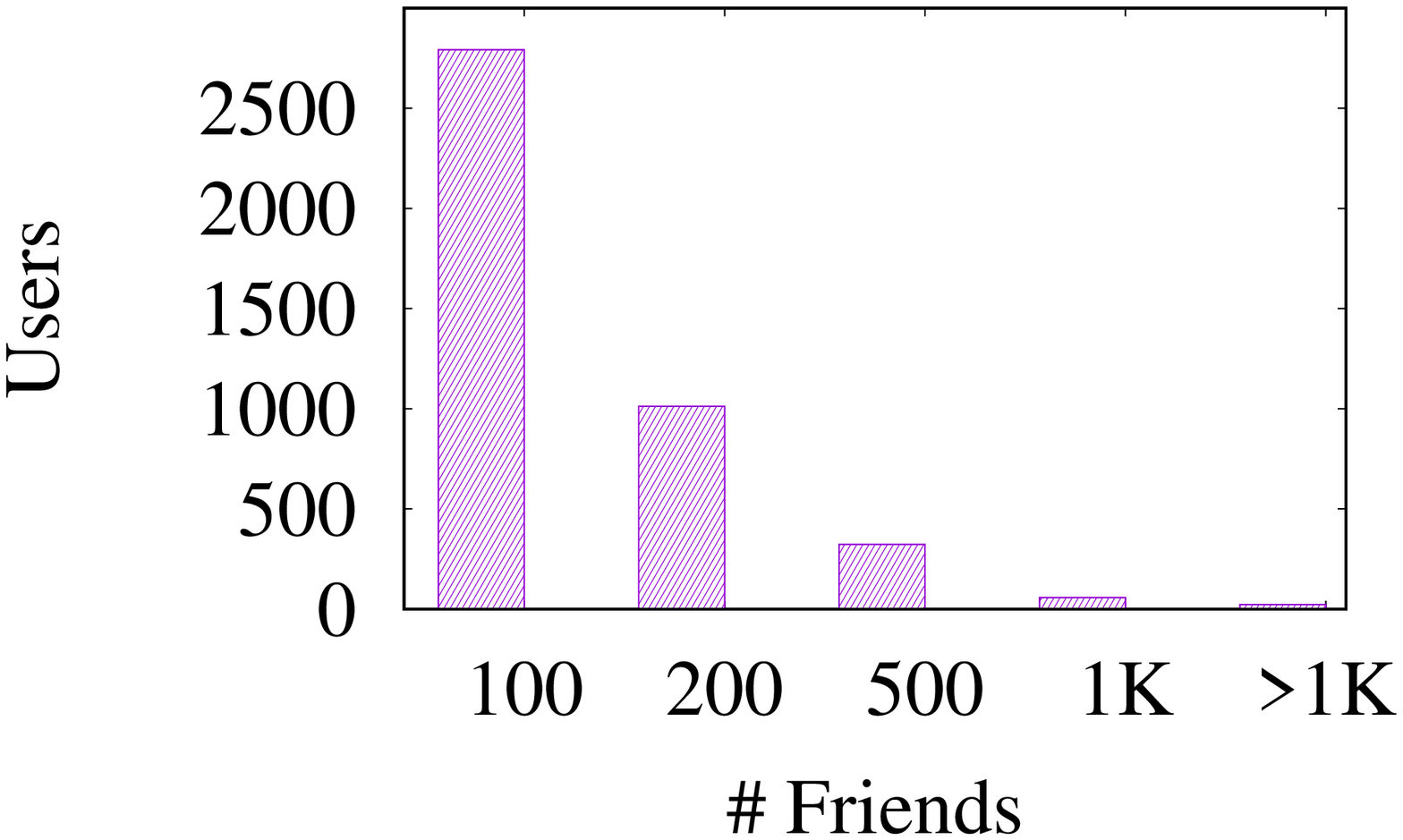}
	}
	\subfigure[Brightkite]{\label{friendshipBK}
		\includegraphics[scale=0.19]{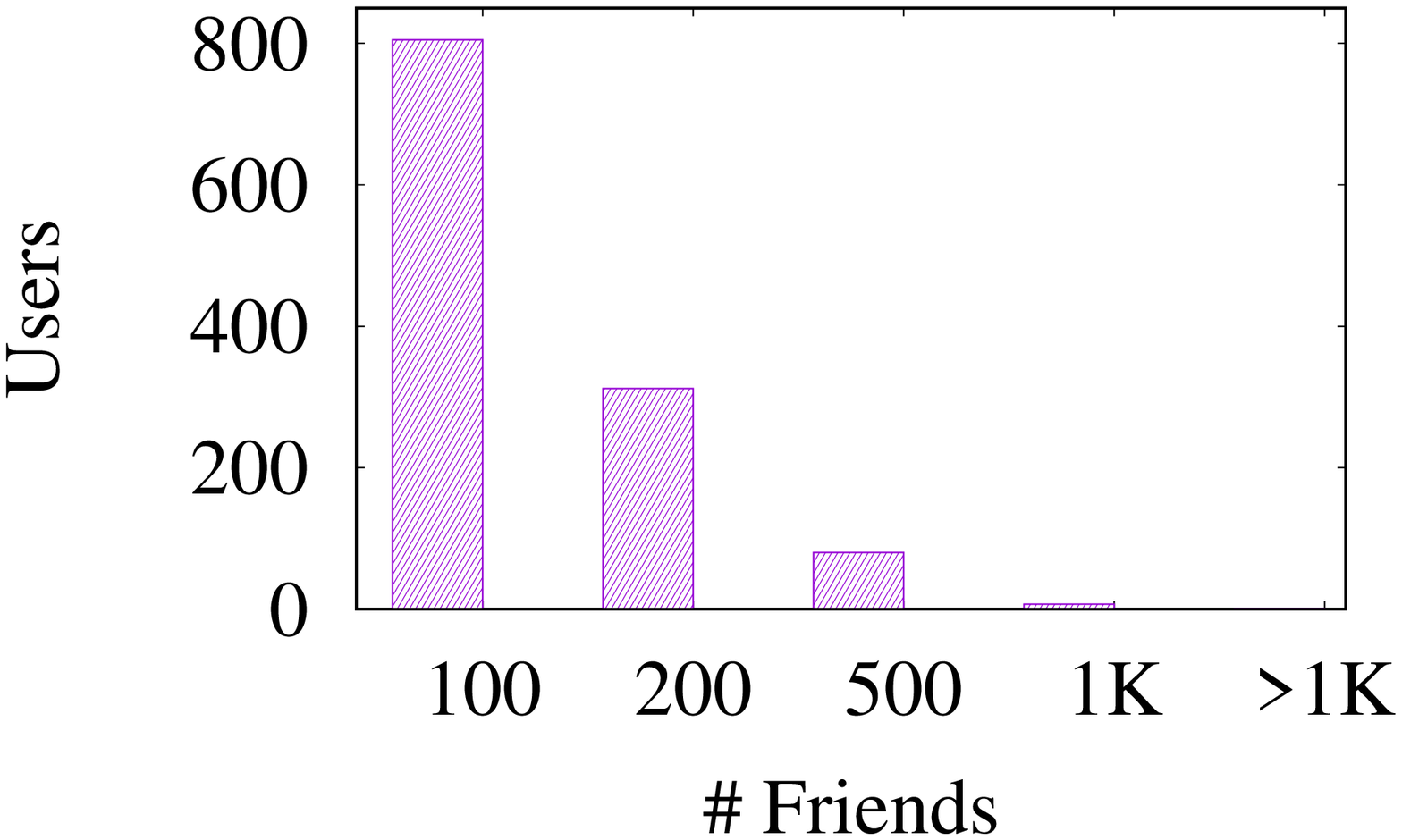}	
	}
	\subfigure[Flickr]{\label{friendshipFL}
		\includegraphics[scale=0.19]{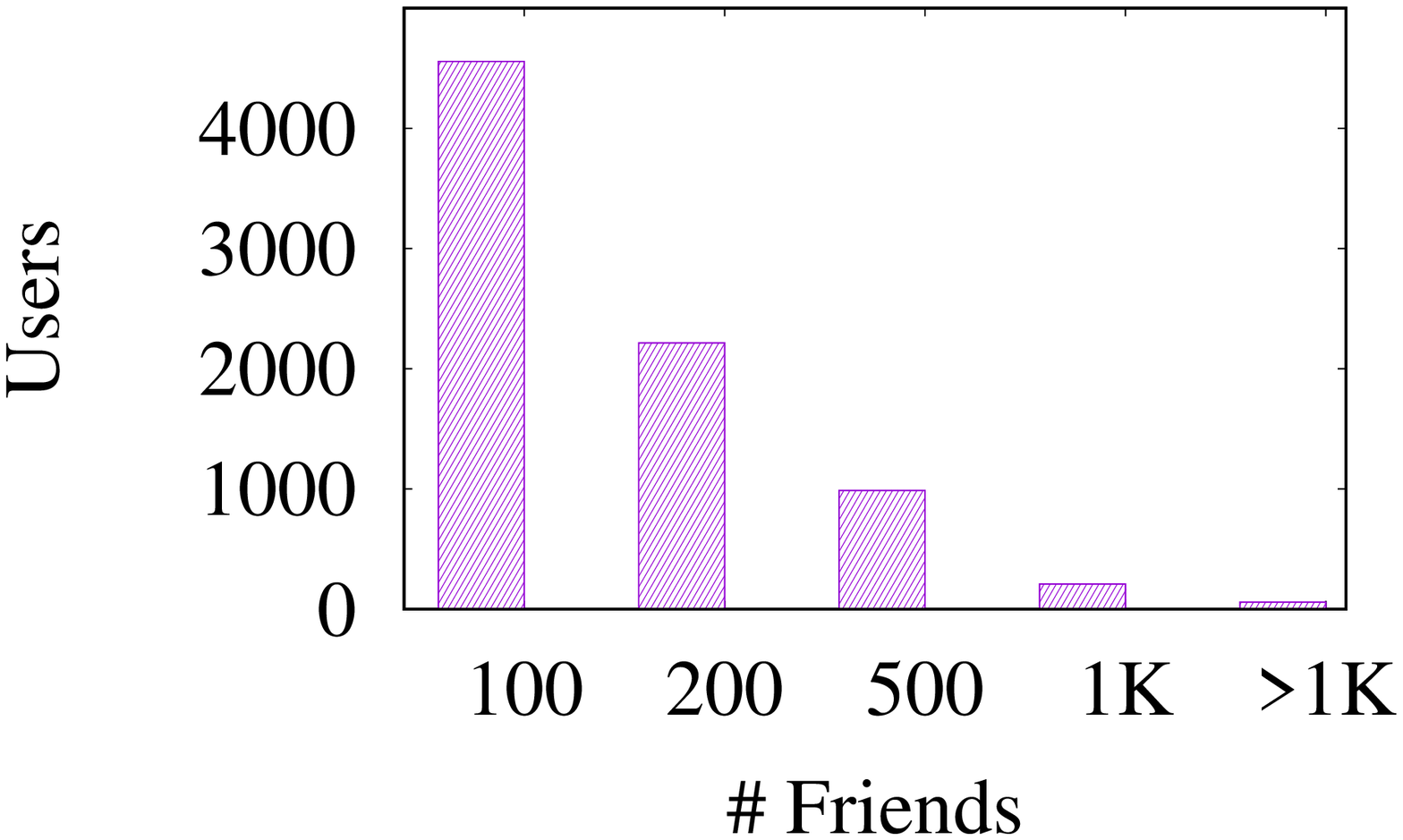}
	}
	\subfigure[Yelp]{\label{friendshipYL}
		\includegraphics[scale=0.19]{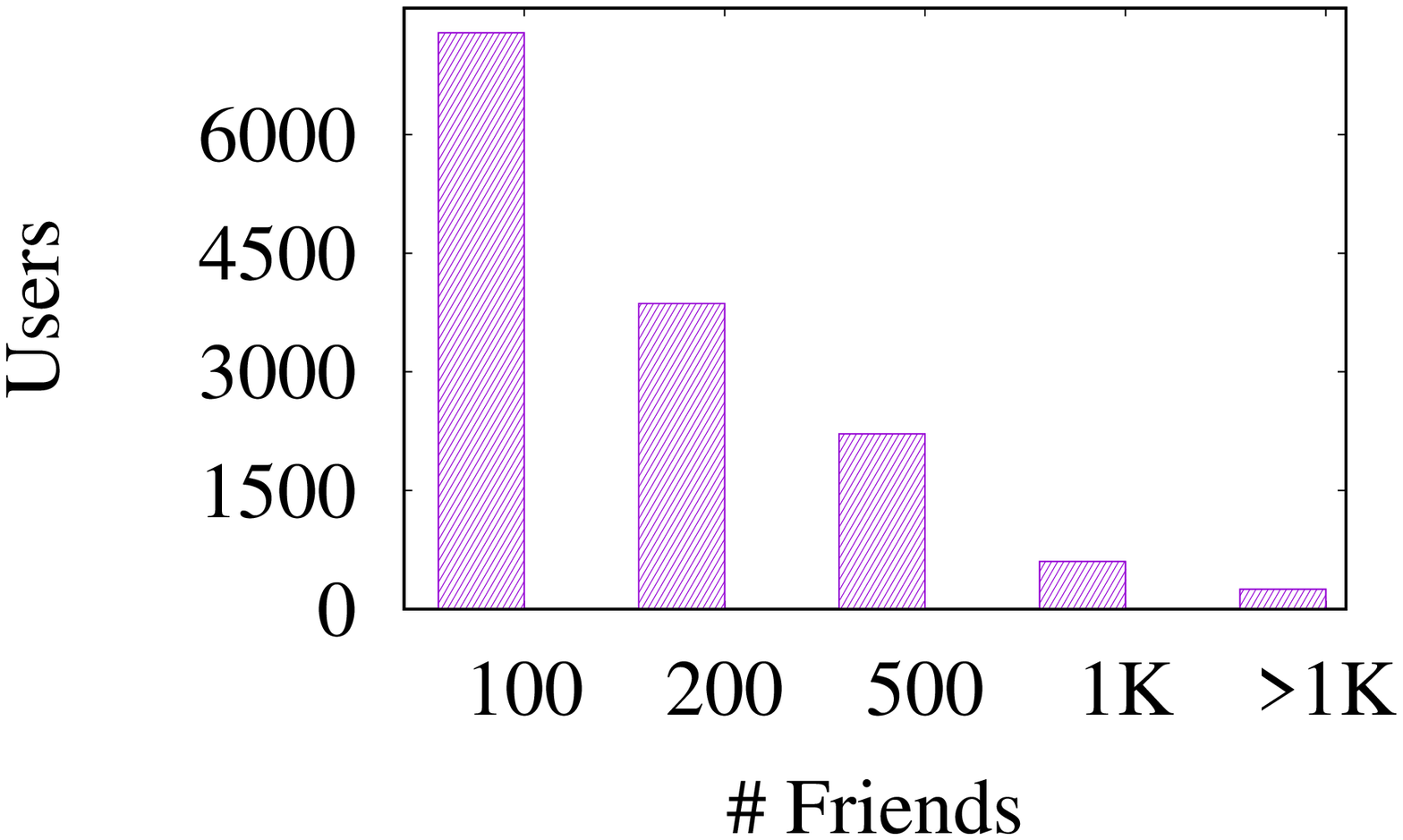}
	}
	\caption{\jianxin{Friendship Distribution}}
	\label{fig:friendshipDistribution}
	\vspace{-3mm}
\end{figure}

Figure~\ref{fig:checkinStats} shows the check-in characteristics of the users in different check-in ranges.
\begin{figure}[htbp]
	\vspace{-5mm}
	\centering
	\subfigure[Gowalla]{\label{checkinGW}
		\includegraphics[scale=0.19]{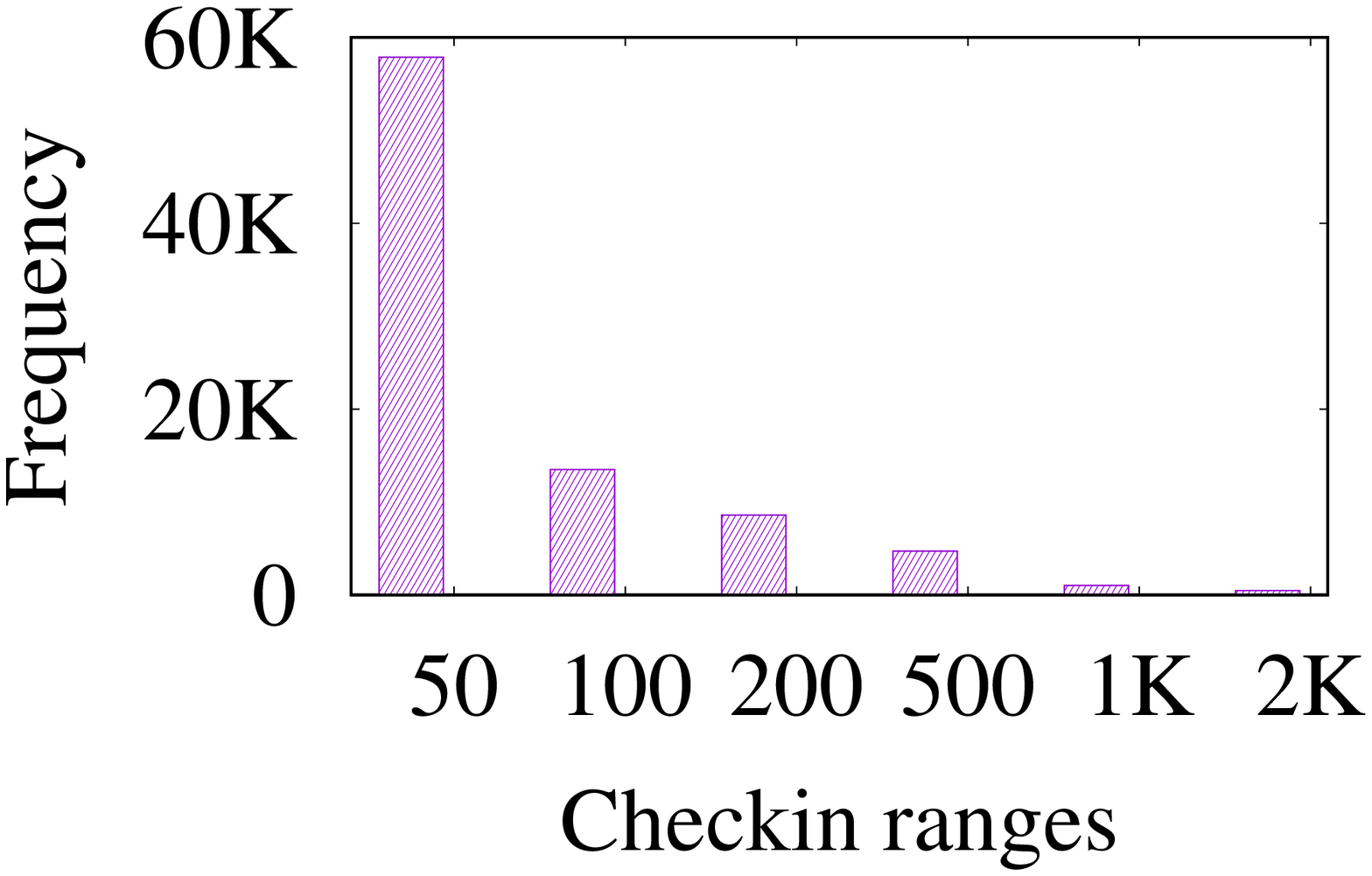}
	}
	\subfigure[Brightkite]{\label{checkinBK}
		\includegraphics[scale=0.19]{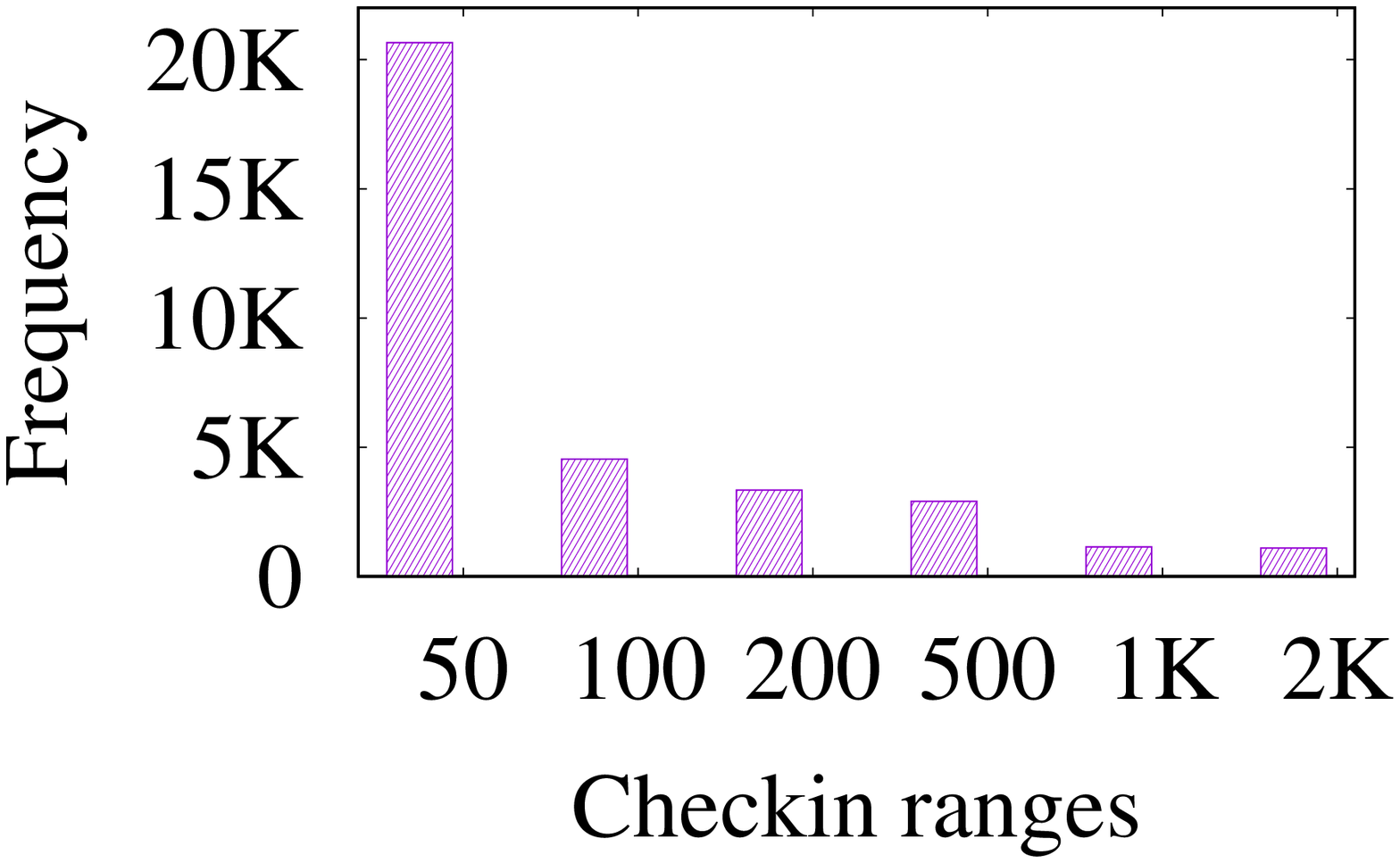}	
	}
	
	\subfigure[Flickr]{\label{checkinFL}
		\includegraphics[scale=0.19]{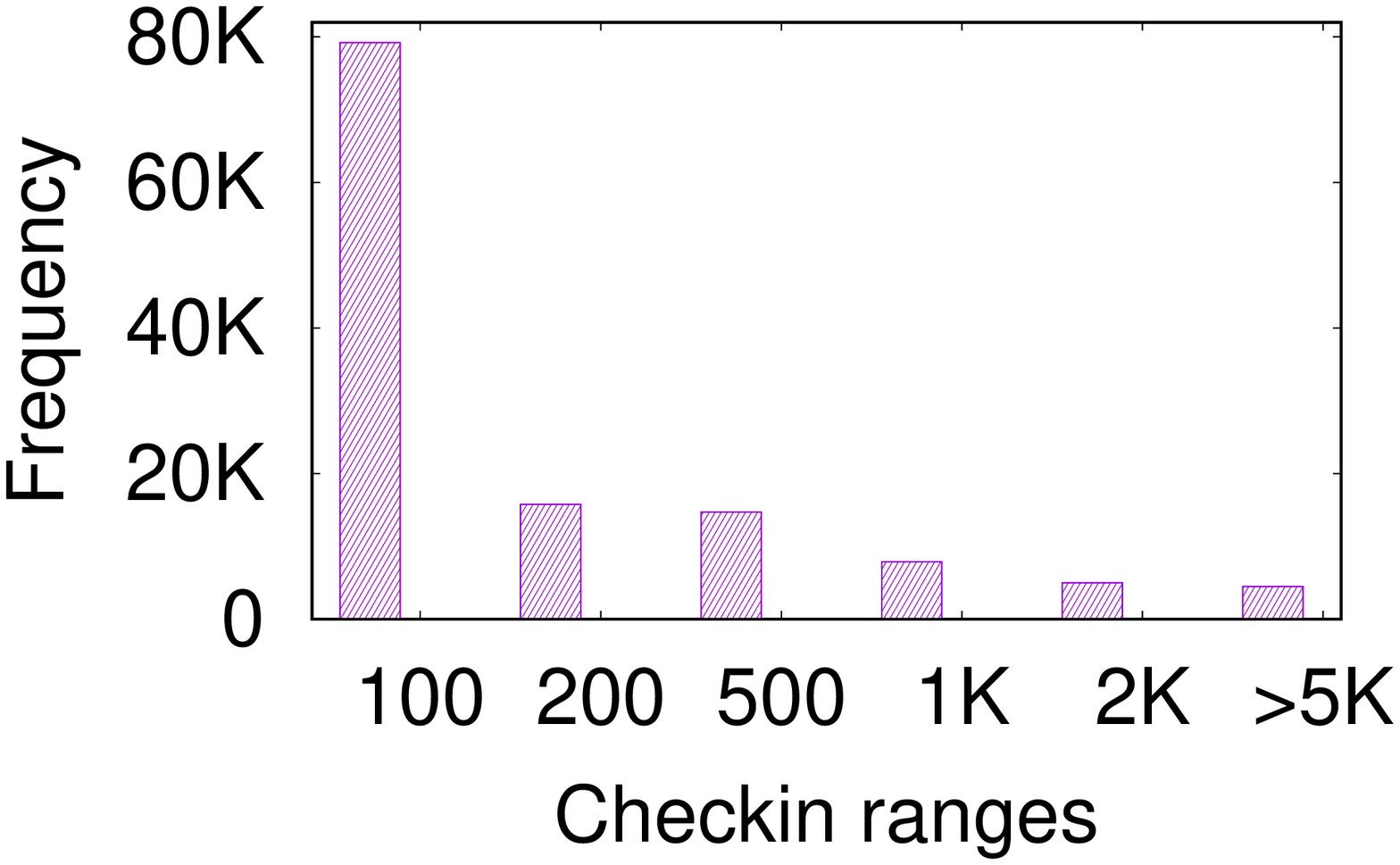}	
	}
	\subfigure[Yelp]{\label{checkinYL}
		\includegraphics[scale=0.19]{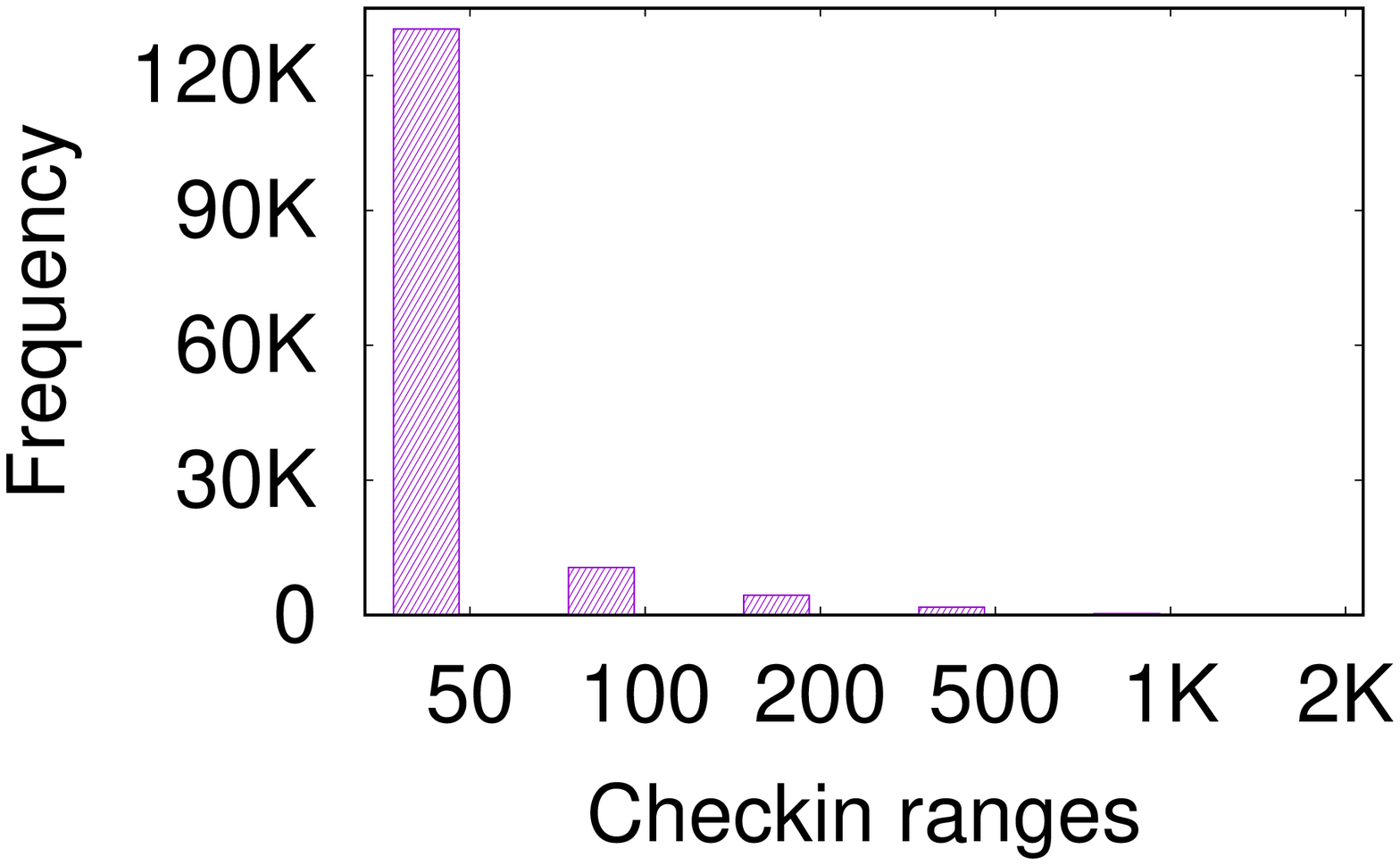}	
	}
	\caption{\jianxin{Characteristics of user check-ins }}
	\label{fig:checkinStats}
	\vspace{-3mm}
\end{figure}


\noindent \textbf{Evaluation Metrics.} 
\noindent $\underline{Precision.}$ It represents the percentage of the common elements (e.g., locations) between the result set returned by an approach and the exact results. 


\noindent $\underline{\textit{Mean of Minimum Diversity (MMD).}}$ 
Likewise, Minimization of the Mean of Shortest Distance (MMSD) \cite{36delmelle2014spatial, 37wang2012review}, we calculate the Mean of Minimum Diversity (MMD) for a query user $u$ w.r.t. neighbors' locations, i.e., $\textit{MMD}(u) = \frac{\sum_{v\in V_u}\min dist(L_v, S)}{|V_u|}$.
This metric shows how well the selected set of locations $S$ for user $u$ can cover her friends $v\in V_u$. Note, in a socio-spatial domain, $dist$ is considered as socio-spatial distance ($D_{ss}$) between two locations. 

\noindent $\underline{\textit{Social Coverage (SC).}}$ To measure the social quality of the selected set, we compute social coverage using the percentage of friends who have at least one check-in within $\theta$ kilometer (KM) from the selected set $S$, i.e., $\textit{SC}(u) = \frac{|v \in V_u \wedge dist(L_v, S) \le \theta |}{|V_u|}*100$.

\noindent $\underline{\textit{Social Entropy (SE).}}$ Given a selected set of locations $S$ of $u$, let $V_{u,l}$ be the set of $u$'s friends who visits $l\in S$. The social entropy of the set $S$ of $u$ is, $\textit{SE} = - \sum_{l \in S}p_l\log_2(p_l)$, where, $p_l = \frac{|V_{u,l}|}{\sum_{l_i \in S}|V_{u,l_i}|}$.
\textit{SE} measures the diversity of a location set w.r.t. the participation of its users across different other groups \cite{3shi2014density}. 
Here, for a selected location $l\in S\subset L_{u}$ of $u$, one of $l$'s corresponding group is considered as the friends who visited $l$ (e.g. $V_{u,l}$). 
A higher social entropy of a set suggests that the selected locations can cover more socially diverse friends.

\begin{table}[htbp]
	\vspace{-3mm}
	\centering
	\small
	\caption{Parameters and their values}
	\scalebox{0.95}{
		\begin{tabular}{|p{2.3cm}|p{3.8cm}|p{0.8cm}|}
			\hline 
			Parameter & Values & Default \\ 
			\hline 
			$\alpha, \omega$  & $(0, 1)$ & $0.5$\\ 
			\hline 
			$k$  & $2, 4, 6, 8, 10$ & $6$\\ 
			\hline
			Check-in group id & $50, 100, 200, 500, 1000$ & $100$\\ 
			\hline
		\end{tabular} 
	}
	\label{tab:param}
	\vspace{-2mm}
\end{table}

\noindent \textbf{Parameter Configuration.} 
\jianxin{Table~\ref{tab:param} presents the varied range of the parameters with default values. If there is no specific declaration, then the default values of the parameters will be used when one parameter varies.}
\jianxin{We only consider users having at least ten check-ins and at least two friends with check-in information. To see the effect of varying number of check-ins, we divide the users of each dataset into five groups based on the number of location check-ins they have. The group ids 50, 100, 200, 500, and 1000 contain the users with check-in locations in the range 10-50, 51-100, 101-200, 201-500, and 501-1000, respectively.} 

\subsection{Efficiency Evaluation}
\label{sec:Efficiency}

In this section, we compare the scalability of our proposed approaches. 

\subsubsection{\underline{Varying Answer Set Size, $k$}} 
Figure \ref{fig:VaryK} shows the average runtime of our proposed methods by varying $k$ between 2 to 10. The runtime of the algorithms follow similar trends, where \textit{E} consumes maximum time to process a query. 
On average, \textit{EP} is 2 to 3 times faster than \textit{AP}, and 3 to 6 times faster than \textit{E}. We notice, \textit{AP} performs efficiently than \textit{EP} for those users who have candidate locations with similar relevance scores, and higher diversity. 
Also, \textit{AP} is three times faster than \textit{E} and 
\textit{FA} is 9 to 15 times faster than \textit{EP} in different datasets when $k$ varies from 2 to 10. 
\begin{figure}[htbp]
	\vspace{-4mm}
	\centering
	\subfigure[Gowalla]{\label{efficiencyKG}
		\includegraphics[scale=0.22]{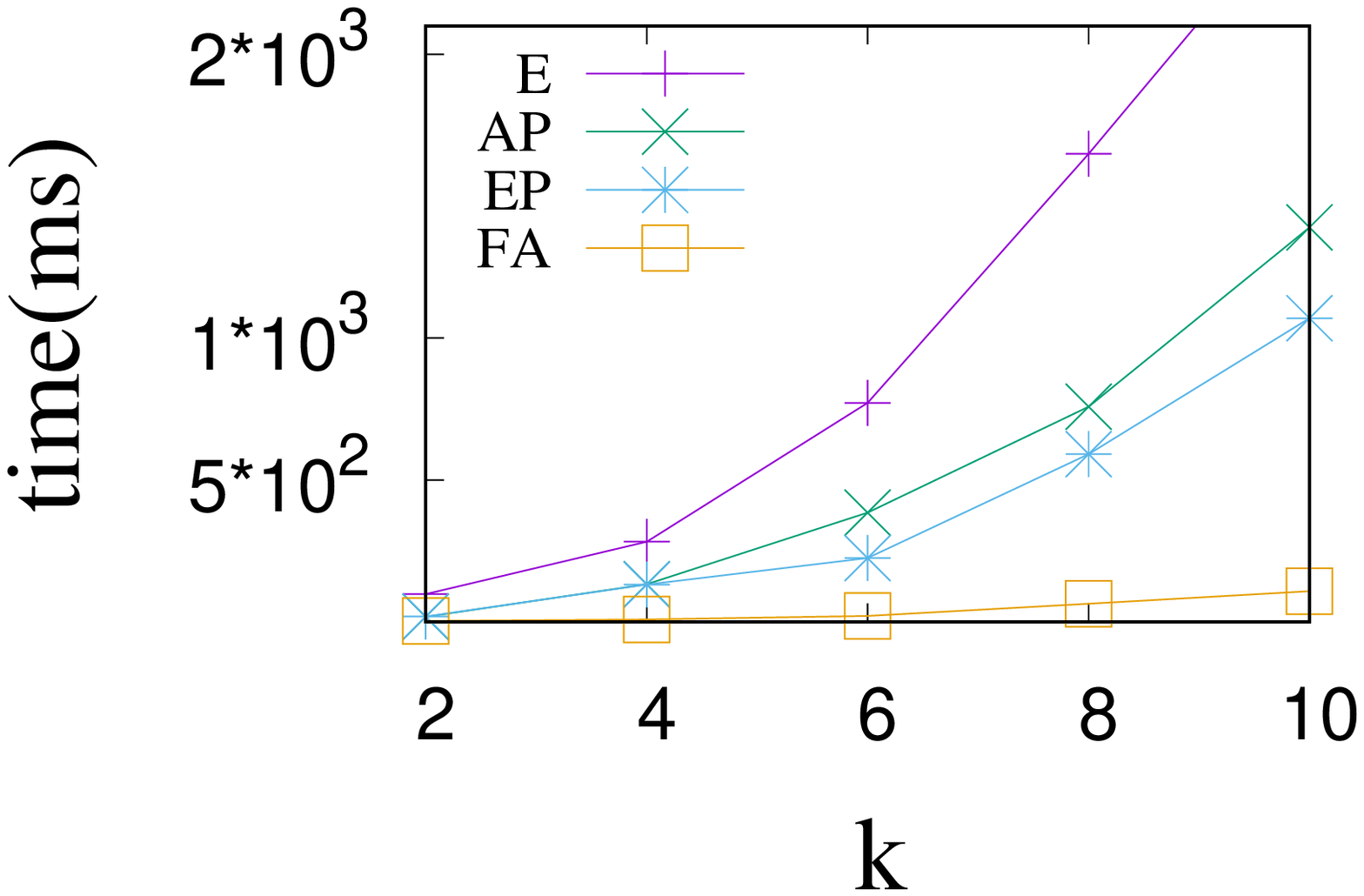}
	}
	\subfigure[Brightkite]{\label{efficiencyKB}
		\includegraphics[scale=0.22]{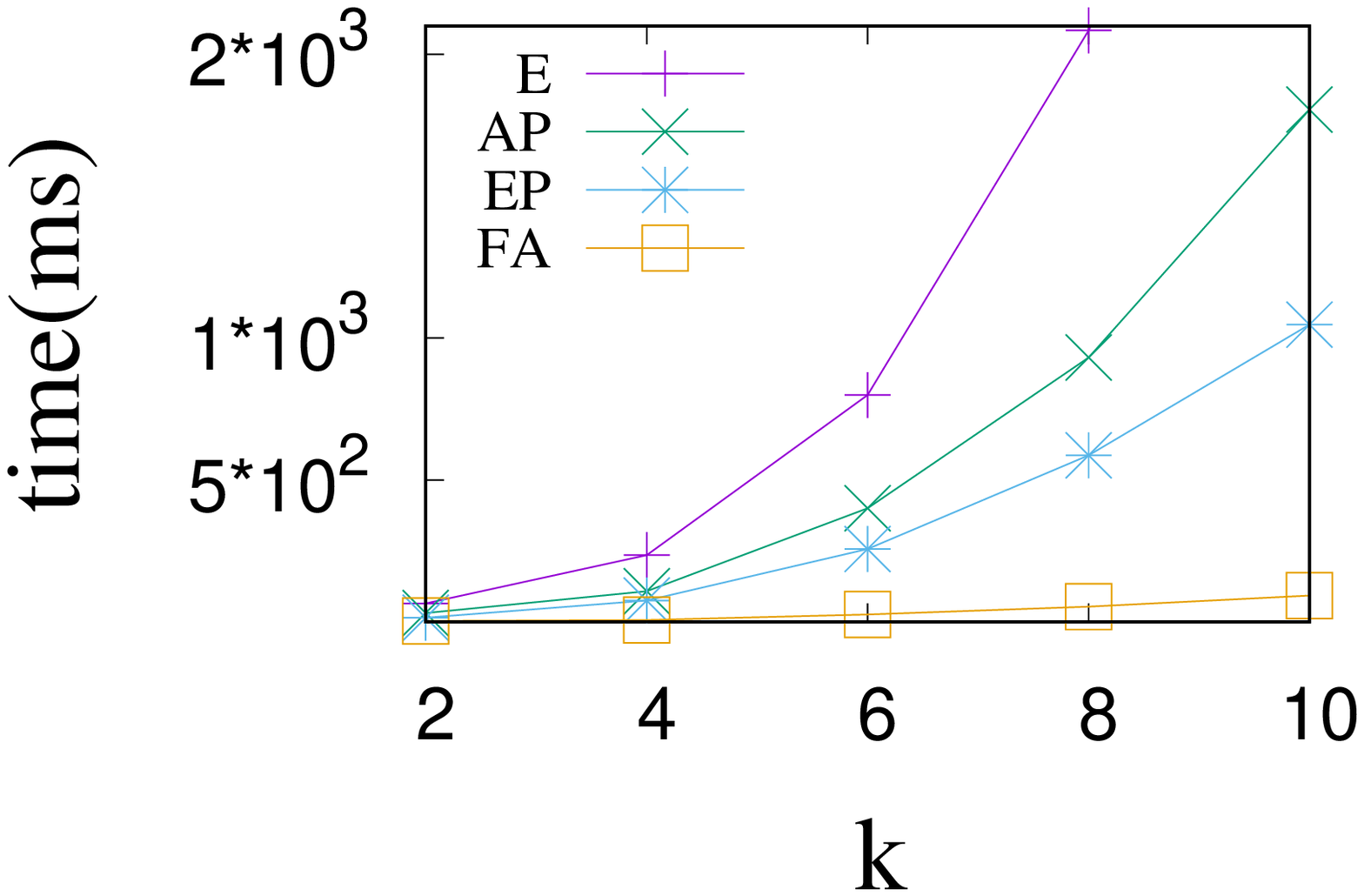}	
	}\vspace{-3.5mm}	
	\subfigure[Flickr]{\label{efficiencyFL}
		\includegraphics[scale=0.22]{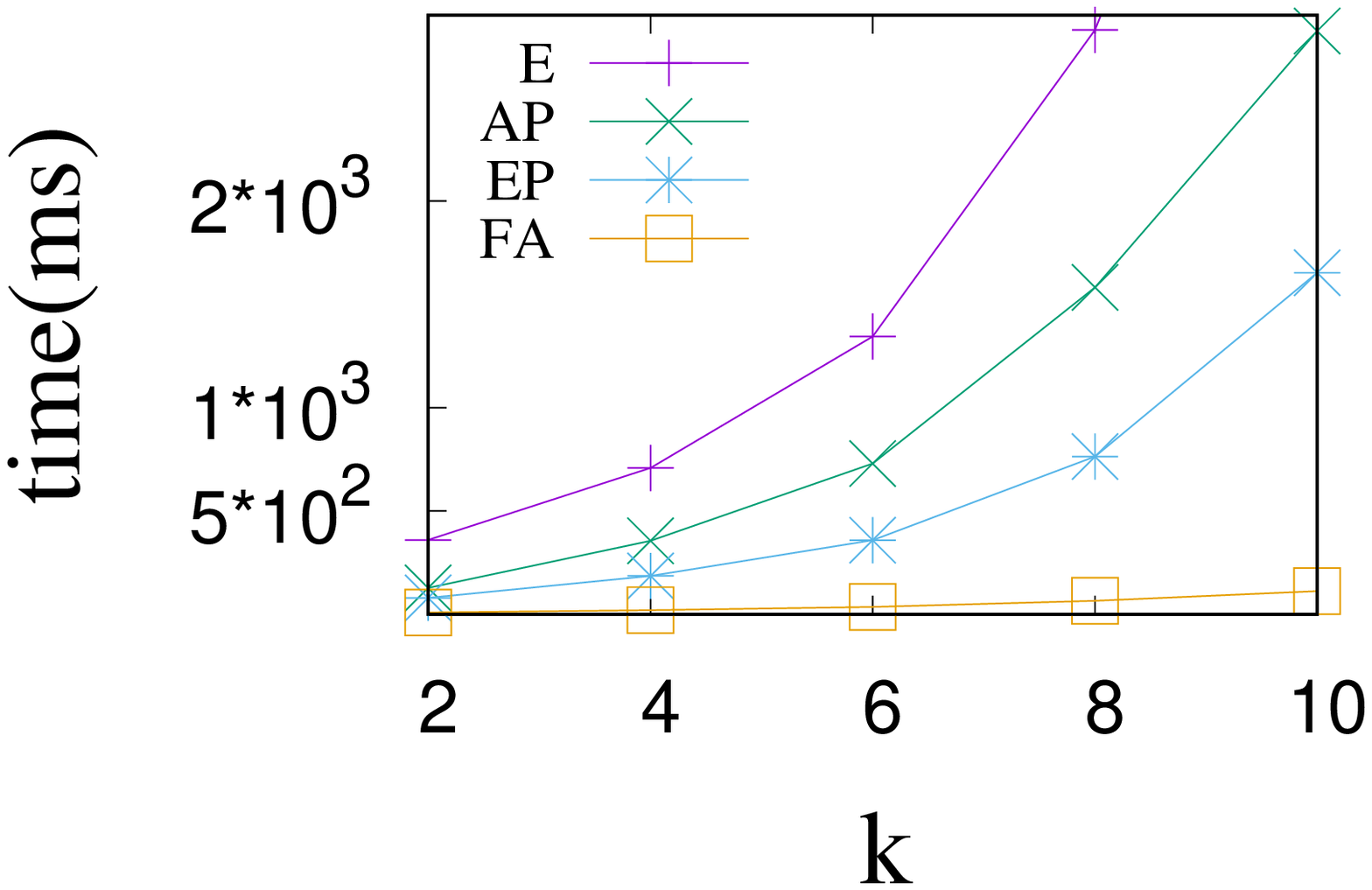}
	}
	\subfigure[Yelp]{\label{efficiencyYL}
		\includegraphics[scale=0.22]{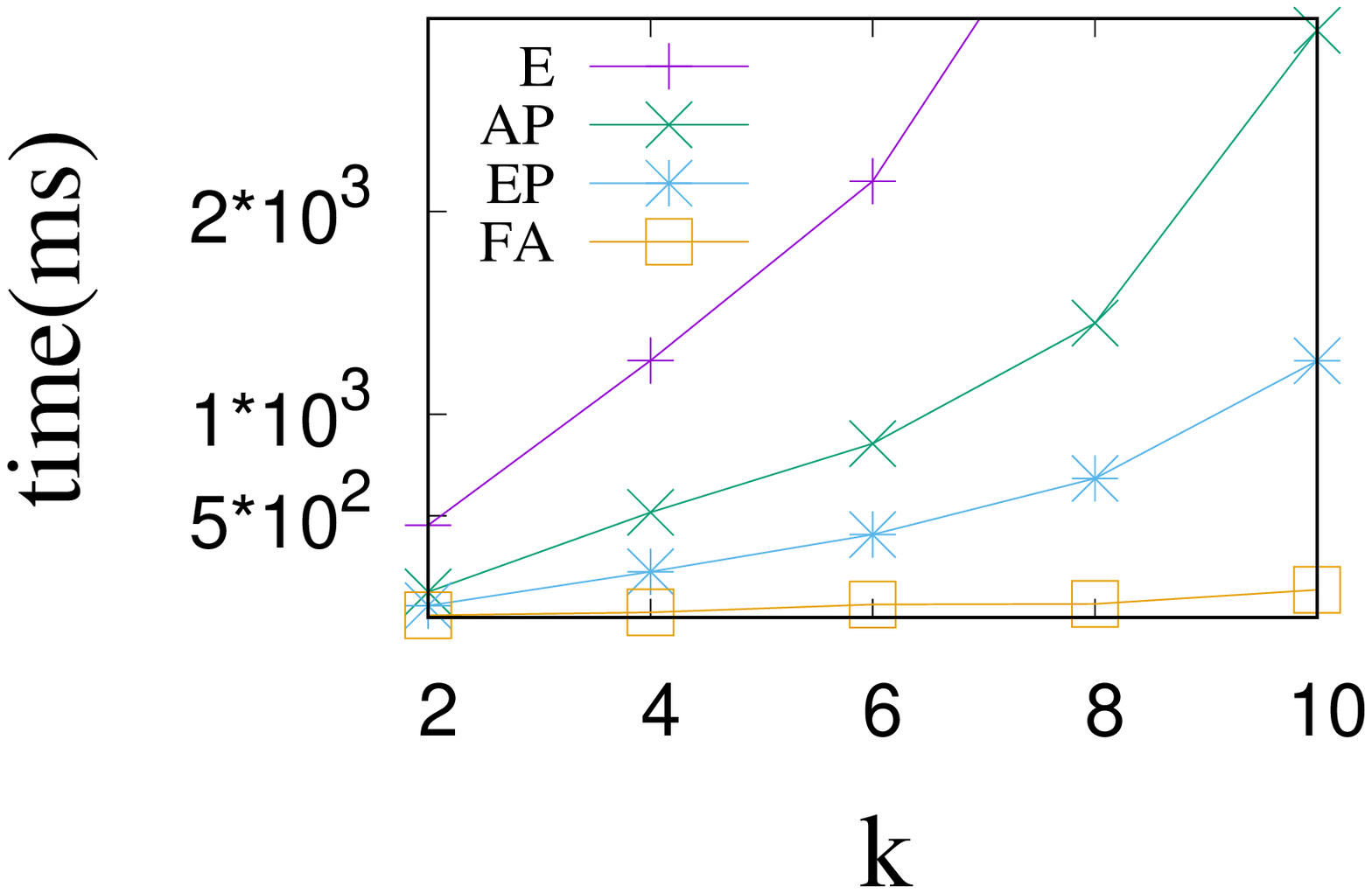}	
	}
	\vspace{-3mm}	
	\caption{Varying $k$}
	\label{fig:VaryK}
	\vspace{-3mm}
\end{figure}


\subsubsection{\underline{Varying Check-in Group Size}} 
\jianxin{In this experiment, we study the performance of the proposed approaches on the distribution of number of check-ins.} Specifically, we show how the size of the check-in locations (e.g., candidate set) of users affects the performance on the approaches. 
From\ignore{the results shown in} Figure \ref{fig:varyBins_k}, we find the runtime of the proposed approaches, except \textit{FA}, increases fast with the check-in group size. This is because a considerable amount of possible groups of locations are needed to compare in \textit{E}, \textit{AP}, and \textit{EP} when the check-in group size is large. We also notice \textit{EP} is much efficient than \textit{E} and \textit{AP}. For example, in check-in group 500 in \textit{Brightkite}, \textit{EP} reports 2.5 and 4.7 times faster than \textit{AP} and \textit{E}, respectively. \textit{FA} performs significantly efficient, even for the large candidate set. For example, in \textit{Gowalla}, \textit{FA} is 57 times faster than EP when the check-in group id is 1000.



\begin{figure}[htbp]
	\vspace{-4mm}
	\centering
	\subfigure[Gowalla]{\label{varyBins_k_GW}
		\includegraphics[scale=0.22]{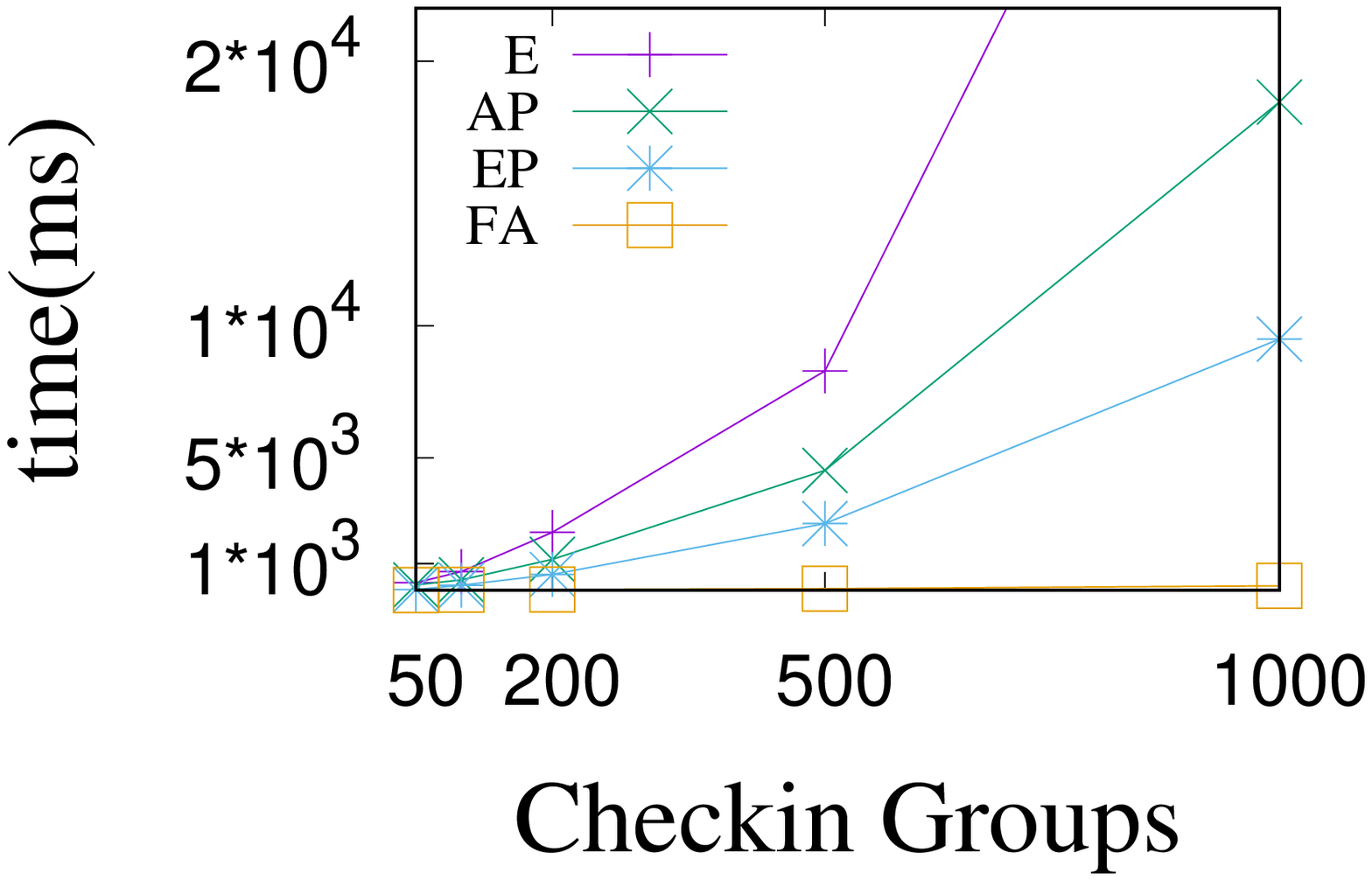}
	}
	\subfigure[Brightkite]{\label{varyBins_k_BK}
		\includegraphics[scale=0.22]{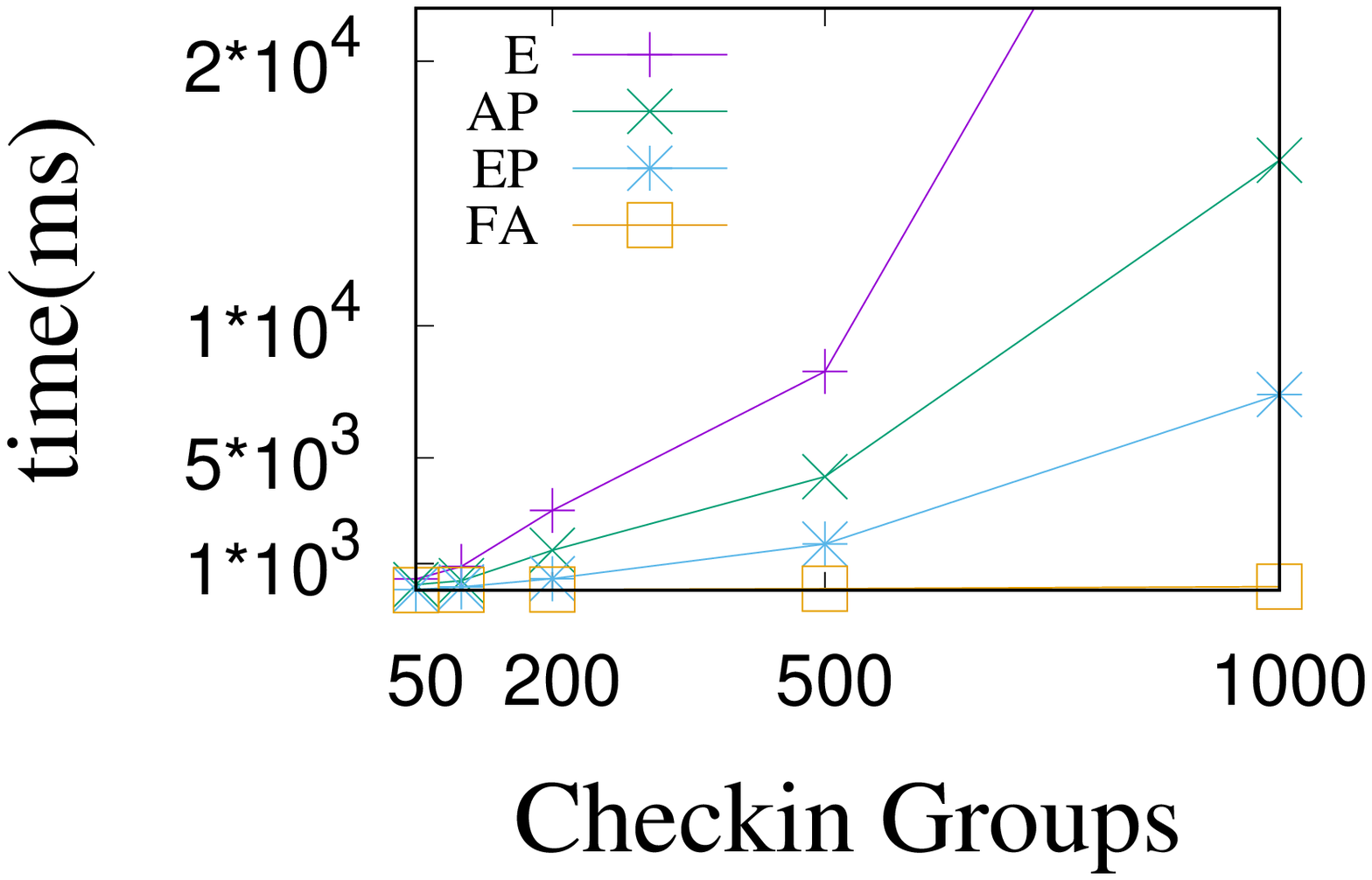}	
	}\vspace{-3mm}	
	\subfigure[Flickr]{\label{varyBins_k_FL}
		\includegraphics[scale=0.22]{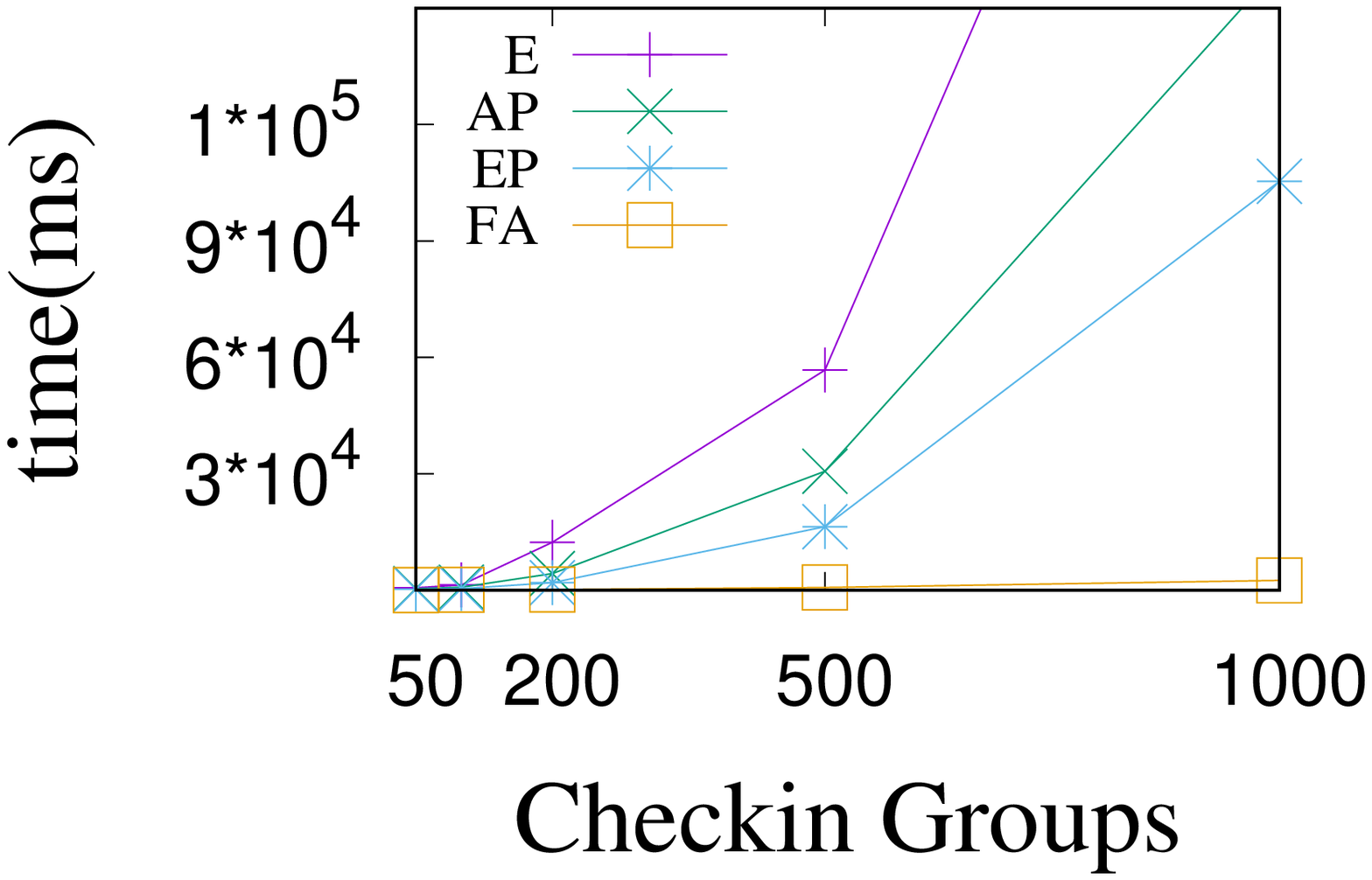}
	}
	\subfigure[Yelp]{\label{varyBins_k_YL}
		\includegraphics[scale=0.22]{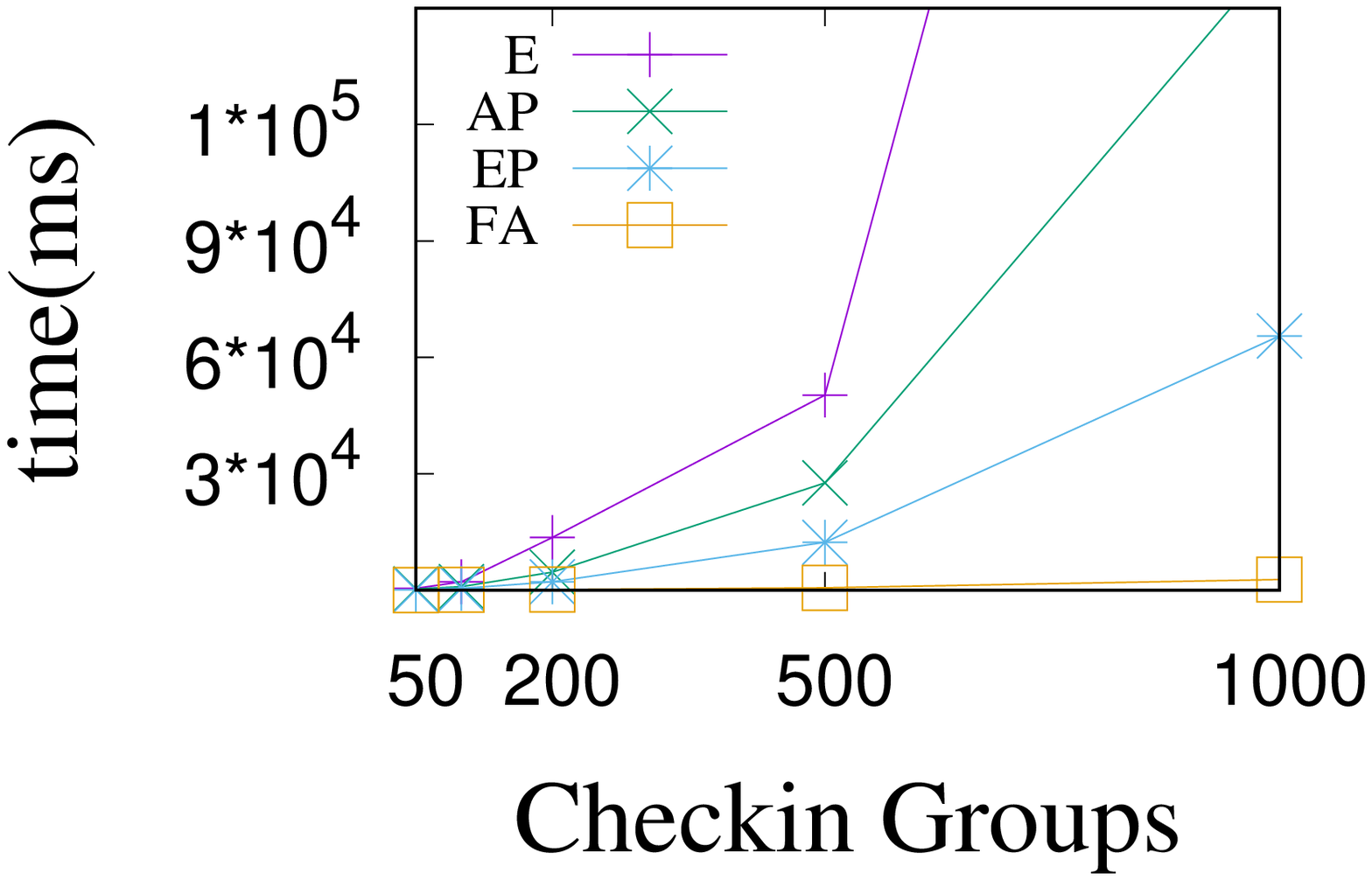}	
	}
	\vspace{-3mm}	
	\caption{\jianxin{Varying check-in groups}}
	\label{fig:varyBins_k}
	\vspace{-3mm}
\end{figure}

\subsubsection{\underline{Varying Number of Friends}} \jianxin{We compare the efficiency of our proposed algorithms by varying the number of social connections users have. To balance user count with sufficient check-ins, we first select the check-in group 500, then divide the users into five groups with medium to a higher number of social connections. The user group ids 100, 200, 500, 1000 contain the users with 50-100, 101-200, 201-501, 501-1000 friends, respectively. 
In Figure~\ref{fig:varyNodeDegreeBins}, 
	we notice a similar trend in the proposed methods, where a higher number of friends do not affect the efficiency. This is because, the social and spatial relevance scores are pre-computed using the location information of the friends. 
	The proposed algorithms only depend on the number of check-ins a query user has. Therefore, it merely gets affected by the number of social connections. 
}

\begin{figure}[htbp]
	\vspace{-4mm}
	\centering
	\subfigure[Gowalla]{\label{varyNodeDegree_GW}
		\includegraphics[scale=0.22]{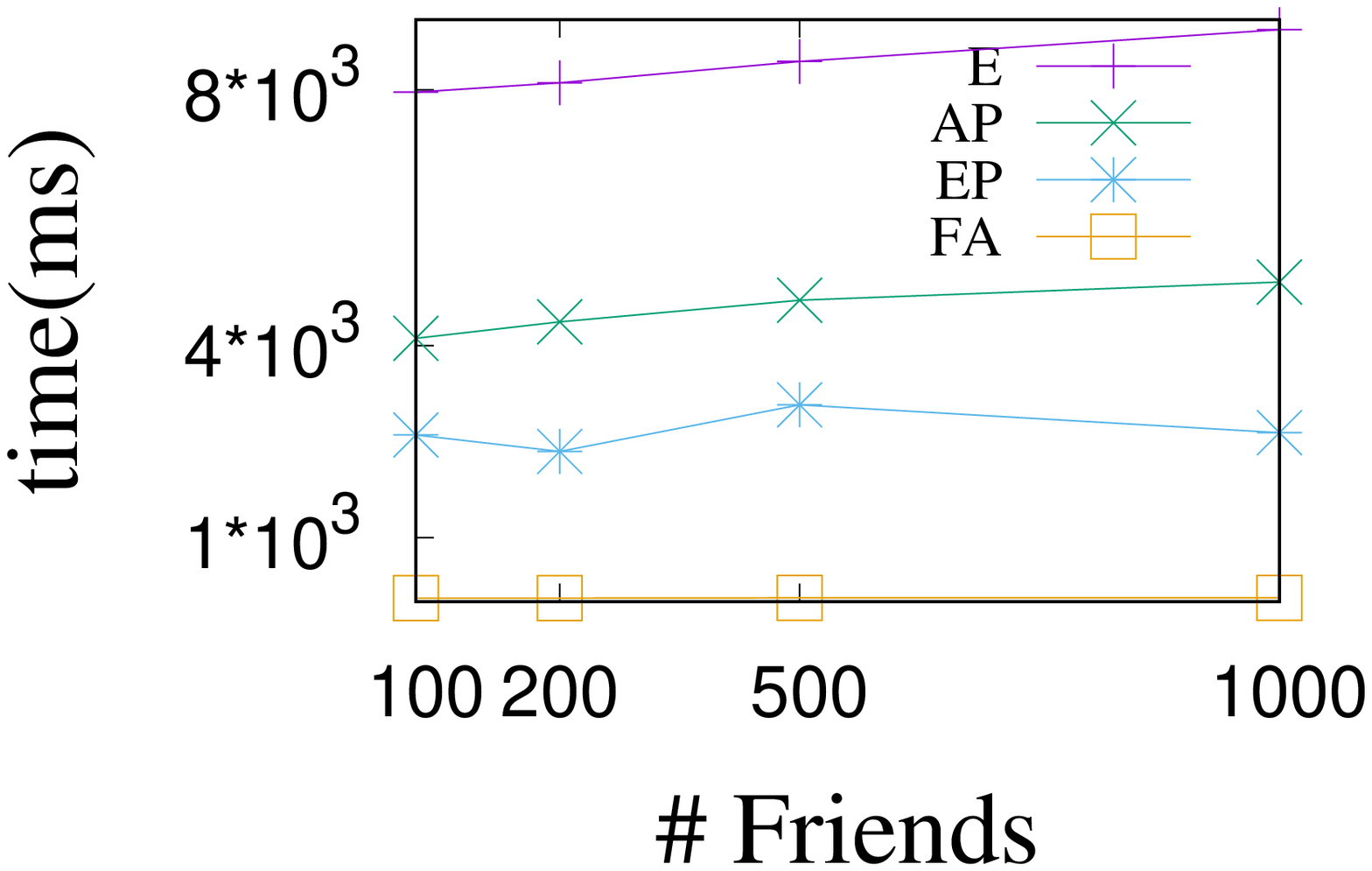}
	}
	\subfigure[Brightkite]{\label{varyNodeDegree_BK}
		\includegraphics[scale=0.22]{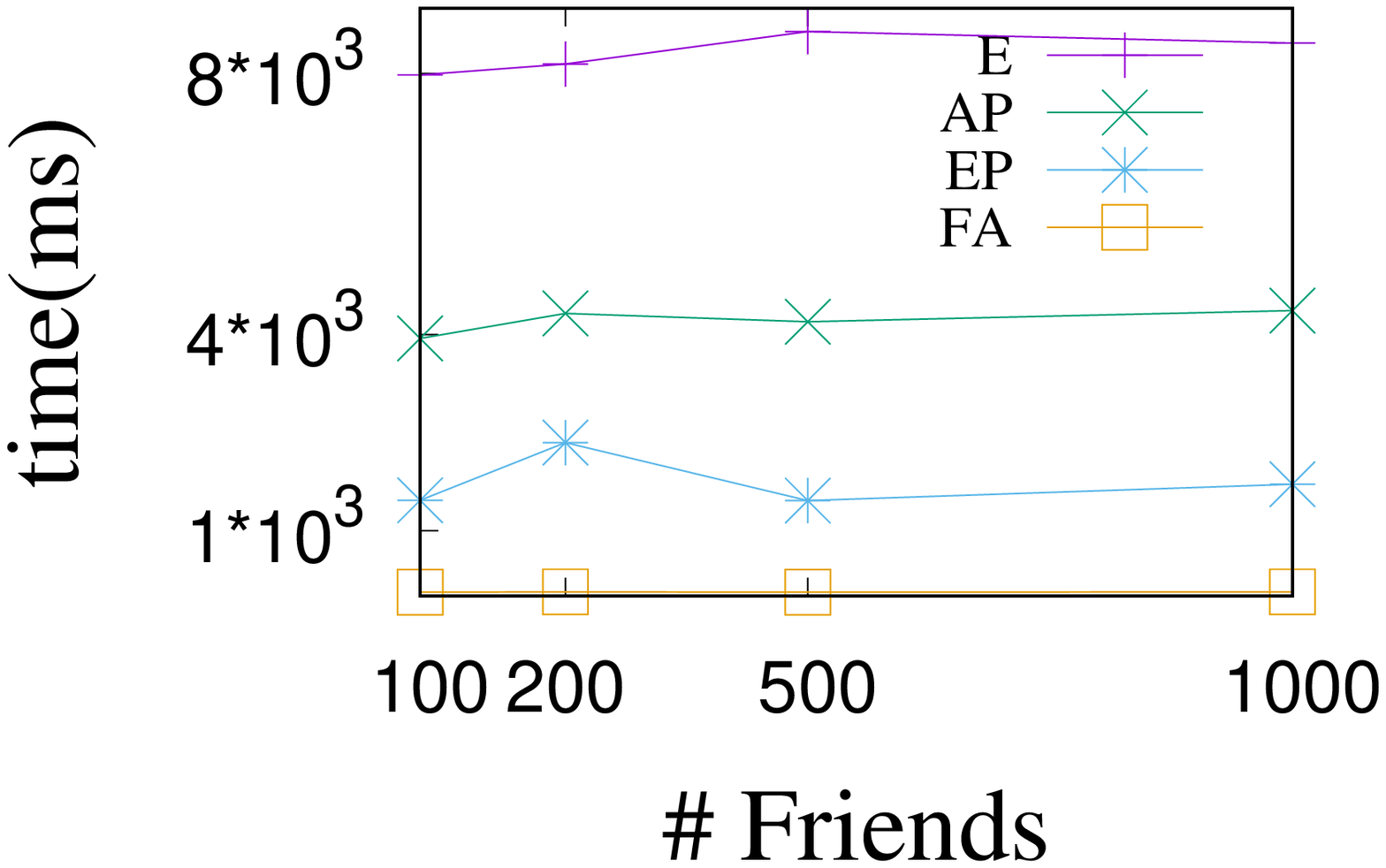}	
	}\vspace{-3mm}	
	\subfigure[Flickr]{\label{varyNodeDegree_FL}
		\includegraphics[scale=0.22]{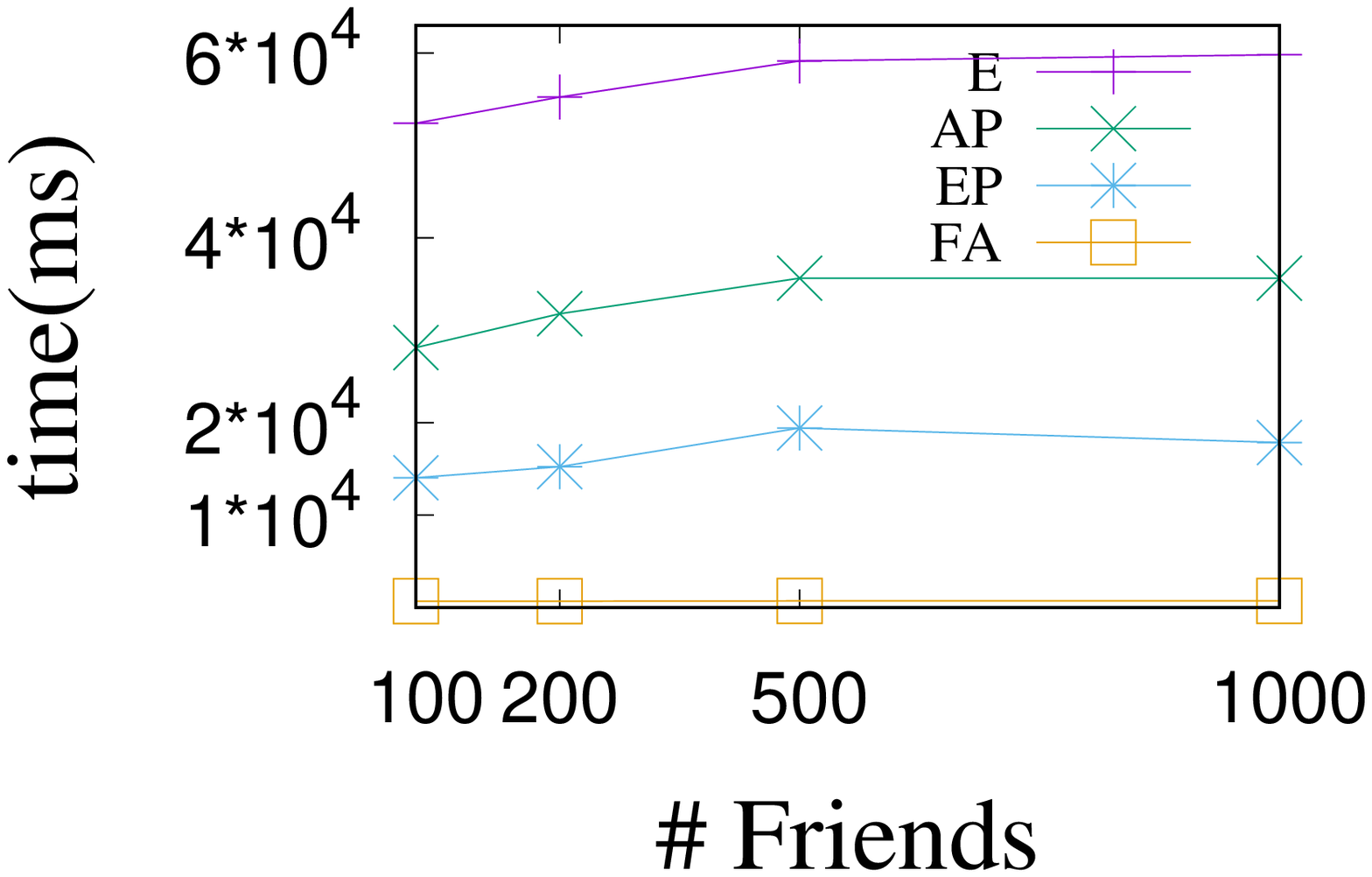}
	}
	\subfigure[Yelp]{\label{varyNodeDegree_YL}
		\includegraphics[scale=0.22]{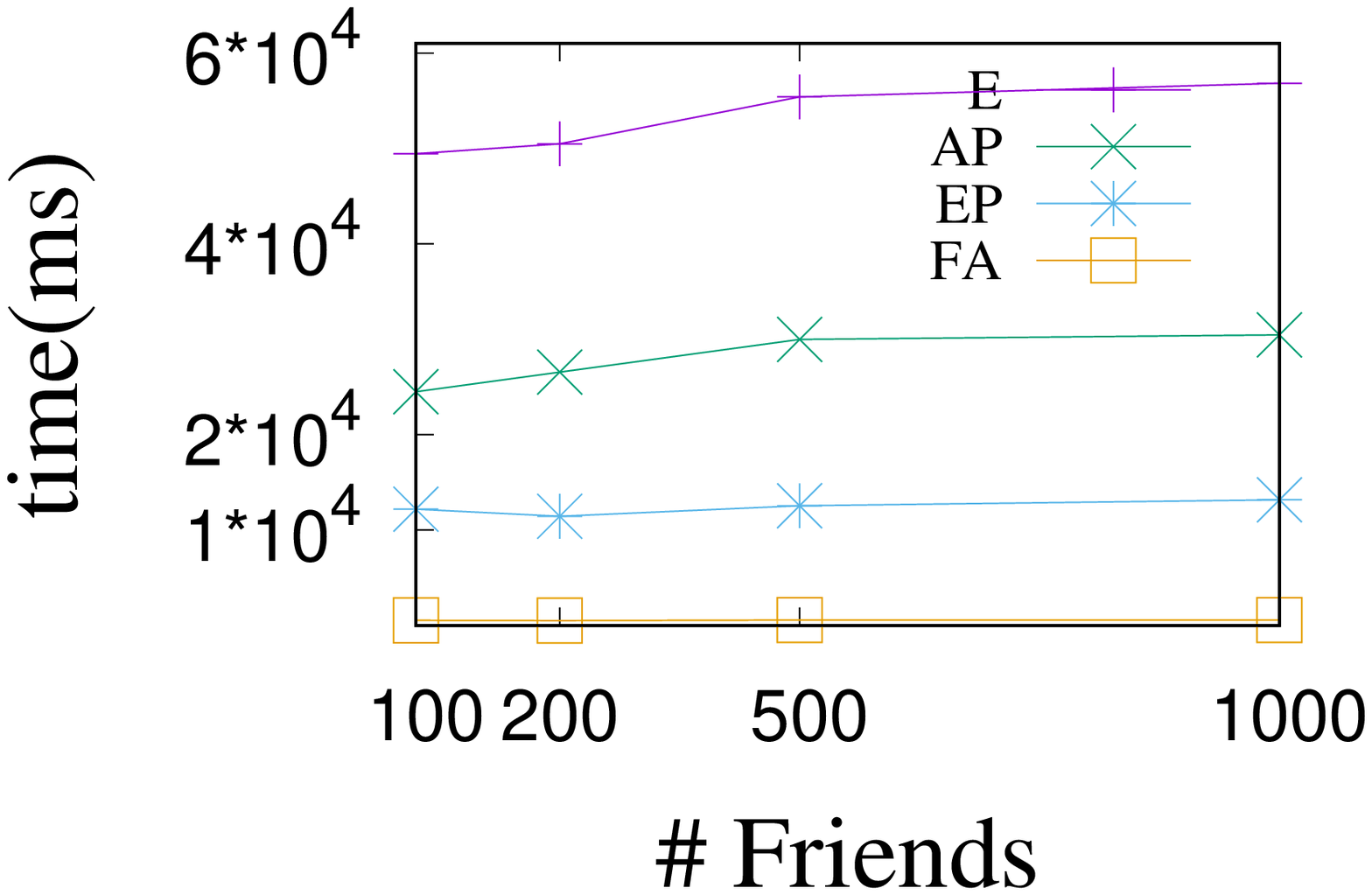}	
	}
	\vspace{-3mm}	
	\caption{\jianxin{Varying Number of Friends}}
	\label{fig:varyNodeDegreeBins}
\end{figure}

\subsubsection{\jianxin{\underline{Varying $\alpha$, and $\omega$}}}
We also test our proposed algorithms by varying the trade-off parameters $\alpha$, $\omega$. Figure~\ref{fig:varyingParams_Efficiency} shows the average runtime of our proposed algorithms in \textit{Gowalla} and \textit{Yelp} datasets when the trade-off parameters $\alpha$, $\omega$ vary from $0.1$ to $0.9$. As expected, we do not observe any noticeable change in the efficiency trends in each datasets, where the average execution time of individual algorithm almost remains constant. 
This is because, these trade-off parameters do not interfere on how a method operates, but only precepts in selecting locations in the result set. 

\begin{figure}[htbp]
	\vspace{-4mm}	
	\centering
	\subfigure[Gowalla]{\label{varyOmegaAplgaGW}
		\includegraphics[scale=0.22]{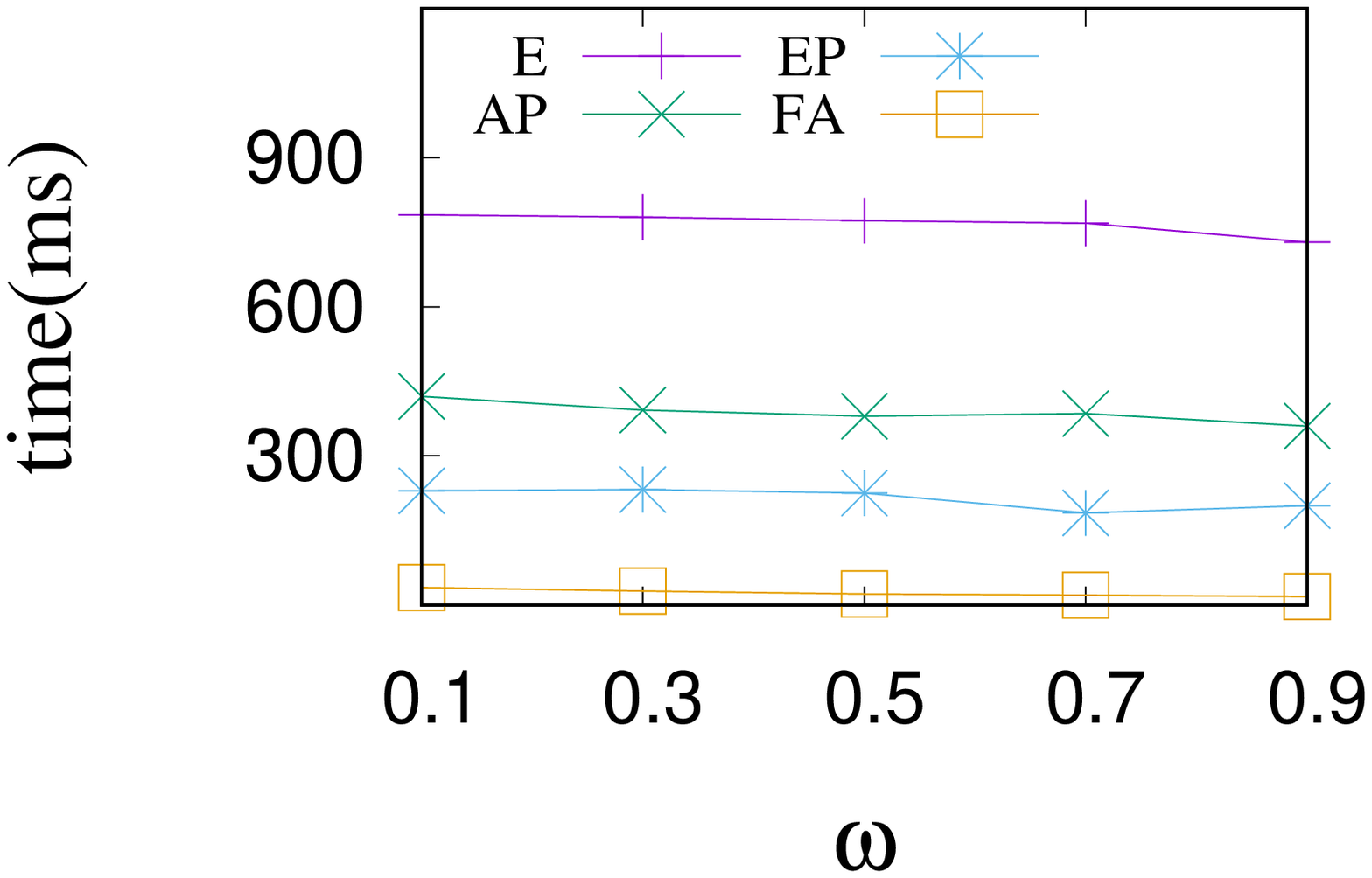}		
		\includegraphics[scale=0.22]{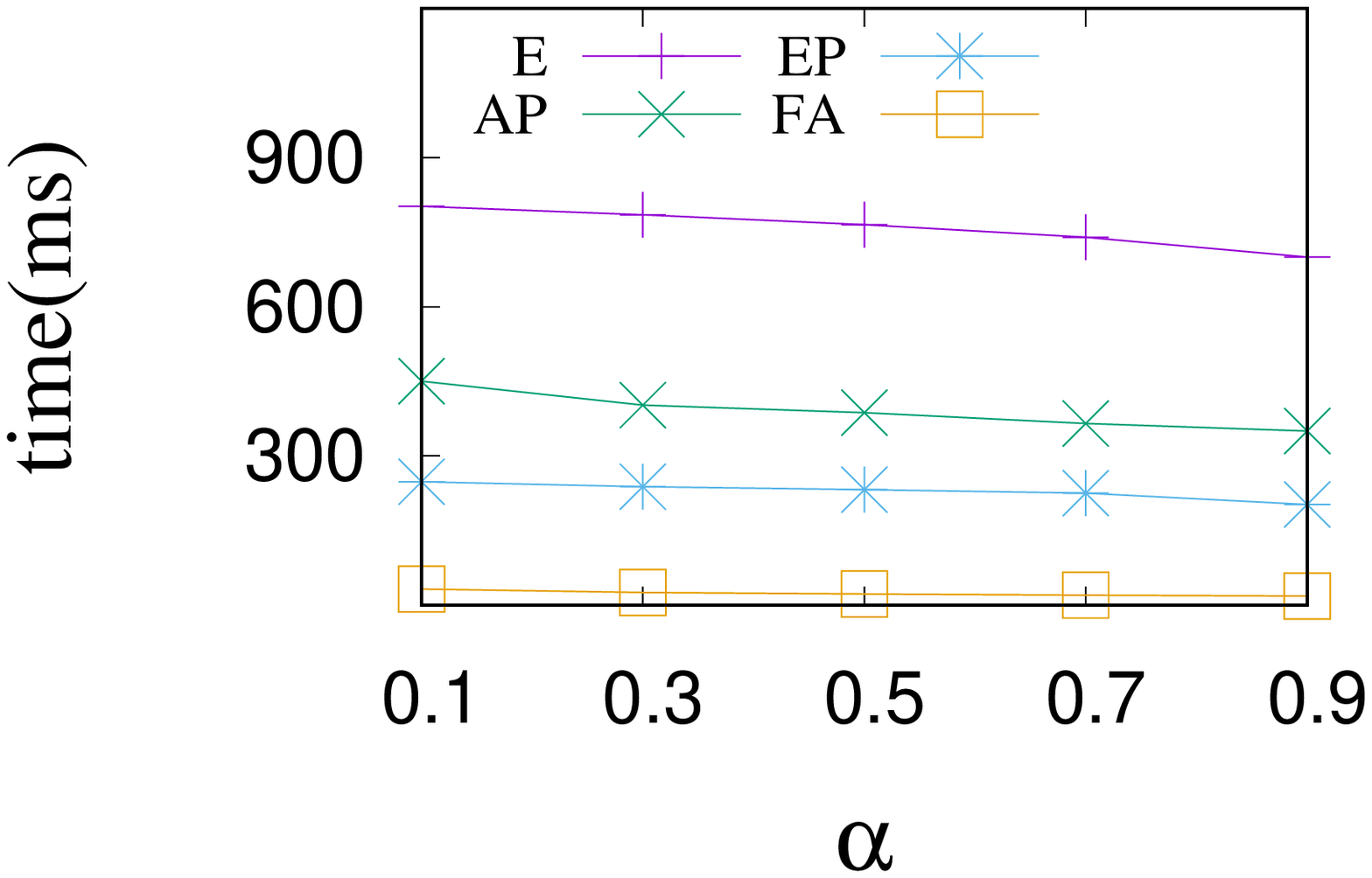}
	}\vspace{-3mm}
	\subfigure[Yelp]{\label{varyOmegaAlphaYL}
		\includegraphics[scale=0.22]{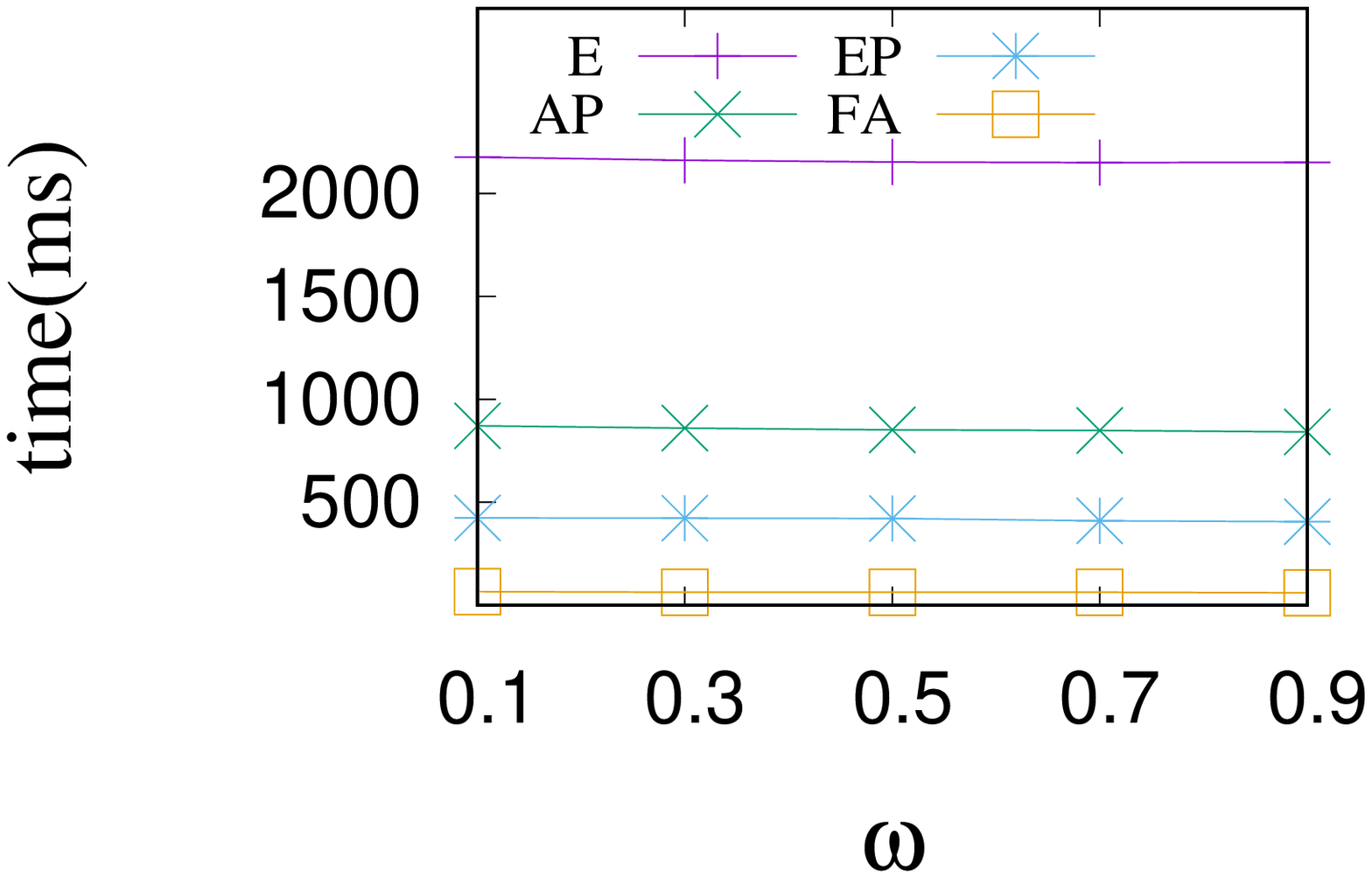}		
		\includegraphics[scale=0.22]{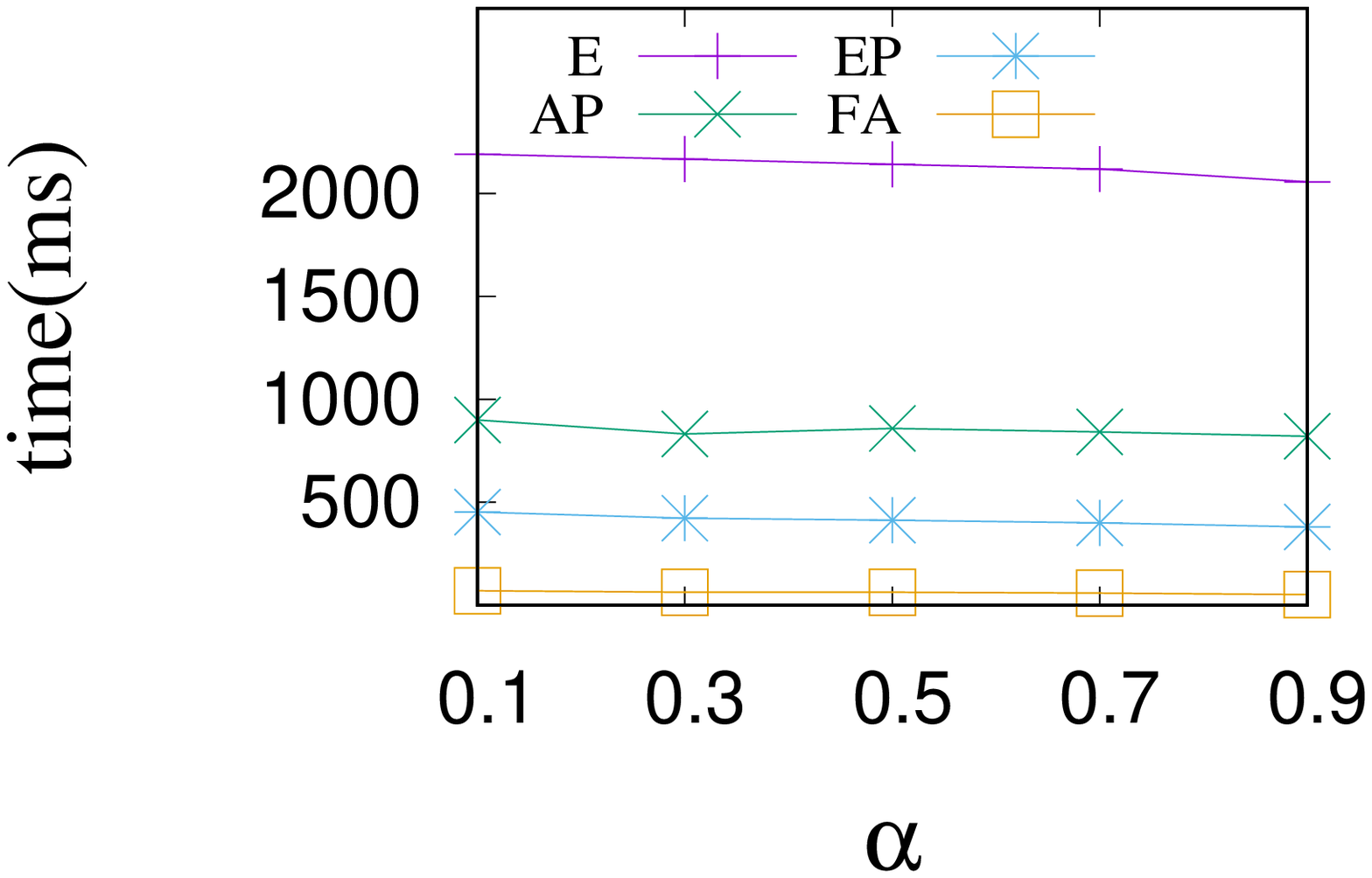}
	}\vspace{-3mm}	
	\caption{\jianxin{Varying $\omega$, $\alpha$}}
	\label{fig:varyingParams_Efficiency}
	\vspace{-3mm}	
\end{figure}

\subsection{Comparison with Existing Models}
We compare the performance of the existing greedy solutions, e.g., \textit{GMC}, \textit{AS}, \textit{GNE}, with our proposed approaches. For brevity of the presentation, we only show the results using the medium-sized dataset \textit{Gowalla} and the large dataset \textit{Yelp}. 

\subsubsection{\underline{Efficiency}}
To make a fair comparison between the greedy based existing works and our proposed solutions, we consider our top two efficient algorithms, \textit{EP} and \textit{FA}, in this experiment. 
Figure \ref{fig:varyK_baseline} depicts the runtime of the approaches by varying the answer set size $k$ in default check-in group. In \textit{Gowalla} dataset, \textit{GNE} has higher efficiency than \textit{EP}, but in \textit{Yelp}, it shows an opposite trend. This is because the candidate locations in check-in group 100 of \textit{Yelp} is higher than \textit{Gowalla}. 
\textit{GNE} always performs slower than \textit{FA}; e.g., in \textit{Yelp}, \textit{GNE} is two times slower than \textit{FA}. In each dataset with moderate-sized candidate locations, \textit{GMC} performs faster than the others. 

\begin{figure}[htbp]
	\vspace{-4mm}
	\centering
	\subfigure[Gowalla]{\label{varyK_GW_baseline}
		\includegraphics[scale=0.22]{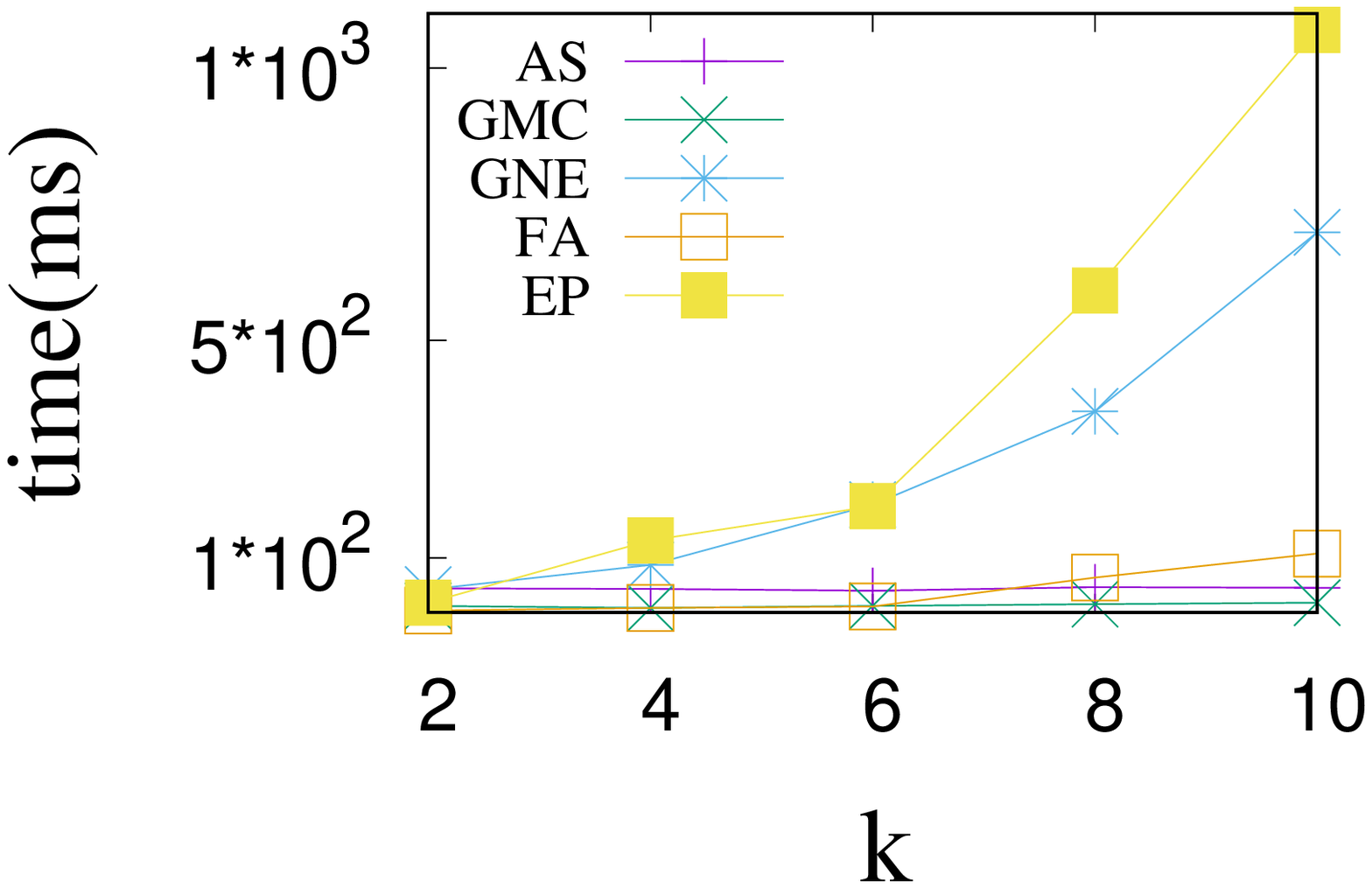}
	}
	\subfigure[Yelp]{\label{varyK_YL_baseline}
		\includegraphics[scale=0.22]{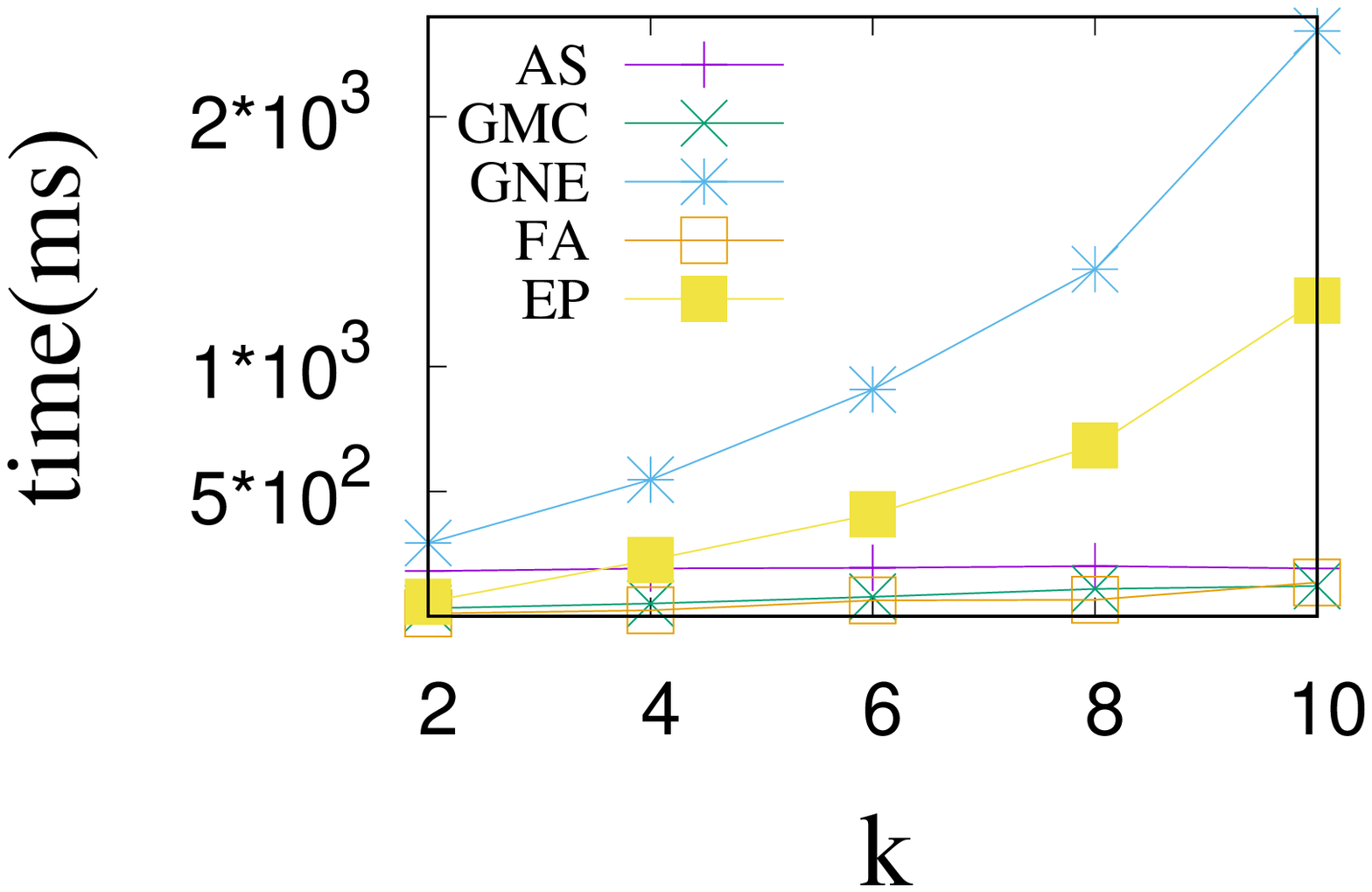}	
	}
	\vspace{-3mm}		
	\caption{Varying $k$}
	\label{fig:varyK_baseline}
	\vspace{-3mm}
\end{figure}

Figure \ref{fig:varyBins_k_baseline} compares the runtime of the approaches when check-in group size varies. We notice that \textit{FA} is faster than \textit{GMC} when the check-in group size is more than 100. This is because, \textit{GMC} needs more time to calculate the marginal contribution of the locations in large candidate sets. In higher check-in groups, \textit{GNE} takes considerable time to swap the locations in the current result set and the most diverse element among the remaining locations which results a lower efficiency.
\begin{figure}[htbp]
	\vspace{-4mm}
	\centering
	\subfigure[Gowalla]{\label{varyBins_k_GW_baseline}
		\includegraphics[scale=0.22]{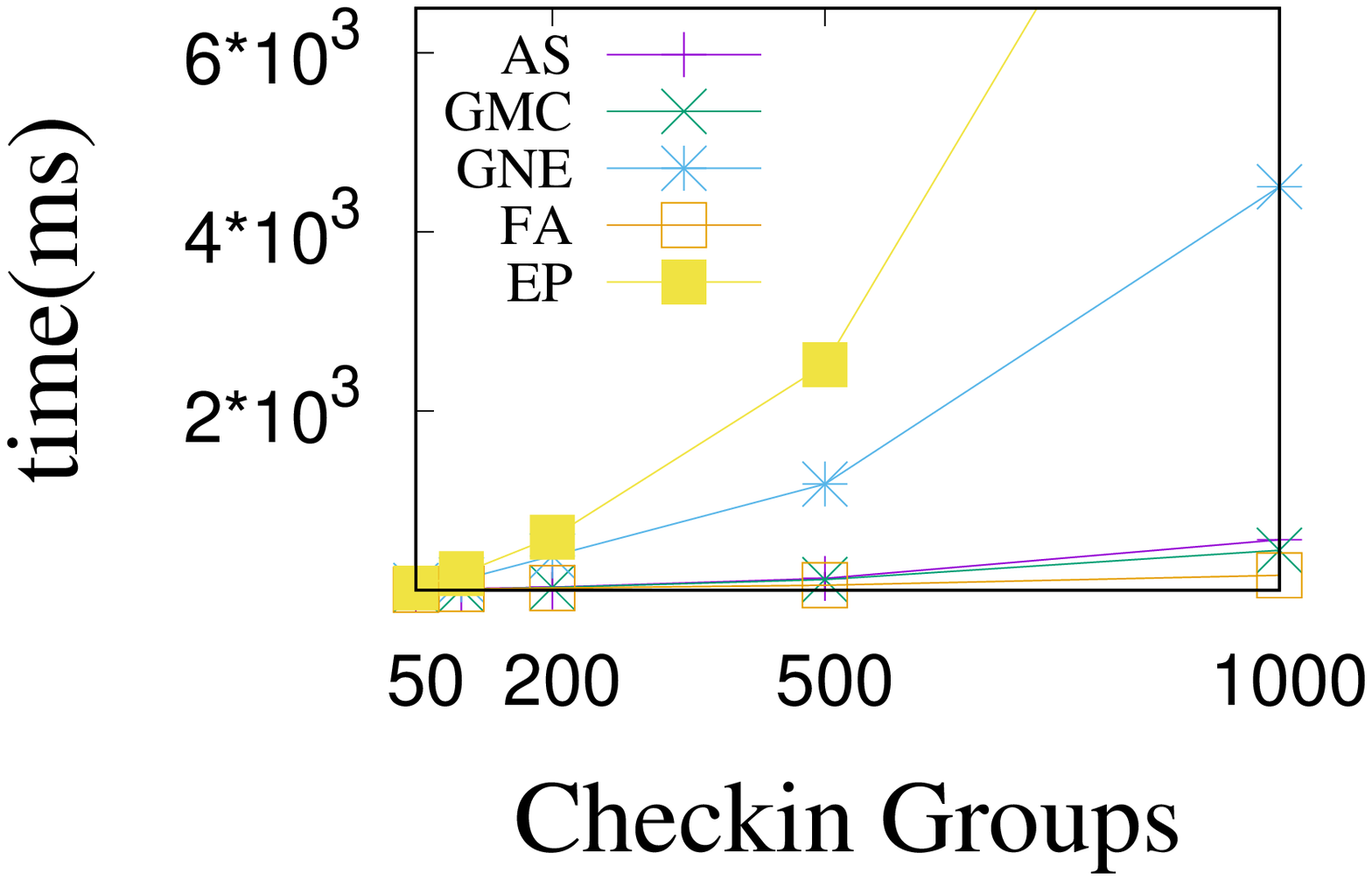}
	}
	\subfigure[Yelp]{\label{varyBins_k_YL_baseline}
		\includegraphics[scale=0.22]{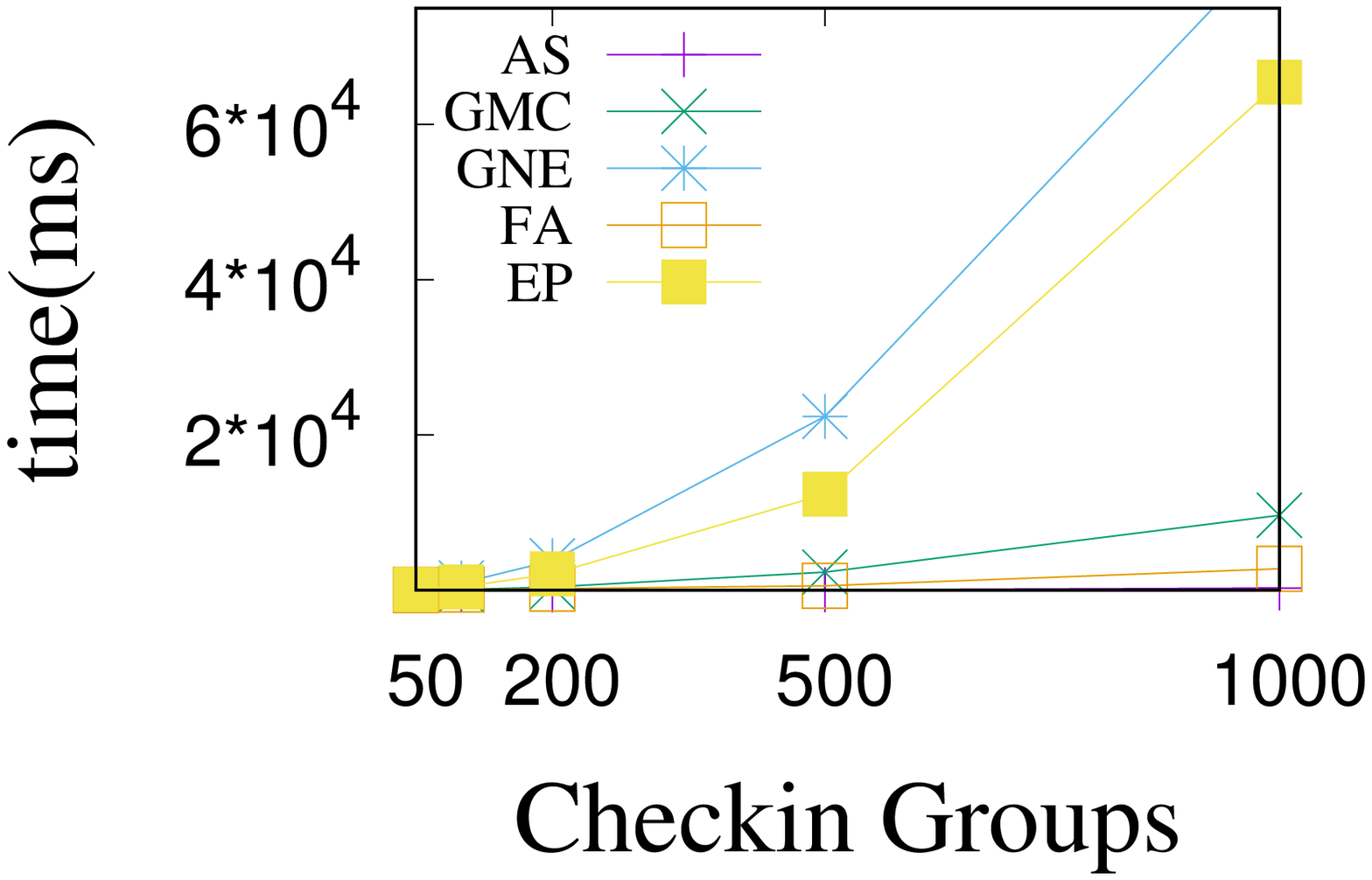}	
	}
	\vspace{-3mm}		
	\caption{\jianxin{Varying check-in groups}}
	\label{fig:varyBins_k_baseline}
	\vspace{-3mm}
\end{figure}

\subsubsection{\underline{Accuracy}} 
\label{sec:Effectiveness}
Figure \ref{fig:precLocNew} demonstrates the precision of the approaches w.r.t. the exact result when $k$ is varied. \textit{AP} has higher precision than the other approaches in each dataset. 
Although the precision of \textit{FA} is lower than \textit{AP}, \textit{FA} is much efficient (e.g., 10-25 times faster, see Figure \ref{fig:VaryK}). For example, in \textit{Yelp}, \textit{FA}'s precision is lower than \textit{AP} by 16\% only, but its efficiency outperforms \textit{AP} by about 20 times when $k=6$. 
\begin{figure}[htbp]
	\vspace{-5mm}
	\centering
	\subfigure[Gowalla]{\label{precLocGWNew}
		\includegraphics[scale=0.20]{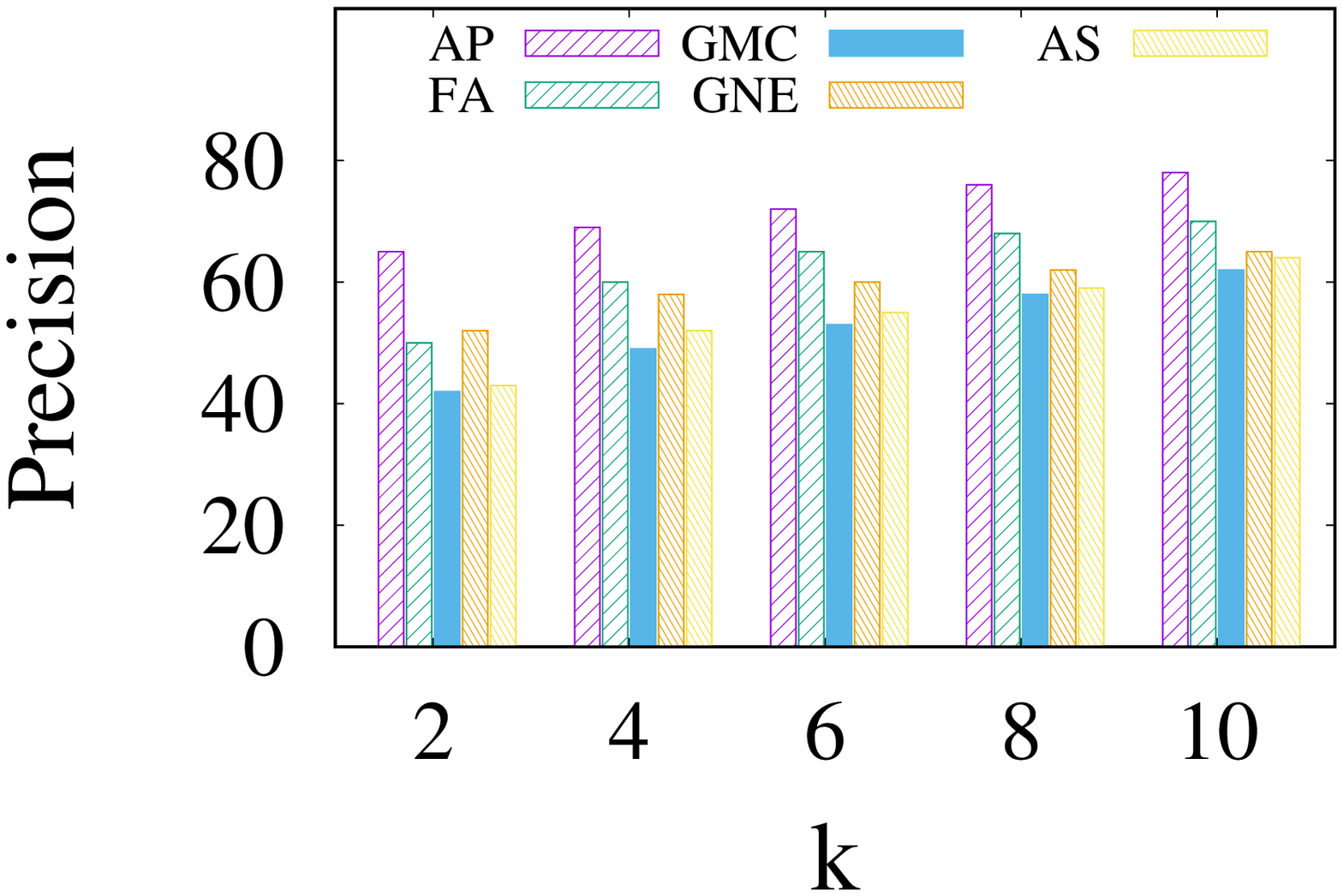}
	}
	\vspace{-3mm}
	\subfigure[Yelp]{\label{precLocNew}
		\includegraphics[scale=0.20]{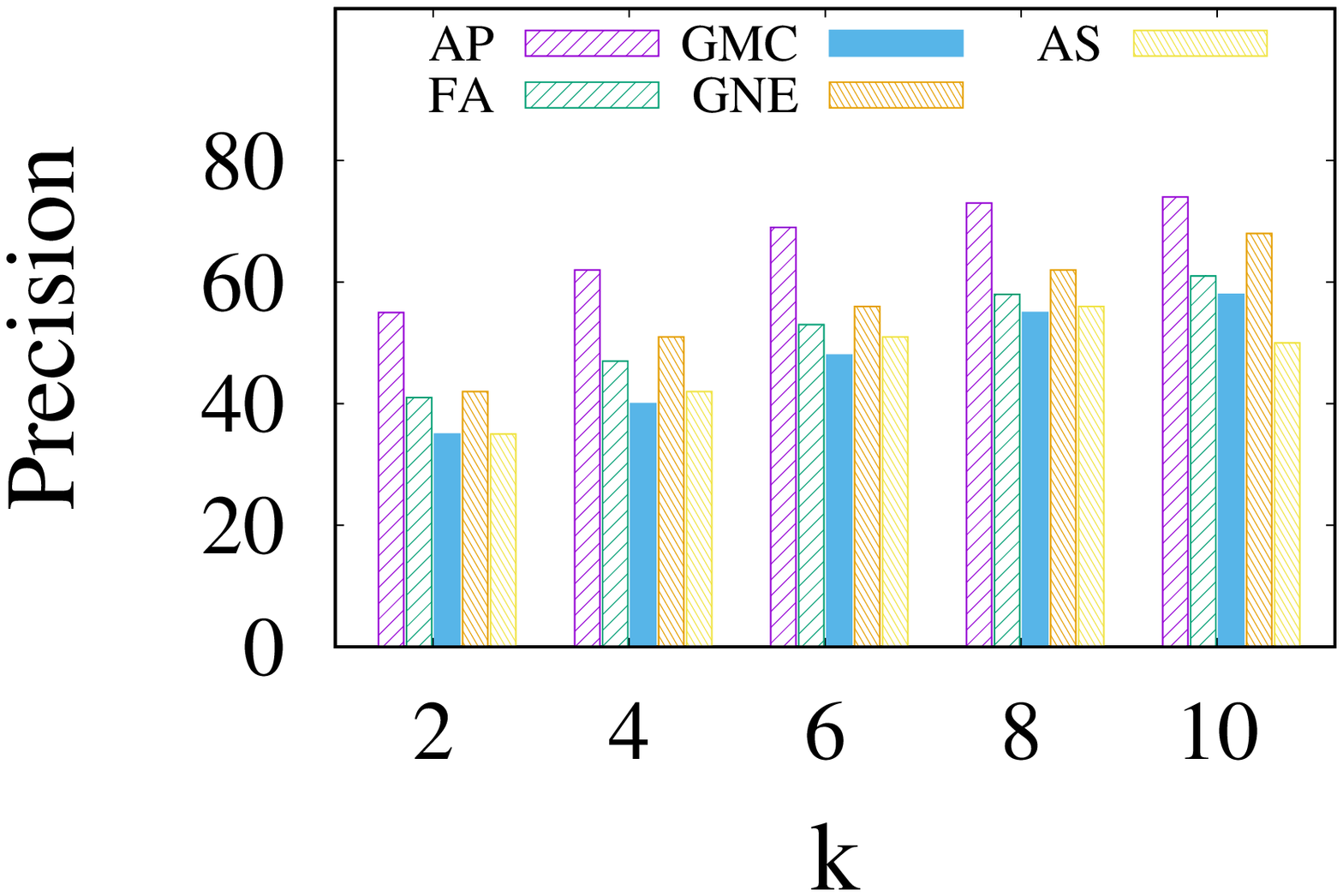} 
	}
	\caption{Precision}
	\label{fig:precLocNew}
	\vspace{-3mm}
\end{figure}

\jianxin{Figure~\ref{fig:precLocVaryAlphaOmega} shows the average precision of the models when $\alpha$ and $\omega$ vary. As the relative trends are similar on other datasets, we only show the effect of $\alpha$ and $\omega$ on \textit{Gowalla}. The precision of the methods typically increases with $\omega$ (e.g., preference to relevance). The \textit{FA} and \textit{AS} methods are influenced by the selection of top relevant location in the result set, which affects the precision when diversity has higher importance than relevance. The precision of \textit{AP}, \textit{GMC}, and \textit{GNE} remain almost constant when $\omega$ varies. 
The variation of $\alpha$ does not affect much in the precision of the approaches when $k$ is set as default. For example, in \textit{Gowalla}, the average precision of \textit{AP} is $71\%$ when $\alpha$ varies and $k=6$. In \textit{Yelp}, the average precision of \textit{AP} is reported as 68\% when $\alpha$ varies from $0.1$ to $0.9$.
}

\begin{figure}[htbp]
	\vspace{-4mm}
	\centering
	\subfigure[Gowalla]{\label{precLocAlphaOmegaGW}
		\includegraphics[scale=0.20]{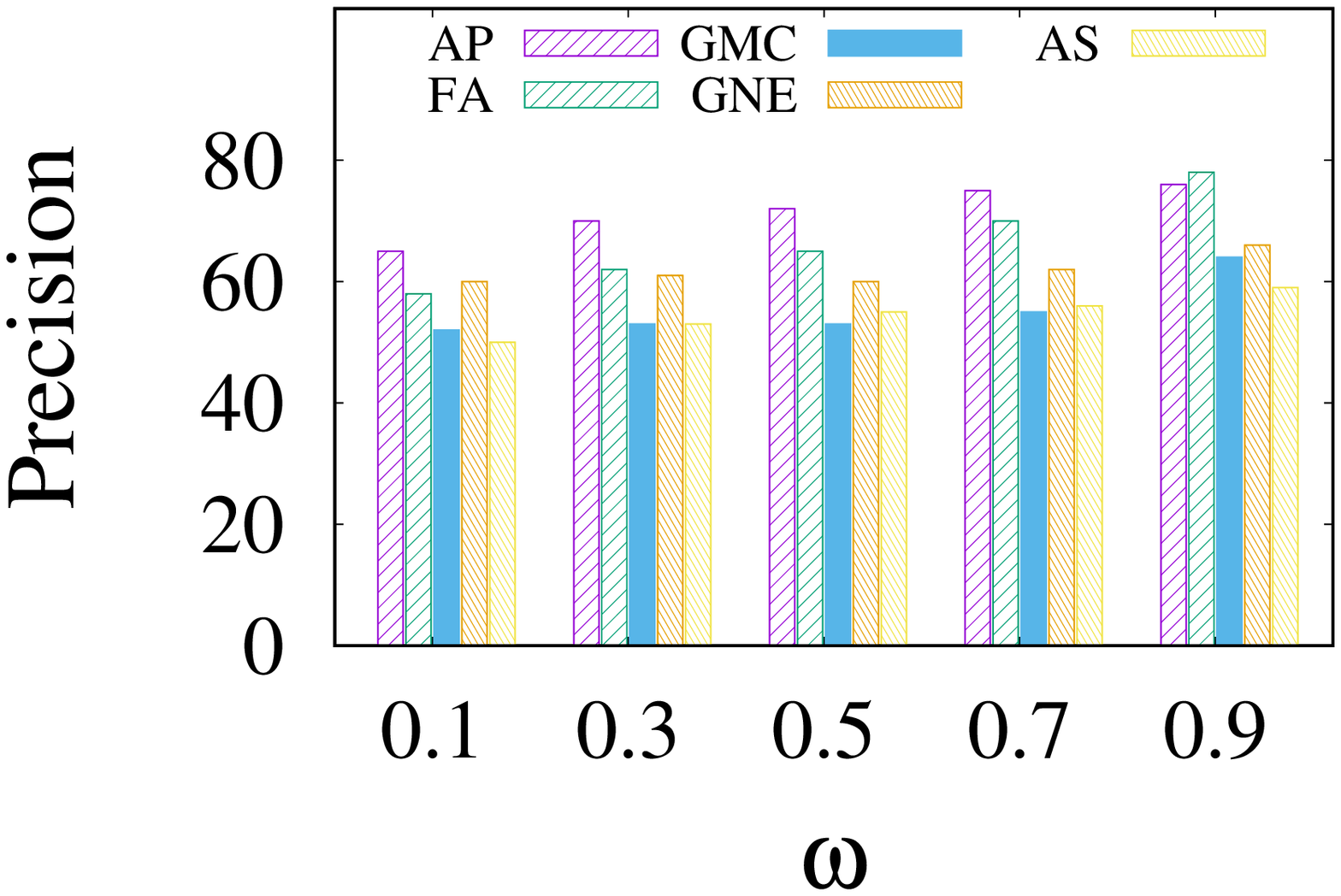}	\hspace{1mm}
		\includegraphics[scale=0.20]{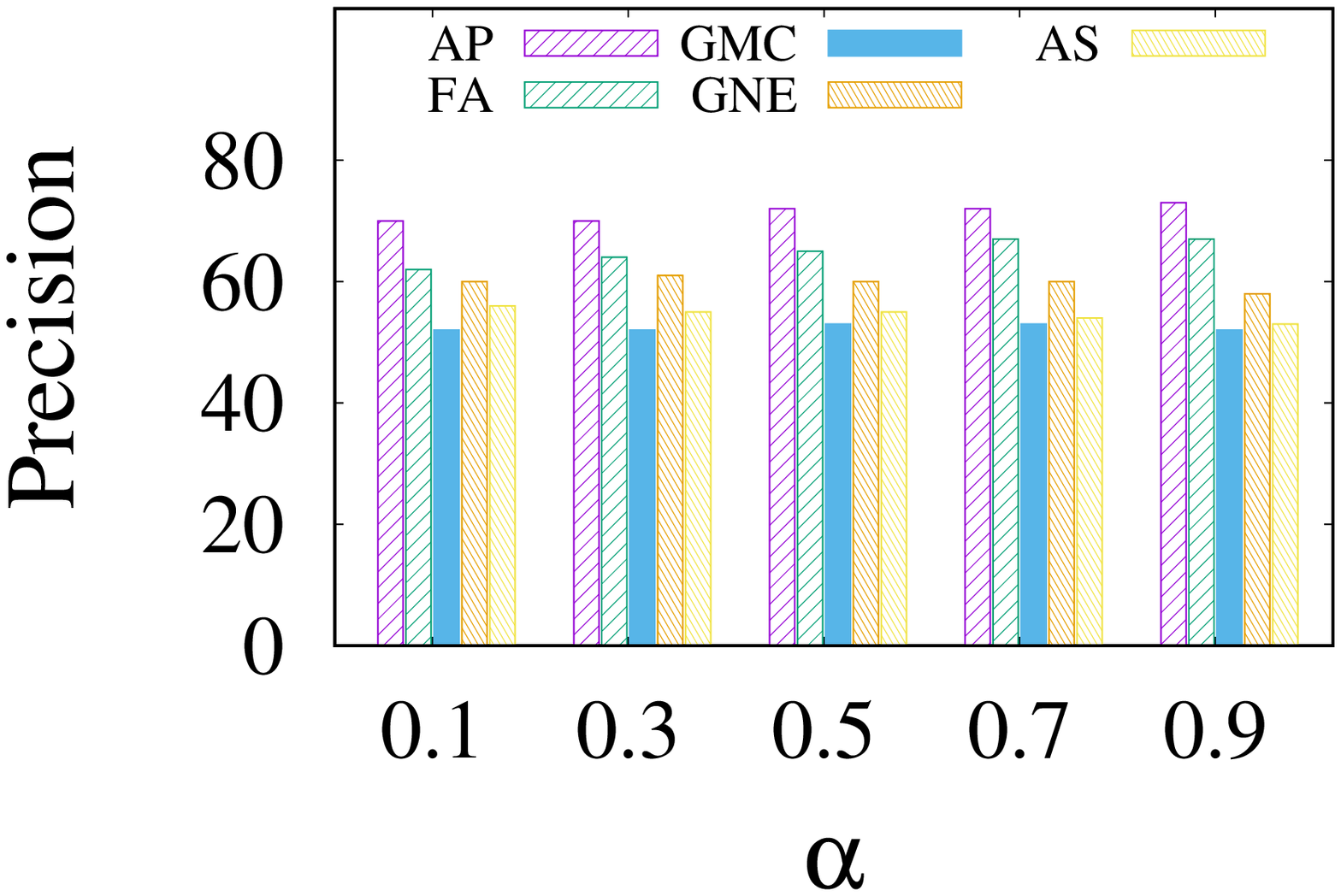}	
	}\vspace{-1mm}	
	\subfigure[Yelp]{\label{precLocAlphaOmegaYL}
		\includegraphics[scale=0.20]{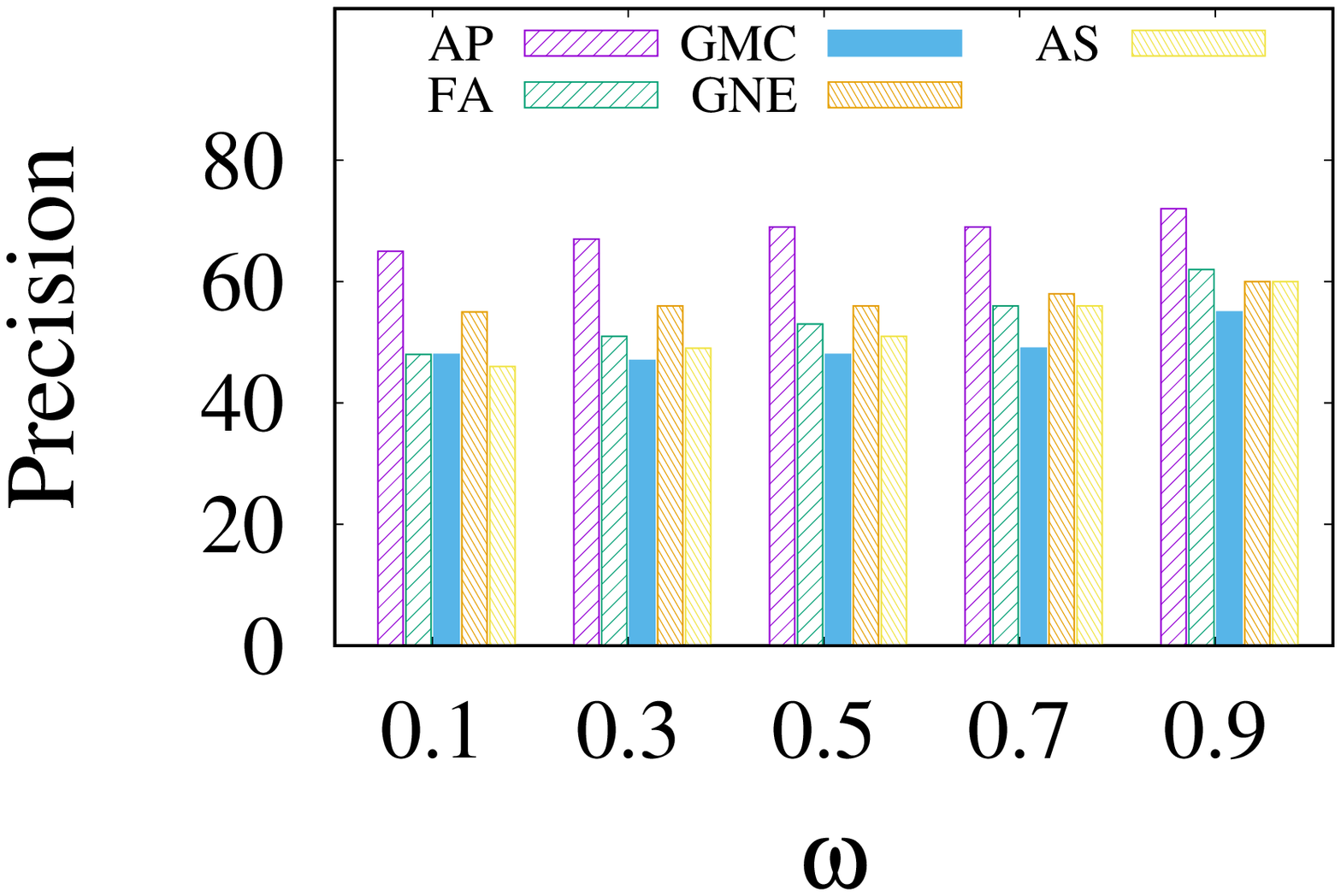}	\hspace{1mm}
		\includegraphics[scale=0.20]{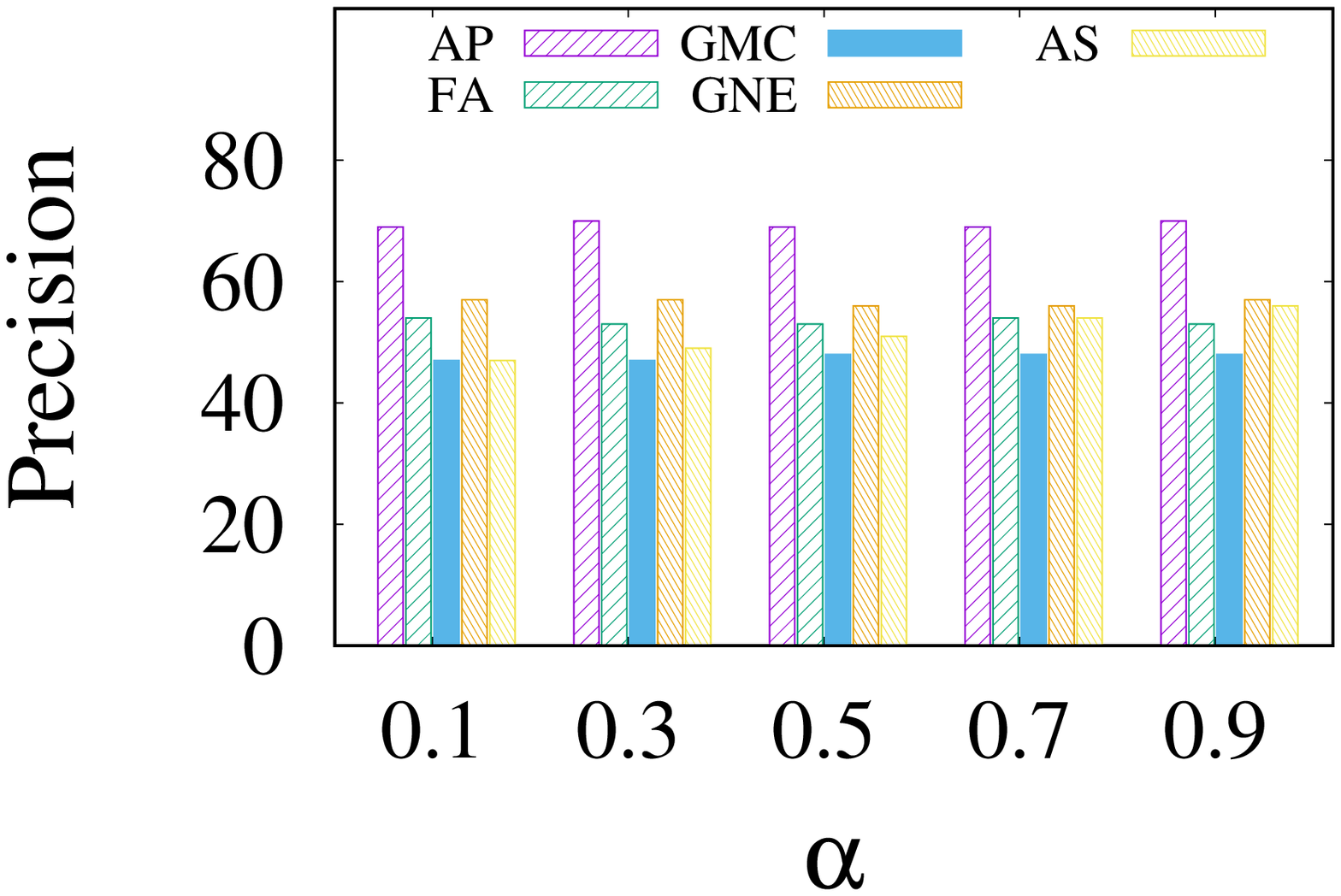}	
	}
	\caption{\jianxin{Precision when varying $\omega$, $\alpha$}}
	\label{fig:precLocVaryAlphaOmega}
\end{figure}

\subsubsection{\underline{Effectiveness}}
We compare the socio-spatial qualities of the selected locations using the \textit{MMD} metric. 
In \textit{Gowalla} (Figure \ref{MMD_GW_Compare}), the \textit{MMD} of \textit{AS} remains almost constant, while for the other approaches, the \textit{MMD} score decreases smoothly with the increase of $k$. This is because the \textit{AS} model considers a fixed user-defined threshold to maintain a minimum diversity. 
In \textit{Yelp}, all the approaches produce lower \textit{MMD} (Figure \ref{MMD_YL_Compare}). This means the majority of user's friends in \textit{Yelp} have closer check-ins to the selected locations. 

\begin{figure}[htbp]
	\vspace{-4mm}
	\centering
	\subfigure[Gowalla]{\label{MMD_GW_Compare}
		\includegraphics[scale=0.195]{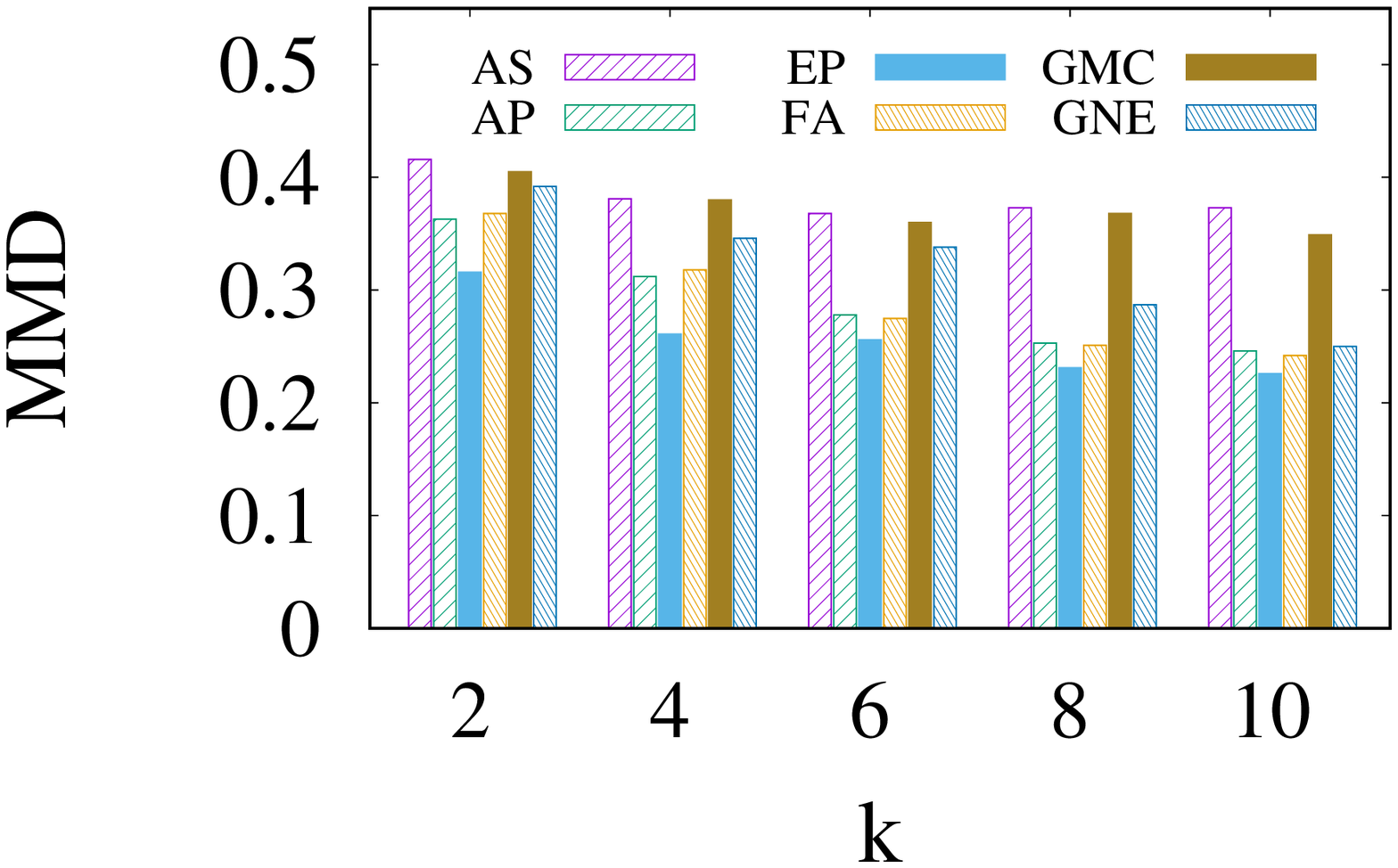}
	}
	\subfigure[Yelp]{\label{MMD_YL_Compare}
		\includegraphics[scale=0.195]{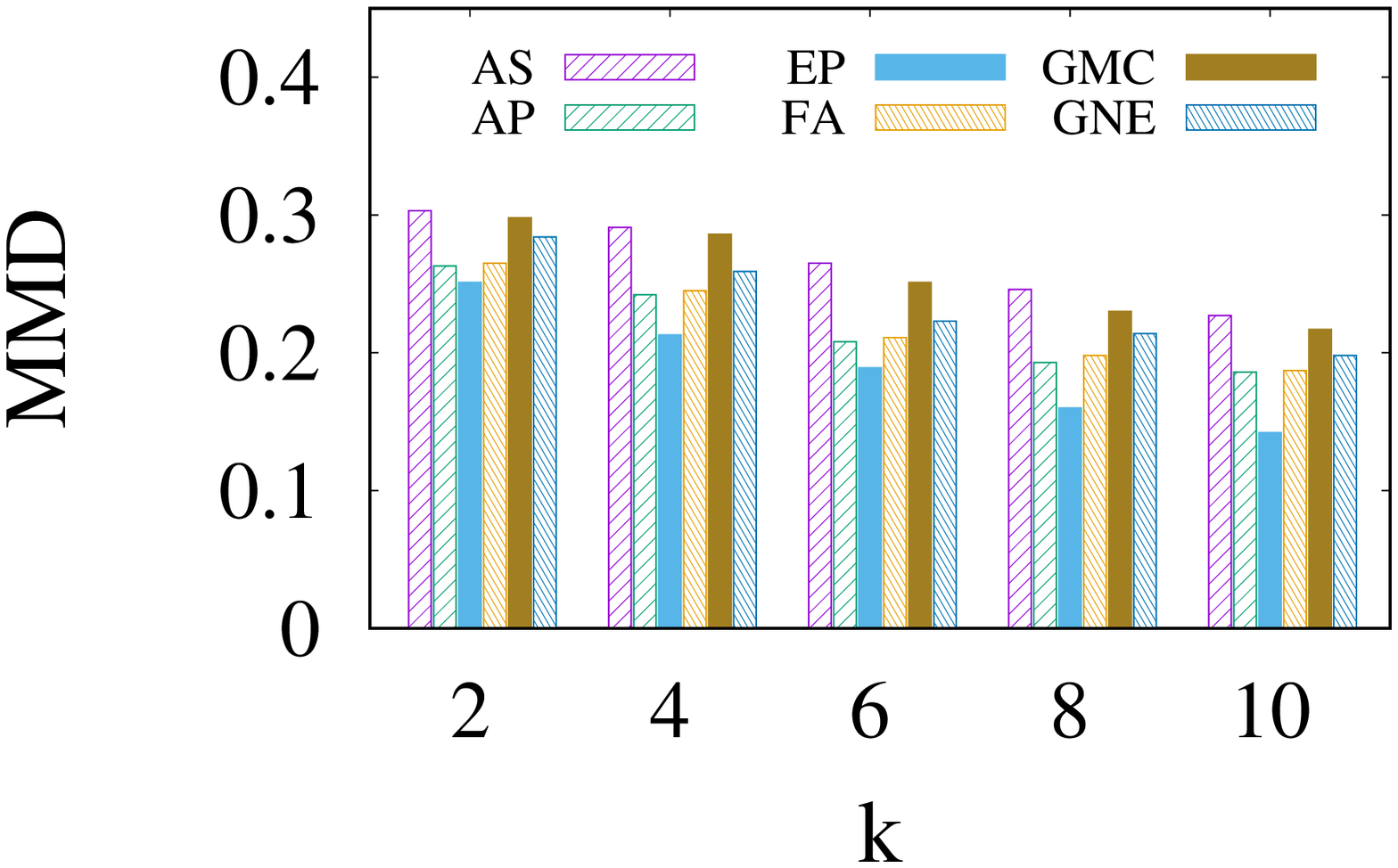}	
	}
	\vspace{-3mm}
	\caption{MMD Comparison}
	\label{fig:MMDComparisonNew}
	\vspace{-3mm}
\end{figure}

Figure \ref{fig:SocCoverageNew} compares the social coverage (\textit{SC}) of the algorithms. In both the datasets, the relative trends are similar. The top-$6$ $SSLS$ locations in \textit{EP} are co-located with 64\% and 74\% neighbors in \textit{Gowalla} and \textit{Yelp} datasets, respectively. 
The \textit{GMC} method has the lowest \textit{SC}, i.e., it reports only 30\% in \textit{Yelp}. 
Interestingly, we find that the social coverage of \textit{FA} is marginally higher than \textit{EP}. This is because, \textit{FA} includes the top socio-spatial relevant locations in the result set. Therefore, the selected set has exact check-ins by a large number of friends. 

\begin{figure}[htbp]
	\vspace{-3mm}
	\centering
	\subfigure[Gowalla]{\label{SC_Exact_GW_New}
		\includegraphics[scale=0.195]{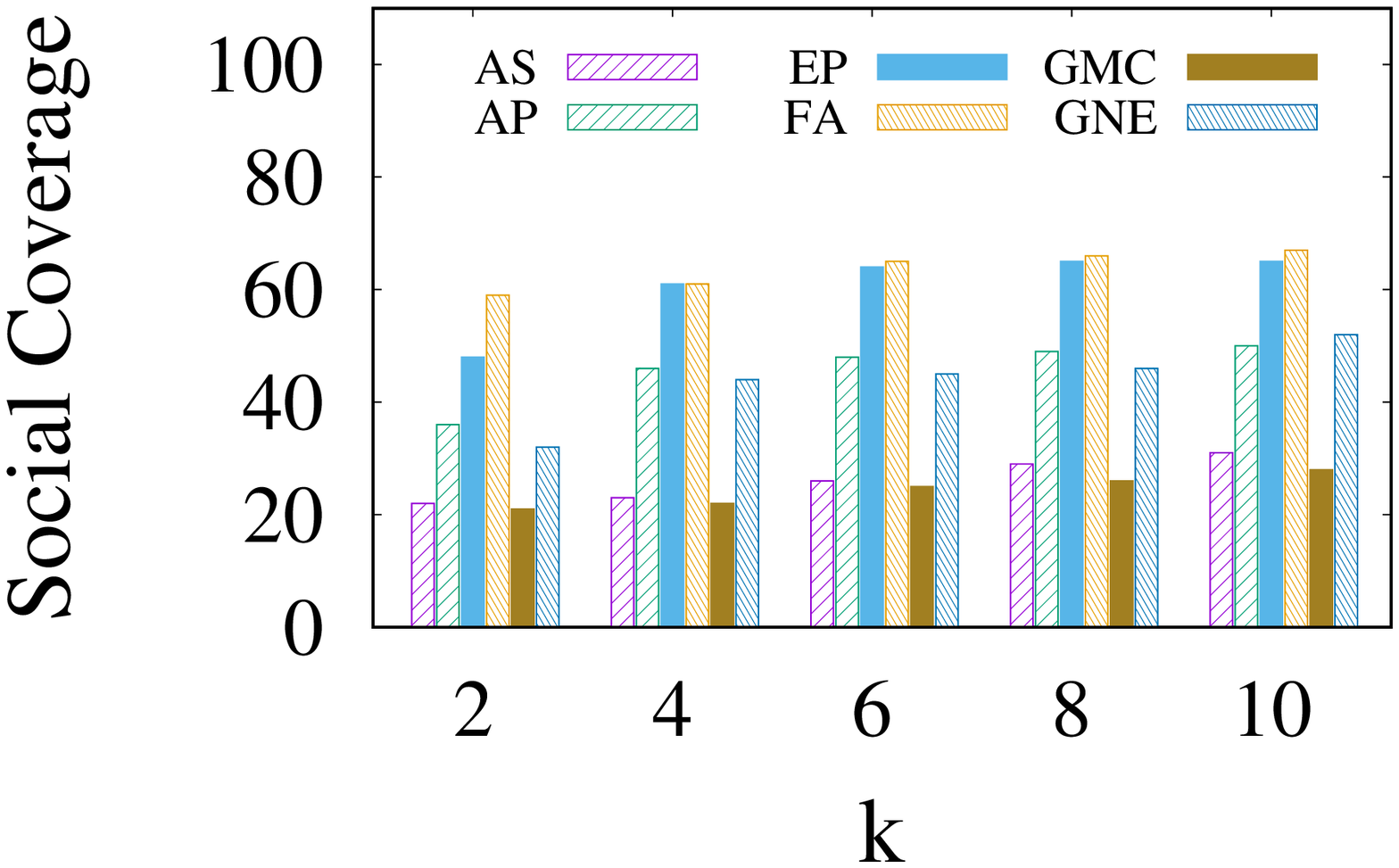}
	}	
	\subfigure[Yelp]{\label{SC_Exact_YL_New}
		\includegraphics[scale=0.195]{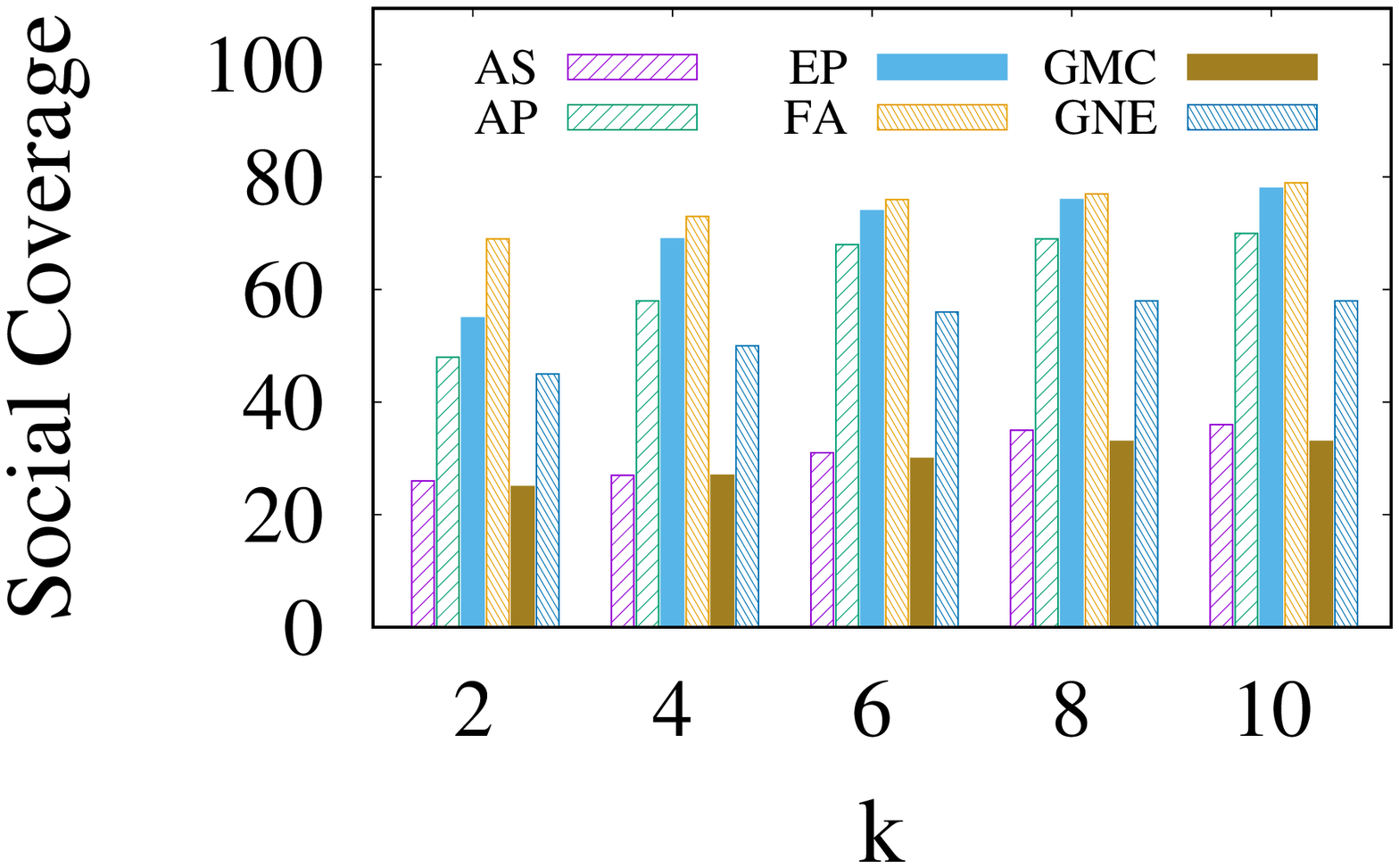}	
	}
	\vspace{-3mm}
	\caption{Social Coverage}
	\label{fig:SocCoverageNew}
	\vspace{-3mm}
\end{figure}

Figure~\ref{fig:SocEntropy} shows the average social entropy (\textit{SE}) of the approaches when answer set size $k$ varies. Similar trends are followed in both the datasets. The \textit{EP} approach has the highest average \textit{SE}, which means the selected locations by \textit{EP} have diverse participation of friends. Meanwhile, \textit{EP} also has higher social coverage (\textit{SC}) (Figure~\ref{fig:SocCoverageNew}). These two metrics \textit{SE} and \textit{SC} together establish that the selected locations in \textit{EP} not only cover a large number of friends, but represent diverse groups. Compared with \textit{GMC} and \textit{GNE}, \textit{AS} has higher social entropy. 

\begin{figure}[htbp]
	\vspace{-3mm}
	\centering
	\subfigure[Gowalla]{\label{SocEntropy_GW}
		\includegraphics[scale=0.195]{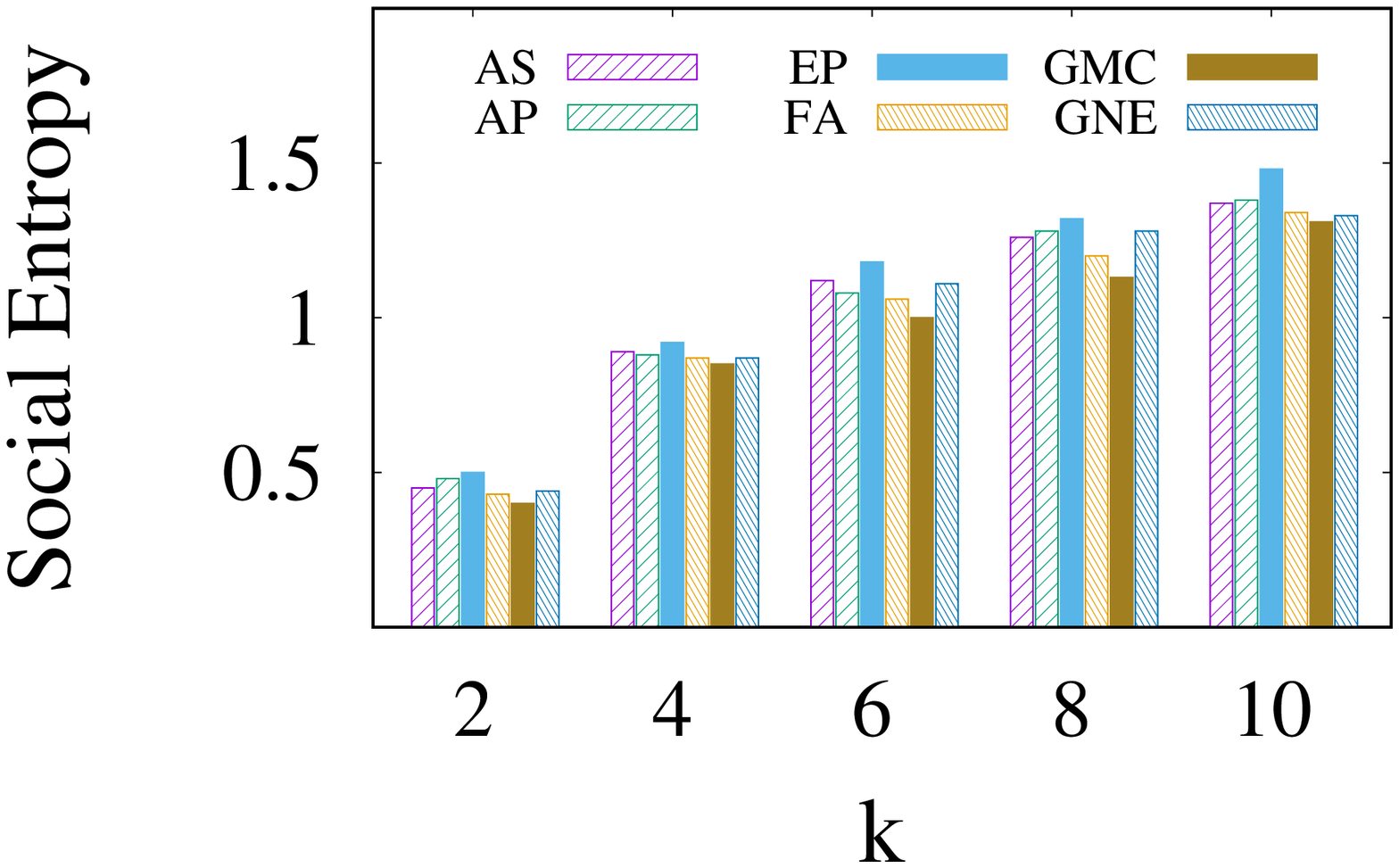}
	}	
	\subfigure[Yelp]{\label{SocEntropy_YL}
		\includegraphics[scale=0.195]{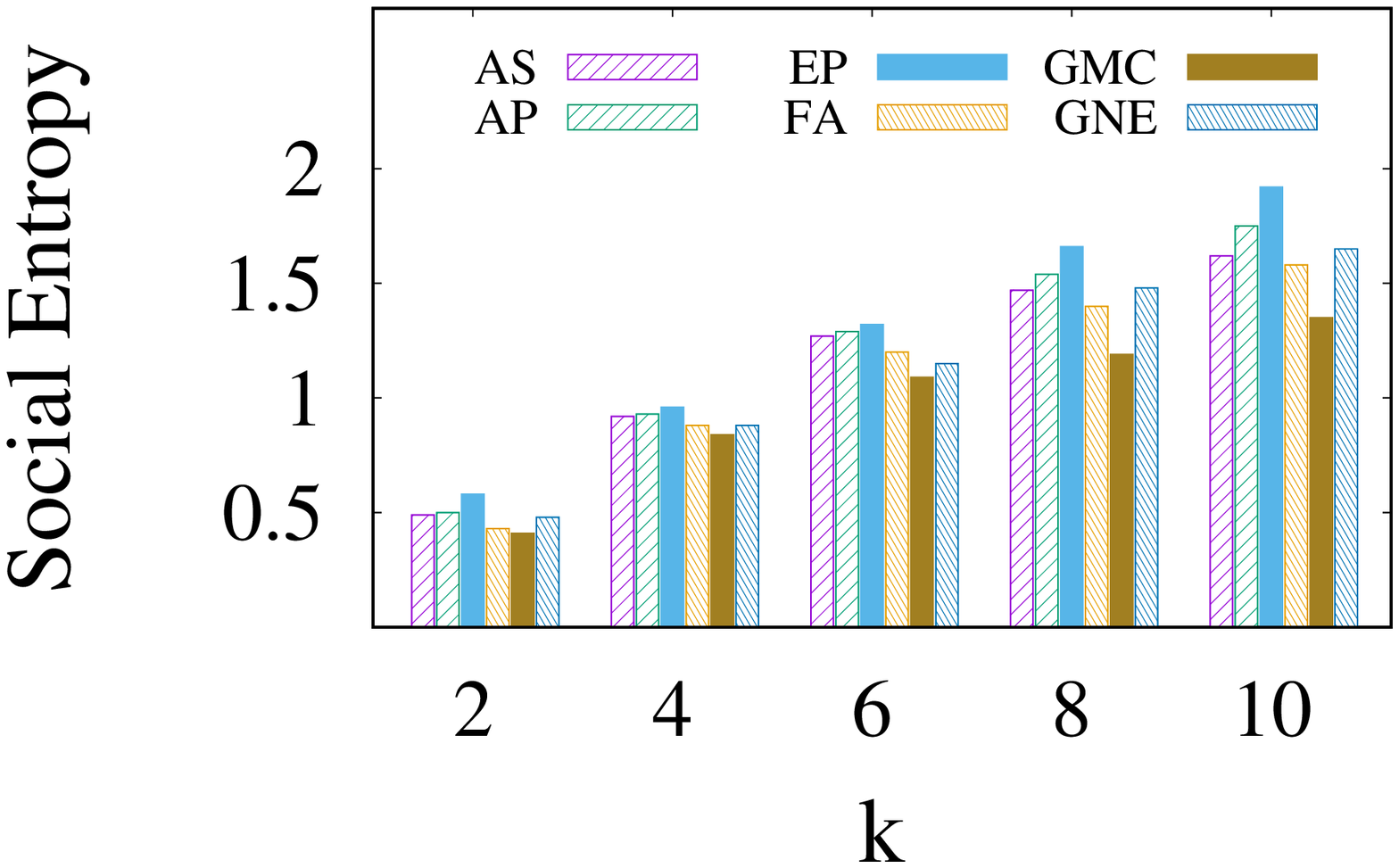}
	}
	\vspace{-3mm}
	\caption{Social Entropy}
	\label{fig:SocEntropy}
	\vspace{-3mm}
\end{figure}

\subsubsection{\underline{Memory Consumption}}
We observe that \textit{EP, FA, GMC, GNE, AS} has similar memory consumption, where the average memory usages are reported as 1195MB, 845MB, 2940MB, 1410MB on \textit{Gowalla}, \textit{Brightkite}, \textit{Flickr}, and \textit{Yelp}, respectively.
The \texttt{Exact} and \textit{AP} methods need to store the intermediate set information in a priority queue, which leads to higher memory cost. For example, in Brightkite, $E$ and $AP$ consume average 1150MB for the users in check-in group 100.

\subsection{A Case Study on Location Set Selection} 
\jianxin{In Figure \ref{fig:visualCompare}, we visualize the selection result of top-$5$ \textit{SSLS} using Adaptive SOS, \texttt{Exact}, and \textit{Approximate} methods considering $\alpha=0.5$, $\omega=0.5$. First, we choose a query user (userid `10') from  Gowalla~\cite{snapnets} dataset, and select the region (38.85, -94.85) to (39.11, -94.58) on map where the user has majority of its check-ins. Further, we obtain the check-in information of the neighbors of the user `10' having at least ten check-in in the mentioned area. There are nine such neighbors available in the selected region. Locations of the user `10' and its neighbors are marked in yellow and blue, respectively (best visible in color with zooming). 
	The user `10' has frequent check-ins concentrated at the red bordered region shown in Figure \ref{fig:SOS_Adaptive}. The five locations selected by the Adaptive SOS (\textit{AS}) model are quite distant (shown in red icons in Figure \ref{fig:SOS_Adaptive}). 
	However, \textit{AS} has ignored one important location (39.10, -94.59) (marked as red at NE corner in Figure \ref{fig:Exact}) which is included in top-$5$ SSLS result by our proposed \texttt{Exact} and \textit{Approximate} approaches. This location (39.10, -94.59) is spatially relevant to the user `10', as six neighbors (out of nine) have multiple check-ins (total 62) in 7 nearby places within $1.5KM$. In such a configuration, our \textit{Approximate} approach has four common selection as \texttt{Exact}. Meanwhile, we provide the snippet of the socio-spatial information of the user `10' (of Gowalla dataset) and its neighbors (who had at least check-ins at the region (38.85, -94.85) to (39.11, -94.58)) at \url{https://github.com/nurjamia/SSLS/blob/master/CaseStydyUserid10_GW.txt}.}
\begin{figure*}[htbp]
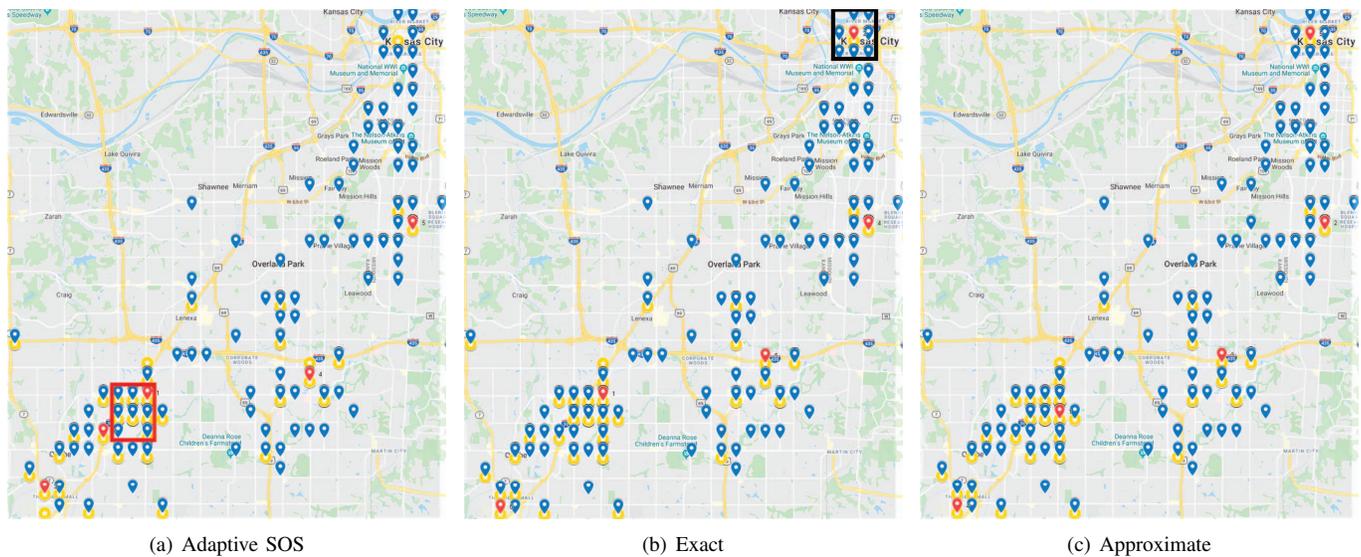

	\centering
	\subfigure[Adaptive SOS]{\label{fig:SOS_Adaptive}
		\includegraphics[scale=0.19]{SOS_Adaptive6_New.eps}
	}
	\subfigure[Exact]{\label{fig:Exact}
		\includegraphics[scale=0.19]{Exact6_New.eps}
	}
	\subfigure[Approximate]{\label{fig:Approximate}
		\includegraphics[scale=0.19]{Approximate6.eps}	
	}
	
	\caption{A case study using Gowalla dataset}
	\label{fig:visualCompare}
	\vspace{-5mm}
\end{figure*}

\section{Conclusion}
\label{sec:Conclusion}

\jianxin{In this paper, we propose a novel problem of  \textit{identifying top-k \textbf{S}ocio-\textbf{S}patial co-engaged \textbf{L}ocation \textbf{S}election}. It selects $k$ locations for a user from a large number of candidate locations based on the dominance of the combined socio-spatial diversity and relevance scores. We develop two exact and two approximate solutions to solve this NP-hard problem. Finally, the quality of our proposed approaches has been validated by comparing with the state-of-the-art object selection models. The extensive experimental studies on four real datasets with various socio-spatial characteristics have verified the performance of our proposed approaches.}

\bibliographystyle{IEEEtran}

\bibliography{SSLS}

\end{document}